\documentclass[11pt,letterpaper]{article}
\def\Anonymity{0}%

\usepackage{geometry}
\def\PrintMode{0}%

\usepackage[T1]{fontenc}
\usepackage{graphicx}
\usepackage{algorithm,algpseudocode}
\usepackage{xspace}
\usepackage[utf8]{inputenc}
\usepackage[noadjust]{cite}
\usepackage{amsmath, amssymb, amsfonts, amsthm,mathtools}
\usepackage{enumerate}
\usepackage{xcolor}
\usepackage[pagebackref=true,colorlinks,linkcolor=magenta,citecolor=blue,bookmarks,bookmarksopen,bookmarksnumbered]{hyperref}%
\usepackage{bm}
\usepackage{paralist}
\usepackage{thmtools}
\usepackage{rotating}
\usepackage{mdframed}
\usepackage{multirow}
\usepackage{verbatim}
\usepackage{mathrsfs}
\usepackage{indentfirst}
\usepackage{mleftright}
\usepackage[nottoc]{tocbibind}%
\usepackage{complexity}
\usepackage{tablefootnote}
\usepackage{relsize}
\usepackage{exscale}

\ifnum\PrintMode=1

\usepackage{savetrees}
\else
\fi

\usepackage[capitalize]{cleveref}%

\newtheorem{theorem}{Theorem}[section]
\newtheorem{lemma}[theorem]{Lemma}
\newtheorem{corollary}[theorem]{Corollary}
\newtheorem{claim}[theorem]{Claim}

\newtheorem{fact}[theorem]{Fact}

\theoremstyle{definition}
\newtheorem{definition}[theorem]{Definition}

\theoremstyle{remark}

\algtext*{EndWhile}%
\algtext*{EndIf}%
\algtext*{EndFor}%
\algtext*{EndFunction}%

\algnewcommand{\IfThen}[2]%
{\State \algorithmicif\ #1\ \algorithmicthen\ #2}

\makeatletter
\def\moverlay{\mathpalette\mov@rlay}
\def\mov@rlay#1#2{\leavevmode\vtop{%
		\baselineskip\z@skip \lineskiplimit-\maxdimen
		\ialign{\hfil$\m@th#1##$\hfil\cr#2\crcr}}}
\newcommand{\charfusion}[3][\mathord]{
	#1{\ifx#1\mathop\vphantom{#2}\fi
		\mathpalette\mov@rlay{#2\cr#3}
	}
	\ifx#1\mathop\expandafter\displaylimits\fi}
\makeatother

\DeclareMathOperator{\GF}{GF}

\renewcommand{\poly}{\mathrm{poly}}

\renewcommand{\polylog}{\mathrm{polylog}}

\renewcommand{\emptyset}{\varnothing}
\newcommand{\eps}{\varepsilon}

\newcommand{\Base}{\mathsf{Base}}
\newcommand{\DSR}{\mathsf{DSR}}

\newcommand{\WT}{\widetilde}
\newcommand{\HSG}{\mathsf{H}}
\def\BF{\mathsf{BF}}
\newcommand{\tl}{t^{(\ell)}}
\newcommand{\nl}{n^{(\ell)}}
\newcommand{\Hl}{\HSG^{(\ell)}}
\newcommand{\TBF}{T'}
\newcommand{\dBF}{d'}

\newcommand{\descTM}[1]{\left|#1\right|}
\newcommand{\TurM}{\mathsf{TM}}

\newcommand{\posN}{\N_{\ge 1}}

\newlang{\MCSP}{MCSP}
\newlang{\MFSP}{MFSP}
\newlang{\MKtP}{MKtP}
\newlang{\MKTP}{MKTP}
\newlang{\itrMCSP}{itrMCSP}
\newlang{\itrMKTP}{itrMKTP}
\newlang{\itrMINKT}{itrMINKT}
\newlang{\MINKT}{MINKT}
\newlang{\MINK}{MINK}
\newlang{\MINcKT}{MINcKT}
\newlang{\CMD}{CMD}
\newlang{\DCMD}{DCMD}
\newlang{\CGL}{CGL}
\newlang{\PARITY}{PARITY}
\renewlang{\Gap}{Gap}
\newlang{\Empty}{\textsc{Empty}}
\newlang{\Avoid}{\textsc{Avoid}}
\newlang{\Sparsification}{\textsc{Sparsification}}
\newlang{\HamEst}{\mathsf{HammingEst}}
\newlang{\HamHit}{\mathsf{HammingHit}}
\newlang{\CktEval}{\textsc{Circuit-Eval}}
\newlang{\Hard}{\textsc{Hard}}
\newlang{\cHard}{\textsc{cHard}}
\newlang{\CAPP}{CAPP}
\newlang{\GapUNSAT}{GapUNSAT}
\newlang{\OV}{OV}
\newlang{\PRIMES}{PRIMES}
\renewlang{\PCP}{PCP}
\newlang{\PCPP}{PCPP}
\newclass{\FMA}{FMA}
\newclass{\Avg}{Avg}
\newclass{\ZPEXP}{ZPEXP}
\newclass{\DLOGTIME}{DLOGTIME}
\newclass{\ALOGTIME}{ALOGTIME}
\newclass{\ATIME}{ATIME}%
\newclass{\SZKA}{SZKA}
\newclass{\Laconic}{Laconic\text{-}}
\newclass{\APEPP}{APEPP}
\newclass{\SAPEPP}{SAPEPP}
\newclass{\TFSigma}{TF\Sigma}
\newclass{\NTIMEGUESS}{NTIMEGUESS}

\newlang{\Formula}{Formula}
\newlang{\THR}{THR}
\newlang{\MAJ}{MAJ}
\newlang{\DOR}{DOR}
\newlang{\ETHR}{ETHR}
\newlang{\Midbit}{Midbit}
\newlang{\LCS}{LCS}
\newlang{\TAUT}{TAUT}

\newcommand{\NAND}{\mathsf{NAND}}

\newcommand{\ie}{\textit{i}.\textit{e}.\@\xspace}
\newcommand{\eg}{\textit{e}.\textit{g}.\@\xspace}

\newcommand{\fx}{\mathbf{x}}
\newcommand{\fy}{\mathbf{y}}

\newcommand{\algoCT}{\mathbb{A}^{\sf ct}}
\newcommand{\set}[1]{\ensuremath{\left\{#1\right\}}}
\newcommand{\ra}{\rightarrow}

\newcommand{\bitset}{\ensuremath{\{0,1\}}}
\newcommand{\bs}[1]{\ensuremath{\bitset^{#1}}}
\newcommand{\bfunc}[3]{\ensuremath{#1\colon\bs{#2}\rightarrow\bs{#3}}}

\renewcommand{\func}[3]{\ensuremath{#1\colon#2\ra#3}}

\newcommand{\calC}{\mathcal{C}}

\newcommand{\N}{\mathbb{N}}

\newcommand{\F}{\mathbb{F}}

\newcommand{\Eval}{\mathsf{Eval}}%

\newcommand{\NW}{\mathsf{NW}}%

\newcommand{\TM}{\mathsf{TM}}

\newcommand{\al}{\ensuremath{\alpha}}
\newcommand{\de}{\ensuremath{\delta}}

\renewcommand{\tt}{\mathtt{tt}}
\newcommand{\mathdot}{\ensuremath{\;\text{.}}}
\newcommand{\mathcomma}{\ensuremath{\;\text{,}}}

\usepackage[draft,inline,index]{fixme}  %
\fxsetup{theme=color,mode=multiuser}

\definecolor{color1}{RGB}{46,134,193}
\FXRegisterAuthor{hanlin}{ahanlin}{\colorbox{color1}{\color{white}Hanlin}}

\definecolor{color2}{RGB}{20,60,100}
\FXRegisterAuthor{zhikun}{azhikun}{\colorbox{color2}{\color{white}Zhikun}}
\FXRegisterAuthor{clj}{aclj}{\colorbox{color2}{\color{white}Lijie}}

\definecolor{color3}{RGB}{0,0,255}
\FXRegisterAuthor{igor}{aigor}{\colorbox{color3}{\color{white}Igor}}
\definecolor{color4}{RGB}{255,0,0}
\FXRegisterAuthor{zhenjian}{azhenjian}{\colorbox{color4}{\color{white}Zhenjian}}

\definecolor{color5}{RGB}{0,255,0}
\FXRegisterAuthor{rahul}{arahul}{\colorbox{color5}{\color{white}Rahul}}

\usepackage{aliascnt} %

\newcommand{\eqdef}{\triangleq}

\begin{document}

\title{Polynomial-Time Pseudodeterministic Construction of Primes\vspace{0.3cm}}
\ifnum\Anonymity=0
\author{
	Lijie Chen\\ \small{UC Berkeley} \\ \small{\texttt{\href{mailto:lijiechen@berkeley.edu}{lijiechen@berkeley.edu}}}
 \and
	Zhenjian Lu\\ \small{University of Oxford} \\ \small{\texttt{\href{mailto:zhenjian.lu@cs.ox.ac.uk}{zhenjian.lu@cs.ox.ac.uk}}}
	\and
	Igor C. Oliveira\\ \small{University of Warwick} \\ \small{\texttt{\href{mailto:igor.oliveira@warwick.ac.uk}{igor.oliveira@warwick.ac.uk}}\vspace{0.4cm}}
 \and
 Hanlin Ren\\ \small{University of Oxford} \\ \small{\texttt{\href{mailto:hanlin.ren@cs.ox.ac.uk}{hanlin.ren@cs.ox.ac.uk}}}
 \and
	Rahul Santhanam\\ \small{University of Oxford} \\ \small{\texttt{\href{mailto:rahul.santhanam@cs.ox.ac.uk}{rahul.santhanam@cs.ox.ac.uk}}\vspace{0.3cm}}
 }
\fi

\maketitle

\vspace{-0.75cm}

\begin{abstract}
A randomized algorithm for a search problem is \emph{pseudodeterministic} if it produces a fixed canonical solution to the search problem with high probability. In their seminal work on the topic, Gat and Goldwasser \cite{DBLP:journals/eccc/GatG11} posed as their main open problem whether prime numbers can be pseudodeterministically constructed in polynomial time. 
 
We provide a positive solution to this question in the infinitely-often regime. In more detail, we give an \emph{unconditional} polynomial-time randomized algorithm $B$ such that, for infinitely many values of $n$, $B(1^n)$ outputs a canonical $n$-bit prime $p_n$ with high probability. More generally, we prove that for every dense property $Q$ of strings that can be decided in polynomial time, there is an infinitely-often pseudodeterministic polynomial-time construction of strings satisfying $Q$. This improves upon a subexponential-time construction of Oliveira and Santhanam \cite{DBLP:conf/stoc/OliveiraS17}.  

Our construction uses several new ideas, including a novel bootstrapping technique for pseudodeterministic constructions, and a quantitative optimization of the uniform hardness-randomness framework of Chen and Tell \cite{ChenT21b}, using a variant of the Shaltiel--Umans generator \cite{ShaltielU05}.

\end{abstract}

\newpage

\setcounter{tocdepth}{3}
\tableofcontents

\newpage

\section{Introduction}\label{sec:intro}

How hard is it to construct an $n$-bit prime\footnote{Recall that a positive integer $q$ is an $n$-bit prime if $q$ is a prime number and $2^{n-1} \leq q \leq 2^n - 1$.}? This is a fundamental problem in number theory and in complexity theory. Under reasonable assumptions, the problem is solvable in deterministic polynomial time. In more detail, Cram{\'e}r's conjecture \cite{cramer1936order} in number theory asserts that the largest prime gap in any consecutive sequence of $n$-bit numbers is $O(n^2)$. Assuming this conjecture, we can solve the prime construction problem efficiently by testing the first $O(n^2)$ integers greater than $2^{n-1}$ for primality and outputting the first one, where the primality tests are done efficiently using the algorithm of Agrawal, Kayal and Saxena \cite{Agrawal02primesis}. An independent source of evidence for the efficiency of prime construction is the complexity-theoretic conjecture that $\DTIME(2^{O(n)})$ requires Boolean circuits of exponential size on almost all input lengths. Under this conjecture, we can use the Impagliazzo--Wigderson pseudorandom generator \cite{ImpagliazzoW97} to {\it derandomize} the simple randomized algorithm that outputs a random $n$-bit number, using the facts that primality testing is in polynomial time and that an $\Omega(1/n)$ fraction of $n$-bit numbers are prime.

However, we seem very far from either settling Cram{\'e}r's conjecture or proving strong complexity lower bounds. The best upper bound we can prove on the gap between consecutive $n$-bit primes is $2^{(0.525 + o(1))n}$ \cite{MR1851081}, and no super-linear circuit lower bounds are known for $\DTIME(2^{O(n)})$ \cite{LiY22}. Indeed, the best unconditional result we have so far is that deterministic prime construction can be done in time $2^{(0.5 + o(1))n}$ \cite{DBLP:journals/jal/LagariasO87}, which is very far from the polynomial-time bound we seek. The Polymath 4 project (see \cite{Tao-Croot-Helfgott12}) sought to improve this upper bound using number-theoretic techniques but did not achieve an unconditional improvement.

In contrast to the situation with deterministic prime construction, it is easy to generate an $n$-bit prime {\it randomly}, as mentioned above: simply generate a random $n$-bit number, test it for primality in polynomial time, and output it if it is a prime. This algorithm has success probability $\Omega(1/n)$ by the Prime Number Theorem, and the success probability can be amplified to be exponentially close to $1$ by repeating the process $\poly(n)$ times independently, and outputting the first of these $\poly(n)$ numbers that is verified to be prime, assuming that there is at least one.

Gat and Goldwasser \cite{DBLP:journals/eccc/GatG11} asked whether it is possible to generate primes efficiently by a randomized process, such that the output is essentially {\it independent} of the randomness of the algorithm. In other words, is there a polynomial-time randomized algorithm, which on input $1^n$, constructs a {\it canonical} prime of length $n$ with high probability?
They call such an algorithm a {\it pseudodeterministic} algorithm, since the output of the algorithm is (almost) deterministic even though the algorithm might use random bits in its operation. Note that the randomized algorithm for prime generation we described in the previous paragraph is very far from being pseudodeterministic, as different runs of the algorithm are unlikely to produce the same prime. It is easy to see that a pseudodeterministic construction serves as an intermediate notion between a randomized construction (which is trivial for primes) and a deterministic construction  (where little progress has been made so far).

\cite{DBLP:journals/eccc/GatG11} initiate a general theory of pseudodeterminism for search problems, motivated by applications in cryptography and distributed computing. Since then, there have been a number of papers on pseudodeterminism, in various contexts, such as query complexity \cite{DBLP:conf/innovations/GoldreichGR13, DBLP:conf/coco/GoldwasserIPS21, CDM23}, streaming algorithms \cite{DBLP:conf/innovations/GoldwasserGMW20, BKKS23}, parallel computation \cite{DBLP:conf/icalp/GoldwasserG17, DBLP:conf/approx/GhoshG21}, learning algorithms \cite{DBLP:conf/approx/OliveiraS18}, Kolmogorov complexity \cite{DBLP:conf/icalp/Oliveira19, DBLP:conf/stoc/LuOS21}, space-bounded computation \cite{DBLP:conf/soda/GrossmanL19}, proof systems \cite{DBLP:conf/innovations/GoldwasserGH18, DBLP:journals/eccc/GoemansGH19}, number theory and computational algebra \cite{DBLP:journals/eccc/Grossman15, DBLP:conf/stoc/OliveiraS17}, approximation algorithms \cite{DBLP:conf/mfcs/DixonPV18}, and many other settings (see, \eg,~\cite{DBLP:conf/scn/BergerB18, DBLP:journals/eccc/Goldreich19, paperITCS21, DBLP:conf/stoc/0002PWV22, DBLP:journals/eccc/Woude00RV22,   cryptoeprint:2023/343}). 

Despite all this progress, the main problem about pseudodeterminism posed in \cite{DBLP:journals/eccc/GatG11} has remained open: Is there a pseudodeterministic polynomial-time algorithm for prime construction? They describe this problem as ``the most intriguing'' and ``perhaps the most compelling challenge for finding a unique output''. 
 
Unlike in the case of deterministic construction, number-theoretic techniques have so far not proven useful for the pseudodeterministic construction problem for primes. Using complexity-theoretic techniques, Oliveira and Santhanam \cite{DBLP:conf/stoc/OliveiraS17} (see also \cite{DBLP:conf/stoc/LuOS21}) showed that for any $\varepsilon > 0$, there is an algorithm that runs in time $2^{n^{\varepsilon}}$ and succeeds on infinitely many input lengths.

\subsection{Our Results}
In this paper, we design a significantly faster algorithm and provide an affirmative answer to the question posed by Gat and Goldwasser in the infinitely-often regime. Our main result can be stated in full generality as follows.

\begin{restatable}[Infinitely-Often Polynomial-Time Pseudodeterministic Constructions]{theorem}{ThmMain}\label{theo:intro_main}
	Let $Q \subseteq \bs{*}$ be a language with the following properties: 
	\begin{itemize}
		\item[{\bf (Density.)}] there is a constant $\rho \ge 1$ such that for every $n \in \posN$, $Q_n \eqdef Q \cap \bs{n}$ satisfies $|Q_n| \ge n^{-\rho}\cdot 2^n$; and
		\item[{\bf (Easiness.)}] there is a deterministic polynomial-time algorithm $A_{Q}$ that decides whether an input $x \in \bs{*}$ belongs to $Q$.
	\end{itemize}	 
	Then there exist a probabilistic polynomial-time algorithm $B$ and a sequence $\{x_n\}_{n \in \posN}$ of $n$-bit strings in $Q$ such that the following conditions hold:
	\begin{enumerate}
		\item On every input length $n \in \posN$, $\Pr_B[B(1^n) \notin \{x_n, \bot\}] \leq 2^{-n}$.
		\item On infinitely many input lengths $n \in \posN$, $\Pr_B[B(1^n) = x_n] \geq 1 - 2^{-n}$.
	\end{enumerate}
\end{restatable}

Interestingly, our construction is ``non-black-box'', in the sense that changing the \emph{code} of the algorithm $A_Q$ deciding property $Q$ affects the canonical output of the corresponding algorithm $B$. We will revisit this point when we discuss our techniques (see the remark at the end of \autoref{sec:improved_CT_exposition}).

Letting $Q$ be the set of prime numbers and noticing that $Q$ is both dense (by the Prime Number Theorem) and easy (by the AKS primality test \cite{Agrawal02primesis}), we immediately obtain the following corollary of \autoref{theo:intro_main}.

\begin{corollary}[Infinitely-Often Polynomial-Time Pseudodeterministic Construction of Primes]\label{cor:intro_primes}
There is a randomized polynomial-time algorithm $B$ such that, for infinitely many values of $n$, $B(1^n)$ outputs a canonical $n$-bit prime $p_n$ with high probability.
\end{corollary}

\autoref{cor:intro_primes} improves upon the subexponential-time infinitely-often pseudodeterministic construction of primes from \cite{DBLP:conf/stoc/OliveiraS17} mentioned above. Note that the result for prime construction is a corollary of a far more general result about properties that are dense and easy. This is evidence of the surprising power of complexity theory when applied to a problem which seems to be about number theory (but where number-theoretic techniques have not so far been effective). The famous efficient primality testing algorithm of \cite{Agrawal02primesis} similarly applied complexity-theoretic derandomization ideas to solve a longstanding open problem in computational number theory, though their argument does require more information about primes. %

For a string $w\in\{0,1\}^*$ and $t\colon\mathbb{N}\to\mathbb{N}$, we let $\mathsf{rK}^t(w)$ denote the length of the smallest randomized program that runs for at most $t(|w|)$ steps and outputs $w$ with probability at least $2/3$. (We refer to \cite{DBLP:journals/eatcs/LuO22} for a formal definition and for an introduction to probabilistic notions of time-bounded Kolmogorov complexity.)  By encoding the (constant-size) randomized polynomial-time algorithm $B$ and each good input length $n$ using $O(1) + \log n$ bits in total, the following result holds. 

\begin{corollary}[Infinitely Many Primes with Efficient 
 Succinct Descriptions]\label{cor:intro_rK} There is a constant $c \geq 1$ such that, for $t(n) = n ^c$, the following holds. For every $m \geq 1$, there is $n > m$ and an $n$-bit prime $p_n$ such that $\mathsf{rK}^{t}(p_n) \leq \log(n) + O(1)$.
\end{corollary}

In other words, there are infinitely many primes that admit very short efficient descriptions. The bound in  \autoref{cor:intro_rK} improves upon  the sub-polynomial bound on $\mathsf{rK}^{\poly}(p_n)$ from \cite{DBLP:conf/stoc/LuOS21}.

In the next subsection, we describe at a high level the ideas in the proof of \autoref{theo:intro_main}, and how they relate to previous work.

\subsection{Proof Ideas}\label{sec:ideas}

The proof of \autoref{theo:intro_main} relies on \emph{uniform hardness-randomness tradeoffs} \cite{ImpagliazzoW01, DBLP:journals/cc/TrevisanV07}. For concreteness, assume that $Q = \{Q_n\}_{n \in \posN}$, with each $Q_n \subseteq \{0,1\}^n$ consisting of the set of $n$-bit prime numbers. Let $A_Q$ be a deterministic polynomial-time algorithm that decides $Q$ (\eg, $A_Q$ is the AKS primality test algorithm \cite{Agrawal02primesis}). Before we present our algorithm and the main ideas underlying our result, it is instructive to discuss the approach of \cite{DBLP:conf/stoc/OliveiraS17}, which provides a subexponential-time pseudodeterministic construction that succeeds on infinitely many input lengths.

\paragraph{Subexponential-time constructions of \cite{DBLP:conf/stoc/OliveiraS17}.} We first recall how uniform hardness-randomness tradeoffs work. Given a presumed hard language $L$, a uniform hardness-randomness tradeoff for $L$ states that either $L$ is easy for probabilistic polynomial-time algorithms, or else we can build a \emph{pseudorandom set} $G_n \subseteq \{0, 1\}^n$ computable in subexponential time (thus also has subexponential size), which fools probabilistic polynomial-time algorithms on inputs of length $n$ (for infinitely many $n$). In particular, Trevisan and Vadhan \cite{DBLP:journals/cc/TrevisanV07} give a uniform hardness-randomness tradeoff for a $\PSPACE$-complete language $L_{\mathsf{TV}}$ they construct, which has certain special properties tailored to uniform hardness-randomness tradeoffs.\footnote{For the pseudorandomness experts, these special properties are \emph{downward self-reducibility} and \emph{random self-reducibility}.}

The subexponential-time construction in \cite{DBLP:conf/stoc/OliveiraS17} uses a \emph{win-win} argument to derive an \emph{unconditional} pseudodeterministic algorithm from the uniform hardness-randomness tradeoff of \cite{DBLP:journals/cc/TrevisanV07}. There are two cases: either $L_{\mathsf{TV}} \in \mathsf{BPP}$, or it is not. If the former is the case, then $\mathsf{PSPACE} \subseteq \mathsf{BPP}$ by the $\mathsf{PSPACE}$-completeness of $L_{\mathsf{TV}}$. Now, since we can in \emph{polynomial space} test all $n$-bit numbers using $A_Q$ until we find the lexicographic first prime number, we can also do it in \emph{randomized polynomial time}, \ie, there is a randomized algorithm $B(1^n)$ that runs in polynomial time and outputs the lexicographically first $n$-bit prime with high probability. Thus, in this case, the lexicographically first $n$-bit prime is the ``canonical'' output of the pseudodeterministic algorithm, and the algorithm works on \emph{every} input length $n$.

Suppose, on the other hand, that $L_{\mathsf{TV}} \not \in \mathsf{BPP}$. Using the uniform hardness-randomness tradeoff of \cite{DBLP:journals/cc/TrevisanV07}, we have that for each $\varepsilon > 0$, there is a pseudorandom set $G = \{G_n\}$, where each $G_n \subseteq \{0,1\}^n$ is of size at most $2^{n^{\varepsilon}}$, such that for infinitely many $n$, $G_n$ fools the algorithm $A_Q$ on inputs of length $n$. Since $A_Q$ accepts an $\Omega(1/n)$ fraction of strings of length $n$ by the Prime Number Theorem, we have that the fraction of strings in $G_n$ that are prime is $\Omega(1/n)$ (by choosing the error parameter of the uniform hardness-randomness tradeoff to be small enough). In particular, there must exist an element of $G_n$ that is prime. Since $G_n$ is computable in subexponential time, we can define a subexponential time \emph{deterministic} algorithm that enumerates elements of $G_n$ and tests each one for primality until it finds and outputs one that is prime. This algorithm is deterministic but it runs in subexponential time, and is only guaranteed to be correct for infinitely many $n$.

Thus, in either case, we have a pseudodeterministic algorithm for constructing primes that runs in subexponential time and works infinitely often. Note that we do not know a priori which of the two cases above holds, and therefore the argument is somewhat non-constructive. By exploiting further properties of the uniform hardness-randomness tradeoff, \cite{DBLP:conf/stoc/OliveiraS17} manage to give an explicit construction algorithm that runs in subexponential time infinitely often.

\paragraph{Win-win arguments.} The above argument gives a subexponential-time construction, but the win-win structure of the argument seems incapable of giving an optimal polynomial-time construction. Indeed, this is the case for many win-win arguments used in complexity theory:\begin{itemize}
    \item A win-win argument based on the Karp--Lipton theorem \cite{KarpL80} gives that $\Sigma_2\EXP$ requires super-polynomial size Boolean circuits \cite{DBLP:journals/iandc/Kannan82}, but seems incapable of giving truly exponential ($2^{\Omega(n)}$) Boolean circuit lower bounds.\footnote{Partially inspired by this work, subsequent works \cite{CHR24, Li24} proved near-maximum circuit lower bounds for the classes $\Sigma_2\E$ and $\S_2\E$.}
    \item A win-win argument based on uniform hardness-randomness tradeoffs gives that either $\mathsf{E} \subseteq \mathsf{BPP}$ or $\mathsf{BPP}$ can be simulated infinitely often in deterministic subexponential time on average \cite{ImpagliazzoW01}, but it remains unknown if such a tradeoff holds at the ``high end'', \ie, whether it is the case that either $\mathsf{E}$ is in probabilistic subexponential-time or else $\mathsf{BPP}$ can be simulated infinitely often in deterministic polynomial time on average.
    \item A win-win argument based on the Easy Witness Lemma gives that if $\mathsf{NEXP} \subseteq \mathsf{SIZE}(\poly)$, then $\mathsf{NEXP} = \mathsf{MA}$ \cite{ImpagliazzoKW02}, but it is unknown if any interesting uniform collapse follows from the simulation of $\mathsf{NEXP}$ by subexponential-size Boolean circuits.
\end{itemize}
In each of these cases, the win-win argument seems to have inherent limitations that prevent us from getting optimal lower bounds or tradeoffs. Indeed, a paper by Miltersen, Vinodchandran and Watanabe \cite{DBLP:conf/cocoon/MiltersenVW99} studies the ``fractional exponential'' lower bounds that seem to be the best provable using win-win arguments in the context of Boolean circuit lower bounds for exponential-time classes.\footnote{For example, a function $f:\N \to \N$ is \emph{sub-half-exponential} if $f(f(n)^c)^c \le O(2^n)$ for every constant $c$. (The exact definition of sub-half-exponential functions may be different in different papers.) Functions such as $n^k$ and $2^{\log^k n}$ are sub-half-exponential, while $2^{\eps n}$ and $2^{n^\eps}$ are not. It is known that $\Sigma_2\EXP$ cannot be computed by $f(n)$-size circuits for every sub-half-exponential $f$, but it remains open to show that $\Sigma_2\EXP$ requires circuit complexity $2^{n^\eps}$ for any constant $\eps > 0$.}

Thus, in order to obtain a polynomial-time pseudodeterministic algorithm for primality, it seems that we need to go beyond win-win arguments. One natural idea is to apply uniform hardness-randomness tradeoffs \emph{recursively}. However, this seems hard to do with the uniform hardness-randomness tradeoff of  \cite{DBLP:journals/cc/TrevisanV07}. Their tradeoff applies only to the special language $L_{\mathsf{TV}}$. If we argue based on the hardness or other properties of $L_{\mathsf{TV}}$, then in the case where $L_{\mathsf{TV}} \in \mathsf{BPP}$, we get a pseudodeterministic polynomial-time algorithm for constructing primes, but in the case where $L_{\mathsf{TV}} \not \in \mathsf{BPP}$, we get a subexponential-time constructible pseudorandom set, and it is unclear how to apply the uniform hardness-randomness tradeoff to the algorithm for constructing this set.

\paragraph{Recursive application of uniform hardness-randomness tradeoffs.} One of our main ideas is to exploit very recent work on uniform hardness-randomness tradeoffs \cite{ChenT21b} which applies to \emph{generic} computations, as long as they satisfy certain mild properties. These tradeoffs yield \emph{hitting sets} rather than pseudorandom sets based on hardness --- a hitting set $\HSG \subseteq \{0,1\}^{M}$ is a set that has non-empty intersection with every $Q_M \subseteq \{0,1\}^M$ that is dense (\ie, accepts at least a $1/\poly(M)$ fraction of strings) and is efficiently computable. It turns out that for our application to pseudodeterministic algorithms, uniform hardness-randomness tradeoffs that yield hitting sets are sufficient.

Specifically, Chen and Tell \cite{ChenT21b} show that for any multi-output function $f \colon \{1^n\} \to \{0,1\}^n$ computed by uniform Boolean circuits of size $T = T(n)$ and depth $d = d(n)$, either there is a hitting set $\HSG \subseteq \{0,1\}^M$ computable in time $\poly(T)$ that fools the property $Q_M$, or $f (1^n)$ can be computed with high probability in time $(d+n) \cdot \poly(M)$ (which could be much less than $T$). Note that this tradeoff is applicable to \emph{any} multi-output function $f$ given bounds on its uniform circuit complexity.

Our key idea is that this more generic uniform hardness-randomness tradeoff can be applied \emph{recursively}. Indeed, we apply it to multi-output functions which capture the very task we are trying to solve, \ie, constructing a prime! In our base case, we use the function $f$ which does a brute-force search over $n$-bit numbers and outputs the lexicographically first one which is prime. This function can be computed by uniform Boolean circuits of size $2^{O(n)}$ and depth $\poly(n)$, and hence we can apply the Chen--Tell tradeoff to it. We set $M = n^{\beta}$ for some large enough constant $\beta > 1$ in the tradeoff. If we have that $f(1^n)$ is computable with high probability in time $(d+n) \cdot \poly(M)$, then we are done, since this gives us a pseudodeterministic algorithm for primes at length $n$. If not, we have that there is a hitting set $\HSG \subseteq \{0,1\}^{n^{\beta}}$ computable in time $2^{O(n)}$. In particular, by iterating over the elements of $\HSG$ and outputting the first one that is prime, we gain over the na\"ive brute-force search algorithm, since we are now outputting a prime of length $n^{\beta}$ in time $2^{O(n)}$. Now \emph{this} new algorithm can be captured by a multi-output function with output length $n^{\beta}$ to which we apply the Chen--Tell tradeoff again. In each recursive step, we either obtain a pseudodeterministic polynomial-time construction of primes, or we obtain a significantly faster deterministic construction of primes (of a larger input length). Intuitively, analyzing this process after $O(\log n)$ steps of recursion, we can hope to show that at least one of the steps leads to a polynomial-time pseudodeterministic algorithm at the input length considered at that step. 

This doesn't quite work as stated because the Chen--Tell tradeoff uses the Nisan--Wigderson generator \cite{NisanW94}, which is not known to have optimal parameters for all levels of hardness; see e.g.,~\cite{DBLP:conf/focs/ImpagliazzoSW99}.\footnote{Informally speaking, given a ``hard truth table'' of length $T$, we want to construct a hitting set $\HSG \subseteq \{0, 1\}^M$ in $\poly(T)$ time; however, the Nisan--Wigderson generator requires $2^{\Theta(\log^2 T / \log M)}$ time to construct.} Our recursive process explores essentially all possible levels of hardness for the uniform hardness-randomness tradeoff, since each recursive step corresponds to a different level of hardness. Using the original Chen--Tell tradeoff gives a \emph{quasi-polynomial-time} pseudodeterministic construction, but in order to get a polynomial-time pseudodeterministic construction, we need to work harder.

Another crucial idea for us is to optimize the Chen--Tell tradeoff by using the Shaltiel--Umans generator \cite{ShaltielU05} rather than the Nisan--Wigderson generator. This idea comes with its own implementation challenges, since the Shaltiel--Umans generator is not known to have a crucial learnability property that is required for the uniform hardness-randomness tradeoff. We sidestep this issue using a further win-win analysis, together with some other tricks; see \autoref{sec: overview of SU} for details. This enables us to achieve an optimal polynomial-time pseudodeterministic construction on infinitely many input lengths, and thereby establish \autoref{theo:intro_main}.\footnote{While we do not explore this direction in the current work, we believe that our improvement on the Chen-Tell tradeoff can be used to improve the tradeoff from~\cite[Theorem~5.2 and Theorem~5.3]{CRT22}, thus getting a better uniform hardness vs randomness connection in the low-end regime.} We note that the subexponential-time construction of  \cite{DBLP:conf/stoc/OliveiraS17} also only works for infinitely many input lengths, and it is still open even to get a subexponential-time construction that works on all input lengths.

The intuitive description here does not address several subtleties that arise in the proof, such as maintaining the right uniformity and depth conditions when recursively applying the uniform hardness-randomness tradeoff. We refer to  \autoref{sec:techniques} for a more detailed discussion of such matters.

\subsection{Technical Overview}\label{sec:techniques}

As explained above, we consider a chain of $t = O(\log n)$ recursively defined (candidate) HSGs $\HSG_0, \HSG_1, \ldots, \HSG_t$ operating over different input lengths. These HSGs are obtained from the recent construction of Chen and Tell \cite{ChenT21b}, which we informally describe next. Recall that we use $Q_M$ to denote the easy and dense property over inputs of length $M$.

\paragraph{The Chen--Tell \cite{ChenT21b} targeted HSG (``ideal version'').} Let $f \colon \{1^n\} \to \{0,1\}^n$ be a family of unary functions computed by (uniform) Boolean circuits of size $T = T(n)$ and depth $d = d(n)$. Then, for some large enough universal constant $c\ge 1$, for every $\log T \leq M \leq T$ there is a set $\HSG \subseteq \{0,1\}^{M}$ computable in
\[
\text{time}~\widetilde{T} \eqdef T^c~~\text{and}~~\text{depth}~\widetilde{d} \eqdef d \cdot \log(T) + M^c
\]
such that, if $Q_M \subseteq \{0,1\}^M$ \emph{avoids} $\HSG$, (\ie, $Q_M$ is dense but $Q_M \cap \HSG = \emptyset$), then we can compute $f(1^n)$ with high probability in time $(d + n) \cdot M^c$.\\

In other words, if $f$ admits \emph{low-depth} circuits, we can construct a candidate HSG $\HSG$ over length-$M$ inputs such that breaking the generator $\HSG$ allows us to compute $f(1^n)$ in time $\poly(n,d,M)$. For $d, M  \ll T$,  this can be much faster than the original time $T$ required to compute $f$.

The statement above differs from the results in \cite{ChenT21b} (stated for unary functions) in two important ways. First, the claimed upper bound on $\widetilde{T}$ (the running time of the HSG) is not obtained by  \cite{ChenT21b} for all choices of $M$. Secondly, we have not formally specified the \emph{uniformity} of the family of circuits computing $f$. While these are crucial points in \cite{ChenT21b} and when proving our result, for simplicity we will assume for now that this upper bound can be achieved and omit the discussion on uniformity.

\paragraph{Bootstrapping the win-win argument.} We now review the idea discussed in \autoref{sec:ideas}, using notations that will be more convenient for the remainder of this technical overview.  Fix an arbitrary $n \in \posN$, and consider the corresponding property $Q_{n} \subseteq \{0,1\}^{n}$ decided by $A_Q(x)$ on inputs of length $n$. Our initial $\HSG_0$ is trivial and set to $\{0,1\}^{n}$. (Intuitively, this corresponds to the first case of the \cite{DBLP:conf/stoc/OliveiraS17} argument sketched above where $L_{\sf TV} \in \BPP$.) Consider now a ``brute-force'' algorithm $\BF(1^{n})$ that computes the first $x \in \HSG_0$ such that $A_Q(x) = 1$. We let $f(1^{n}) \eqdef \BF(1^{n})$ in the Chen--Tell HSG. Note that $f(1^{n})$ can be uniformly computed in time $T = 2^{O(n)}$ and depth $d = \poly(n)$, since $A_Q(x)$ runs in polynomial time and all elements of $\HSG_0$ can be tested in parallel. We set $M(n) \eqdef n^{\beta}$, where $\beta > 1$ is a large enough constant. Let $\HSG_1 \subseteq \{0,1\}^M$ be the candidate HSG provided by Chen--Tell. Note that $\HSG_1$  can be computed in time $\widetilde{T} = 2^{O(n)}$ and depth $\widetilde{d} = \poly(n)$.

Next, we consider a win-win argument based on whether $Q_M$ avoids $\HSG_1$. If this is the case, then Chen--Tell guarantees that we can compute $f(1^n) = \BF(1^n) \in Q_n$ with high probability in time $(d + n) \cdot M^c = \poly(n)$. In other words, we can pseudodeterministically produce a string in $Q_n$ in polynomial time. On the other hand, if $\HSG_1 \cap Q_M \neq \emptyset$, we now have a set $\HSG_1$ of strings of length $M = n^{\beta}$ that contains a string in $Q_M$ and that can be deterministically computed in time $2^{O(n)}$. That is, we are back to the former case, except that we can compute $\HSG_1$ (a set containing at least one $M$-bit prime) in time much faster than $2^{O(M)}$. Crucially, in contrast  to the approach of \cite{DBLP:conf/stoc/OliveiraS17}, the Chen--Tell HSG does not limit us to the use of the special language $L_{\mathsf{TV}}$,  effectively allowing us to reapply the same argument (with a speedup) over a larger input length.\\

In the next subsection, we discuss the ``bootstrapping'' and its parameters in more detail and explain how it gives a polynomial-time pseudodeterministic  construction, assuming we have the ideal version of \cite{ChenT21b} described above.

\subsubsection{Infinitely-Often Pseudodeterministic Polynomial-Time Constructions}
Let $n_0\in\N$ be an ``initial'' input length and $t = O(\log n_0)$ be a parameter. For each $1\le i\le t$, we define the $i$-th input length to be $n_i \eqdef n_{i-1}^\beta$, for a large enough constant $\beta > 1$. Our goal is to design a pseudodeterministic algorithm for finding elements in $Q$ that will be correct on \emph{at least one of the input lengths} $n_0, n_1, \dots, n_t$. By considering different ``initial'' input lengths $n_0^{(1)}, n_0^{(2)}, \dots$ and setting each $n_0^{(i+1)}$ to be sufficiently large compared to the previous $n_t^{(i)}$, we obtain a pseudodeterministic construction on \emph{infinitely many} input lengths. On each input length $n_i$ we will have:\begin{enumerate}
	\item the property $Q_{n_i}$ that we want to hit;
	\item a candidate hitting set generator $\HSG_i \subseteq \{0, 1\}^{n_i}$; and
	\item the brute-force algorithm $\BF_i:\{1^{n_i}\} \to \{0, 1\}^{n_i}$, which iterates through all elements in $\HSG_i$ and outputs the first element that is in $Q_{n_i}$.
\end{enumerate}

Note that $\BF_i$ is completely defined by $\HSG_i$. Suppose that $\HSG_i$ can be computed (deterministically) in time $T_i$ and depth $d_i$, then $\BF_i$ can also be computed (deterministically) in time $T'_i \eqdef T_i\cdot \poly(n_i)$ and depth $d'_i \eqdef d_i\cdot \poly(n_i)$. As discussed above, initially, $\HSG_0 \eqdef \{0, 1\}^{n_0}$ is the trivial hitting set generator, $T_0 \eqdef 2^{O(n_0)}$, and $d_0 \eqdef \poly(n_0)$.

For each $0\le i < t$, we let $f(1^{n_i}) \eqdef \BF_i, M \eqdef n_{i+1}$, and invoke the Chen--Tell HSG to obtain the HSG $\HSG_{i+1} \subseteq \{0, 1\}^{n_{i+1}}$. Recall that Chen--Tell guarantees the following: Suppose that $Q_M = Q_{n_{i+1}}$ avoids the HSG $\HSG_{i+1}$, then one can use $Q_{n_{i+1}}$ to compute $f(1^{n_i})$ with high probability in time $\poly(d'_i,n_i,M) \le \poly(d_i, n_i)$, by our choice of parameters. Recall that if $\HSG_i$ indeed hits $Q_{n_i}$, then $f(1^{n_i})$ implements the brute-force algorithm and outputs the first element in $\HSG_i\cap Q_{n_i}$ (\ie,~a \emph{canonical} element in $Q_{n_i}$). To reiterate, Chen--Tell gives us the following win-win condition:\begin{itemize}
	\item \emph{either} $Q_{n_{i+1}}$ avoids $\HSG_{i+1}$, in which case we obtain a probabilistic algorithm that outputs a canonical element in $Q_{n_i}$ (thus a pseudodeterministic algorithm) in $\poly(d_i, n_i)$ time;
	\item \emph{or} $\HSG_{i+1}$ hits $Q_{n_{i+1}}$, in which case we obtain a hitting set $\HSG_{i+1}$ that hits $Q_{n_{i+1}}$, thereby making progress on input length $n_{i+1}$.
\end{itemize}

The HSG $\HSG_{i+1}$ can be computed in time $T_{i+1} \eqdef (T'_i)^c$ and depth $d_{i+1} \eqdef d'_i\cdot \log T'_i + n_{i+1}^c$. Crucially, although $T_0$ is exponential in $n_0$, it is possible to show by picking a large enough $\beta > 1$ that the sequence $\{n_i\}_{i\in \N}$ grows faster than the sequence $\{T_i\}_{i\in\N}$, and eventually when $i = t = O(\log n_0)$, it will be the case that $T_t \le \poly(n_t)$ and we can apply the brute-force algorithm to find the first element in $\HSG_t$ that is in $Q_{n_t}$ in time polynomial in $n_t$.

A more precise treatment of the growth of the two sequences $\{n_i\}$ and $\{T_i\}$ are as follows. There is some absolute constant $\alpha \ge 1$ such that $T_0 \le 2^{\alpha n_0}$ and
\[T_{i+1} \le T_i^\alpha\text{ (for each $0\le i < t$)}.\]
We set $\beta \eqdef 2\alpha$ (recall that each $n_{i+1} = n_i^\beta$). It follows from induction that for each $0\le i\le t$,
\[T_{i+1} \le T_0^{\alpha^i} = 2^{\alpha^{i+1}n_0}~~\text{and}~~n_{i+1} = n_i^{\beta} = n_0^{\beta^{i+1}} = n_0^{(2\alpha)^{i+1}}.\]
Since
\[\frac{\log T_t}{\log n_t} \le\frac{\alpha^{t}n_0}{(2\alpha)^t\log n_0}=\frac{n_0}{2^t\log n_0},\]
it follows that when $t \approx \log(n_0 / \log n_0)$, $T_t$ will be comparable to $n_t$ (rather than $2^{n_t}$). Similarly, one can show that $d_i \le \poly(n_i)$ for every $i\le t$.

\paragraph{Informal description of the algorithm and correctness.} %
To wrap up, we arrive at the following pseudodeterministic algorithm that is correct on at least one of the input lengths $n_0, n_1, \dots, n_t$. On input length $n_i$, if $i = t$, then we use $\poly(T_t) \le \poly(n_t)$ time to find the first string in $\HSG_i$ that is also in $Q_{n_i}$ (\ie, simulate $\BF_i$); otherwise, use $Q_{n_{i+1}}$ as a distinguisher for the Chen--Tell hitting set $\HSG_i$ and print the output of $\BF_i$ in $\poly(n_i, d_i) \le \poly(n_i)$ time. To see that our algorithm succeeds on at least one $n_i$, consider the following two cases:\begin{enumerate}
	\item Suppose that $\HSG_t$ indeed hits $Q_{n_t}$. Then clearly, our algorithm succeeds on input length $n_t$.
	\item On the other hand, suppose that $\HSG_t$ does not hit $Q_{n_t}$. Since our trivial HSG $\HSG_0$ hits $Q_{n_0}$, there exists an index $0\le i < t$ such that $\HSG_i$ hits $Q_{n_i}$ but $Q_{n_{i+1}}$ avoids $\HSG_{i+1}$.
	
	Since $Q_{n_{i+1}}$ avoids $\HSG_{i+1}$, Chen--Tell guarantees that we can speed up the computation of $\BF_i$ using $Q_{n_{i+1}}$ as an oracle. Since $\HSG_i$ hits $Q_{n_i}$, the output of $\BF_i$ is indeed a canonical element in $Q_{n_i}$. It follows that our algorithm succeeds on input length $n_i$.
\end{enumerate}
This completes the sketch of the algorithm and its correctness. We note that while this exposition explains how the second bullet of \autoref{theo:intro_main} is achieved, it does not address the behavior of the algorithm on other input lengths (\ie, the first bullet in the same statement). For simplicity, we omit this here and refer to the formal presentation in \autoref{sec:main_proof}.\footnote{Alternatively, the guarantee from the first bullet of \autoref{theo:intro_main} can always be achieved via a general argument. We refer to \cite[Proposition 2]{DBLP:conf/stoc/OliveiraS17} for the details.}\\

While the aforementioned construction conveys the gist of our approach, there are two important issues with our presentation. Firstly, as explained before, the results of \cite{ChenT21b} do not achieve the \emph{ideal parameters} of the HSG stated above. Secondly, we have only vaguely discussed the \emph{circuit uniformity} of the function $f(1^n)$. The uniformity of $f$ is critical for the reconstruction procedure of \cite{ChenT21b} to run in time comparable to the circuit depth of $f$. On the other hand, since our HSGs and functions $f$ (corresponding to the algorithm $\BF$) are recursively defined, the circuit uniformity of the \cite{ChenT21b} generator itself becomes another critical complexity measure in the proof. 

In the next subsection, we discuss the Chen--Tell generator in more detail and explain how to obtain an improved generator construction satisfying our requirements.

\subsubsection{Improving the Chen--Tell Targeted Hitting Set Generator}\label{sec:improved_CT_exposition}

The uniform hardness-to-randomness framework of Chen--Tell builds on two important ingredients:\footnote{Below we will focus on the high-level picture of the Chen--Tell framework without diving into too many details. Our presentation is also somewhat different from the original presentation in~\cite{ChenT21b}.}
\begin{enumerate}
	\item A \emph{layered-polynomial representation} of a shallow uniform circuit.
	
	\item A hitting set generator with a \emph{uniform learning reconstruction} algorithm.
\end{enumerate}

\newcommand{\OUT}{\mathsf{OUT}}

\paragraph*{Layered-polynomial representation.}
We now discuss the first ingredient. Let $f \colon \bs{n} \to \bs{n}$ be a logspace-uniform circuit family of size $T(n)$ and depth $d(n)$.\footnote{Intuitively, a circuit family is logspace-uniform if each circuit in the family can be printed by a fixed machine that runs in  space that is of logarithmic order in the size of the circuits. See~\autoref{sec:circuit_uniformity} for the precise definition of logspace-uniform circuits.} Let $M \colon \mathbb{N} \to \mathbb{N}$ be the parameter for output length. Building on the doubly efficient interactive proof system by~\cite{gkr15} (and its subsequent simplification by~\cite{gol17a}), for any $z \in \bs{n}$,~\cite{ChenT21b} showed that there is a sequence of polynomials $\{ P_i^{z} \}_{i \in [d']}$ for $d' = d \cdot \polylog(T)$ with the following nice properties:
\begin{itemize}
	\item ({\bf Arithmetic setting.}) Let $\F$ be a finite field of size $M^{c}$ for a large universal constant $c > 1$, and let  $m$ be of order $\frac{\log T}{\log M}$. All the $P_i^z$ map $\F^{m}$ to $\F$ and have total degree at most $M$.
	
	\item ({\bf Base case.}) There is an algorithm $\Base$ such that, given the input $z \in \bs{n}$ and $\vec{w} \in \F^m$,  computes $P_1^z(\vec{w})$ in $\poly(M)$ time.
	
	\item ({\bf Downward self-reducibility.}) There is an oracle algorithm $\DSR$ that, given input $i \in \{2,\dotsc,d'\}$ and $\vec{w} \in \F^m$, together with the oracle access to $P_{i-1}^{z}(\cdot)$, computes $P_i^z(\vec{w})$ in $\poly(M)$ time.
	
	\item ({\bf Faithful representation.}) There is an oracle algorithm $\OUT$ that, given input $i \in [n]$ and oracle access to $P_{d'}^z$, outputs $f(z)_i$ in $\poly(M)$ time.
\end{itemize}

Intuitively, these polynomials form an \emph{encoded} version of the computation of $f$ in the sense that they admit both \emph{downward self-reducibility} and \emph{random self-reducibility}: every $P_i^z$ has low degree and hence admits error correction properties; downward self-reducibility follows from definition.

We note that the proof of this result depends in a crucial way on the logspace-uniformity of the circuit family computing $f$. (This allows one to arithmetize a formula of bounded size that computes the direct connection language of the circuit, while also controlling the circuit uniformity of the resulting polynomials.)

\paragraph*{Hitting set generators with a uniform learning reconstruction algorithm.} The second ingredient of~\cite{ChenT21b} is the Nisan-Wigderson generator combined with Reed-Muller codes~\cite{NisanW94,SudanTV01}. The most important property of this generator is that it supports a uniform learning reconstruction algorithm. In more detail, for a polynomial $P\colon \F^{m} \to \F$, the generator $\NW^{P}$ takes $s = O\mleft(\frac{\log^2 T}{\log M}\mright)$ bits as seed, such that there is a uniform oracle algorithm $R$ (for ``reconstruction'') where the following holds. Given oracle access to both $P$ and an oracle $\bfunc{D}{M}{}$ that distinguishes $\NW^{P}(U_s)$ from the uniform distribution, $R^{P,D}$ runs in $\poly(M)$ time and with high probability outputs a polynomial-size $D$-oracle circuit that computes $P$.

Now, the hitting set $H_f(z)$ is defined as
\[
H_f(z) \eqdef \bigcup_{i \in [d']} \NW^{P_i^{z}}\mathdot
\]

\paragraph*{The uniform reconstruction algorithm.} One key observation here is that if a distinguisher $\bfunc{D}{M}{}$ avoids $H_f(z)$, meaning that $D$ accepts a large fraction of inputs from $\bs{M}$ but rejects all strings in $H_f(z)$, then clearly $D$ also distinguishes all $\NW^{P_i^z}(U_s)$ from the uniform distribution. Following~\cite{ImpagliazzoW01},~\cite{ChenT21b} then shows that there is a uniform oracle algorithm $R_f$ that takes input $z \in \bs{n}$ and any ``avoider'' $D$ of $H_f(z)$ as oracle, and outputs $f(z)$ with high probability. In more detail, $R_f$ works as follows:

\begin{enumerate}
	\item It is given input $z \in \bs{n}$ and oracle access to an avoider $\bfunc{D}{M}{}$ of $H_f(z)$.
	
	\item For every $i \in \{2,\dotsc,d'\}$:
	
	\begin{enumerate}
		\item The goal of the $i$-th step is to construct a $\poly(M)$-size $D$-oracle circuit $C_i$ that computes $P_i^z$.
		
		\item It runs the learning reconstruction algorithm $R^{P_i^z,D}$ to obtain a $\poly(M)$-size $D$-oracle circuit. To answer queries to $P_i^z$, we first run the algorithm $\DSR$ to convert them into queries to $P_{i-1}^z$. Next, when $i = 2$, we answer these queries by calling $\Base$ directly, and when $i > 2$ we answer these queries by evaluating our $D$-oracle circuit $C_{i-1}$.
	\end{enumerate}
	
	\item For every $i \in [n]$, output $\OUT^{C_{d'}^D}(i)$.
	
\end{enumerate}

\paragraph*{Issue with the original Chen--Tell construction: Super-logarithmic seed length of $\NW$.} The main issue with the construction above is that $\NW^{P_i^z}$ has seed length $O\!\left(\frac{\log^2 T}{\log M}\right)$. In particular, this means that when $\log M \le o(\log T)$, the hitting set $H_f(z)$ has super-polynomial size, and therefore cannot be computed in $\poly(T)$ time as in the ``ideal version'' of \cite{ChenT21b} stated above. (See, e.g.,~\cite{DBLP:conf/focs/ImpagliazzoSW99}.)\footnote{Indeed, if we rely on the original Chen--Tell construction to implement the bootstrapping method described above, we would only obtain a quasi-polynomial-time pseudodeterministic construction, instead of a polynomial-time one.} Hence, to improve the computation time of $H_f(z)$ to $\poly(T)$, we need an HSG with seed length $O\!\left(\log T\right)$ for all possible values of $M$, together with a uniform learning reconstruction, when it is instantiated with polynomials. Jumping ahead, we will replace $\NW$ with the Shaltiel--Umans Hitting Set Generator~\cite{ShaltielU05}, obtaining an optimized version of the Chen--Tell generator with better parameters. However, the original generator from~\cite{ShaltielU05} does not provide a uniform learning reconstruction procedure. By using the classical construction of a  \emph{cryptographic pseudorandom generator from a one-way permutation} and of another idea, we managed to  modify their construction to allow a uniform learning reconstruction. See the next subsection for more details.

\paragraph*{Controlling the circuit uniformity of the optimized Chen--Tell generator.} As stressed above, in order to construct a layered-polynomial representation for $f$ with the aforementioned parameters, it is crucial that $f$ admits a logspace-uniform circuit family. Since we will rely on multiple applications of the generator, and each new function $\BF$ on which the result is invoked contains as a subroutine the code of the previous generator, we must \emph{upper bound the circuit uniformity} of our optimized Chen--Tell generator. This turns out to require a delicate manipulation of all circuits involved in the proof and of the Turing machines that produce them, including the components of the Shaltiel--Umans generator. For this reason, whenever we talk about a Boolean circuit in the actual proof, we also bound the description length and space complexity of its corresponding machine. Additionally, as we manipulate a super-constant number of circuits (and their corresponding machines) in our construction, we will also consider the complexity of producing the code of a  machine $M_2$ encoding a circuit $C_2$ from the code of a machine $M_1$ encoding a circuit $C_1$ (see, e.g., the ``Moreover'' part in the statement of \autoref{thm:HSG-Chen-Tell-poly}). The details are quite tedious, but they are necessary for verifying the correctness and running time of our algorithm. In order to provide some intuition for it, we notice that as we move from the HSG $\HSG_i$ to $\HSG_{i + 1}$, we also increase the corresponding input length parameter from $n_i$ to $n_{i + 1} = n_i^\beta$. While there is an increase in the uniformity complexity, it remains bounded relative to the new input length.  (Think of a truncated geometric series whose value is dominated by the complexity over the current input length.) We omit the details in this proof overview.

\paragraph{Non-black-box behavior.} We note that the recursive application of the Chen--Tell generator is responsible for the ``non-black-box'' behavior of our pseudodeterministic construction. Indeed, since we invoke the Chen--Tell generator on each function $\BF$ (which contains the code of the algorithm $A_Q$ deciding property $Q$ as a subroutine), the collection of strings in the hitting set generator depends on the layered-polynomial representation that is obtained from the \emph{code} of $\BF$. As a consequence, our construction has the unusual feature that the canonical outputs of the algorithm $B$ in \autoref{theo:intro_main} are affected by the code of $A_Q$. In other words, by using a different primality test algorithm (or by making changes to the code implementing the AKS routine), one might get a different $n$-bit prime!\\

The parameters of our hitting set generator appear in \autoref{sec:main_proof}. The proof of the result is given in \autoref{sec:improved-CT}.

\subsubsection{Modified Shaltiel--Umans Generator with Uniform Learning Reconstruction}\label{sec: overview of SU}

As explained above, 
in order to complete the proof of \autoref{theo:intro_main} we need to design a variant of the Shaltiel--Umans generator \cite{ShaltielU05} with a \emph{uniform learning reconstruction} procedure.

The Shaltiel--Umans generator takes as input a low-degree polynomial $P:\F_{p}^m\to\F_p$ (in our case $p$ will be a power of $2$) and produces a set of binary strings (which is supposed to be a hitting set). The construction of this generator also relies on ``generator matrices''. A matrix $A\in\F_p^{m\times m}$ is a \emph{generator matrix} if it satisfies $\{A^i\cdot \vec{1}\}_{1\leq i<p^m}=\F^m_p\setminus\{\vec{0}\}$. Roughly put, the matrix $A$ can be thought of as performing multiplication with a generator of the multiplicative group of $\F_{p^m}$. 

Recall that a generator has a uniform learning reconstruction algorithm if the following holds. Given an algorithm $D$ that avoids the output of the generator constructed using $P$, as well as $P$ itself, we can \emph{uniformly} and \emph{efficiently} generate (with high probability) a $D$-oracle circuit that computes the polynomial $P$. (In other words, we can query $P$ while producing the circuit, but the circuit itself does not have access to $P$.)

However, the reconstruction procedure provided by the original Shaltiel--Umans generator only guarantees the following: If the generator is constructed using $P$ and some generator matrix $A$, then using an algorithm $D$ that avoids the output of the generator, and \emph{given the matrix} $A$ and oracle access to $P$, one can obtain a ($D$-oracle) circuit $C:[p^m-1]\to\F^m_p$ such that $C(i)=P(A^{i}\cdot \vec{1})$.\footnote{In fact, the circuit only computes $P(A^{i}\cdot \vec{v})$ for some $\vec{v}$ output by the reconstruction algorithm. We assume $\vec{v} = \vec{1}$ here for simplicity.} (For the precise statement, see \autoref{t:HSU}.) That is, this reconstruction is not a uniform learning algorithm in the following sense:
\begin{enumerate}
	\item It needs to know the matrix $A$ (which can be viewed as non-uniform advice).
	\item Given oracle access to $P$, it only learns a circuit that computes the mapping $i\mapsto P(A^{i}\cdot \vec{1})$, instead of a circuit that computes $P(\vec{x})$ on a given $\vec{x} \in \F_p^m$.
\end{enumerate}
We now describe how to modify the Shaltiel--Umans generator to make its reconstruction a uniform learning algorithm. 

For the first issue, our idea is that, instead of using a generator matrix that is obtained by brute-force search as in the original construction (we note that the reconstruction cannot afford to perform the brute-force search due to its time constraints), we will use a generator matrix that is from a small set of matrices that can be constructed \emph{efficiently}. More specifically, using results about finding primitive roots of finite fields (\eg, \cite{Shou92}), we show that one can efficiently and deterministically construct a set  $S$ of matrices that contains at least one generator matrix. %
The advantage is that the reconstruction algorithm can still afford to compute this set $S$. Note that although we don't know which matrix in $S$ is a valid generator matrix (as verifying whether a matrix is a generator matrix requires too much time), we can try all the matrices from $S$, and one of them will be the correct one. This allows us to obtain a list of candidate circuits, one of which computes $P$ (provided that we can also handle the second issue, which will be discussed next). Then by selecting from the list a circuit that is sufficiently close to $P$ (note that given oracle access to $P$, we can easily test whether a circuit is close to $P$ by sampling) and by using the \emph{self-correction} property of low-degree polynomials, we can obtain a circuit that computes $P$ exactly.

With the above idea, we may now assume that in the reconstruction we know the generator matrix $A$ used by the Shaltiel--Umans generator. Next, we describe how to handle the second issue. Recall that the reconstruction algorithm of the Shaltiel--Umans generator gives a circuit $C$ such that $C(i)=P(A^{i}\cdot \vec{1})$, for $i\in [p^m-1]$, and we want instead a circuit that given $\vec{x}\in\F_p^m$ computes $P(\vec{x})$. Now suppose given $\vec{x}\in\F_p^m\setminus\{\Vec{0}\}$, we can also \emph{efficiently} compute the value $i\in[p^m-1]$ such that $A^i\cdot \vec{1}=\vec{x}$. Then we would be able to combine this with $C$ to get a circuit $E$ that computes $P$, \ie, if $\Vec{x}=\Vec{0}$ then $E$ outputs $P(\Vec{0})$ (where the value $P(\Vec{0})$ can be hardcoded); otherwise, $E$ computes $i$ for $\vec{x}$ as described above and then outputs $C(i)$. However, the task of finding such $i$ given $A$ and $\vec{x}$ is essentially the \emph{discrete logarithm problem}, for which no efficient algorithm is known!

A classical result in cryptography is that one can construct a pseudorandom generator based on the hardness of the discrete logarithm problem (see, \eg,~\cite{BlumM84,Yao82}). %
More generally, given a permutation $f$ whose inverse admits \emph{random self-reducibility}\footnote{Roughly speaking, a function has random self-reducibility if computing the function on a given instance can be efficiently reduced to computing the function for uniformly random instances.}, one can construct a generator $G$ based on $f$ so that if there is a distinguisher $D$ that breaks $G$, then it can be used to invert $f$ via a uniform reduction. Our idea is to consider the bijection $f:[p^m-1] \to \F_p^m\setminus \{\vec{0}\}$ such that for each $i\in[p^m-1]$, $f(i)=A^{i}\cdot\vec{1}$ (where the random self-reducibility of $f^{-1}$ follows easily from that of the discrete logarithm problem), and try to construct a pseudorandom generator $G$ based on $f$. We then combine the output of $G$ with that of the Shaltiel--Umans generator constructed with the polynomial $P$ and the generator matrix $A$. Now if there is an algorithm $D$ that avoids this combined generator, which means $D$ \emph{simultaneously} avoids both the Shaltiel--Umans generator and the generator $G$, then $D$ can be used to obtain
\begin{itemize}
	\item a circuit $C$ such that $C(i)=P(A^{i}\cdot \vec{1})$ for every $i\in[p^m-1]$, and
	\item a circuit $C'$ that inverts $f$, \ie, $C'(\vec{x})$ outputs $i$ such that $A^i\cdot\vec{1}=\vec{x}$ for every $\vec{x}\in\F_p^m\setminus\{\Vec{0}\}$.
\end{itemize}
Then it is easy to combine $C$ and $C'$ to obtain a circuit that computes $P$.

A careful implementation of these ideas allows us to obtain a variant of the Shaltiel--Umans generator with uniform learning reconstruction, as needed in our  optimized Chen--Tell generator. We refer to \autoref{theo:SU-main} in \autoref{sec:SU-Gen} for more details. 

This completes the sketch of the proof of \autoref{theo:intro_main}.

\paragraph{Further remarks about the proof.} We note that in our proof the gap between two good input lengths on which the algorithm outputs a canonical prime can be exponentially large. It would be interesting to develop techniques to reduce this gap.

Additionally, the proof assumes the existence of a deterministic polynomial-time algorithm that decides the dense property. In contrast, the sub-exponential time algorithm from \cite{DBLP:conf/stoc/OliveiraS17} also works with a dense property that is decidable by a randomized polynomial-time algorithm. This is caused by the ``non-black-box'' nature of our approach via the Chen-Tell generator, which employs the code of the algorithm $A$ deciding the property as part of the description of the generator. Consequently, as alluded to above, changing the code of $A$ could result in a different canonical output on a given input length. If $A$ is randomized, fixing the randomness of $A$ is similar to the consideration of a different algorithm that decides the property, and it is not immediately clear how to maintain the pseudodeterministic behaviour in this case.

\section{Preliminaries}

For a positive integer $k$, we use $[k]$ to denote the set $\{1,2,\ldots,\}$. We use $\N$ to denote all non-negative integers and $\posN$ to denote all positive integers.

For $x,y \in \bs{*}$, we use $x \circ y$ to denote their concatenation.\footnote{We sometimes also use $C_1 \circ C_2$ to denote the composition of two circuits, but the meaning of the symbol $\circ$ will always be clear from the context.} For a function $\bfunc{f}{\ell}{}$ we use $\tt(f)$ to denote the $2^\ell$-length truth-table of $f$ (\ie, $\tt(f) = f(w_1)\circ f(w_2)\circ \dotsc\circ f(w_{2^\ell})$, where $w_{1},\dotsc,w_{2^\ell}$ is the enumeration of all strings from $\bs{\ell}$ in the lexicographical order).

Unless explicitly stated otherwise, we assume that all circuits are comprised of Boolean $\mathtt{NAND}$ gates of fan-in two. In several places in the paper we will need the following notion, which strengthens the standard notion of a time-computable function by requiring the function to be computable in logarithmic space. The depth of a circuit is defined to be the maximum length (measured by the number of edges) of any input-to-output path.

\begin{definition} [Logspace-Computable Functions]
	We say that a function $\func{T}{\N}{\N}$ is {\sf logspace-computable} if there exists an algorithm that gets input $1^n$, runs in space $O(\log(T(n)))$, and outputs $T(n)$.
\end{definition}

For convenience, we consider circuit families indexed by a tuple of parameters. Specifically, a circuit family with $k$ input parameters $\vec{\ell} = (\ell_1,\ell_2,\dotsc,\ell_k) \in \N^k$ is defined as $\{ C_{\vec{\ell}} \}_{\vec{\ell} \in \N^k}$, where each $C_{\vec{\ell}}$ is a circuit.

\subsection{Finite Fields}\label{sec:fields}

Throughout this paper, we will only consider finite fields of the form $\GF(2^{2 \cdot 3^\lambda})$ for some $\lambda \in \N$ since they enjoy simple representations that will be useful for us. We say $p = 2^{r}$ is a \emph{nice power} of $2$, if $r = 2 \cdot 3^{\lambda}$ for some $\lambda \in \N$.

Let $\ell \in \N$ and $n = 2 \cdot 3^\ell$. In the following, we use $\F$ to denote $\F_{2^n}$ for convenience. We will always represent $\F_{2^n}$ as $\F_2[\fx] / (\fx^{n} + \fx^{n/2} + 1)$.\footnote{$\fx^{2 \cdot 3^\ell} + \fx^{3^\ell} + 1 \in \F_2[\fx]$ is irreducible, see~\cite[Theorem~1.1.28]{van2012introduction}.} That is, we identify an element of $\F_{2^{n}}$ with an $\F_2[\fx]$ polynomial with degree less than $n$. To avoid confusion, given a polynomial $P(\fx) \in \F_2[\fx]$ with degree less than $n$, we will use $(P(\fx))_{\F}$ to denote the unique element in $\F$ identified with $P(\fx)$.

Let $\kappa^{(n)}$ be the natural bijection between $\{0,1\}^{n}$ and $\F = \GF(2^{n})$: for every $a \in \bs{n}$, $\kappa^{(n)}(a) = \left( \sum_{i \in [n]} a_i \cdot \fx^{i-1}\right)_{\F}$. We always use $\kappa^{(n)}$ to encode elements from $\F$ by Boolean strings. That is, whenever we say that an algorithm takes an input from $\F$, we mean it takes a string $x \in \bs{n}$ and interprets it as an element of $\F$ via $\kappa^{(n)}$. Similarly, whenever we say that an algorithm outputs an element from $\F$, we mean it outputs a string $\bs{n}$ encoding that element via $\kappa^{(n)}$. For simplicity, sometimes we use $(a)_{\F}$ to denote $\kappa^{(n)}(a)$. Also, when we say the $i$-th element in $\F$, we mean the element in $\F$ encoded by the $i$-th lexicographically smallest Boolean string in $\bs{n}$.

\subsection{Bounded-Space Turing Machines}

Our argument is robust to specific details about the computational model, but in order to estimate the relevant bounds, we must fix a model. We use the standard model of space-bounded computation (see~\cite[Section~5]{Goldreich-book} or~\cite[Section~4]{AB09-book}). A deterministic space-bounded Turing machine has three tapes: an input tape (that is read-only); a work tape (that is read/write) and an output tape (that is write-only and uni-directional). We assume that the machine’s alphabet is $\Sigma \eqdef \{0, 1\}$. The space complexity of the machine is the number of used cells on the work tape. For concreteness, we assume that the work tape contains initially only $\square$ (``blank'') symbols, and that the machine writes symbols from $\Sigma$ in the tape.

Throughout the paper, we will describe a space-bounded Turing machine by fixing a universal Turing machine $U$ that has an additional read-only \emph{program tape} such that $\TurM(x)$ is defined to be the output of $U$ with the program tape initialized as $\TurM$.\footnote{The advantage of fixing a universal Turing machine is that now our Turing machine always has a constant number of states, which is helpful when bounding the number of configurations of a Turing machine of  super-constant size.} Abusing the notation, we often use $\TurM$ to denote both the Turing machine and a binary string description of the Turing machine. Without loss of generality, we also assume our description is \emph{paddable} meaning that for every $\TurM \in \bs{*}$ and $k \in \N$, $\TurM$ and $\TurM \circ 0^{k}$ represent the same machine. To avoid certain technicalities, we will always assume that the space bound of a Turing machine $\TurM$ is greater than its description size.

\paragraph*{Configurations of space-bounded machines.} On a fixed input $x \in \bs{n}$, a space-$s$ Turing machine $\TurM$ has $2^{s'}$ possible configurations, where $s' = s'(s,n) = s + O(\log s) + \log n$. Each configuration can be described by $s'$ bits. Here, $s$ measures the space used by the universal Turing machine $U$ that simulates $\TurM$ on input $x$. In more detail, it can be described by the content of $U$'s work tape, $U$'s current state, and the location of $U$'s heads, including the head on the input/program tape. (Note that a configuration does not include the content of the output tape, which does not affect the next step of the machine.)

We will need the following fact for determining the relationship between configurations of a Turing machine. Recall that a sequence $\{D_n\}_{n \geq 1}$ of size-$T(n)$ computational devices is \emph{logspace-uniform} if there is a machine $M(1^n)$ that runs in space $O(\log T(n))$ and outputs $D_n$ (or equivalently, decides the direct connection language of $D_n$).  

\newcommand{\algoNext}{\mathbb{A}_{\sf nxt}}

\begin{fact}\label{fact:check-next-configuration}
	Given a description of Turing machine $\TurM \in \bs{*}$, a space bound $s \in \N$, an input $x \in \bs{n}$, and two configurations $\gamma,\gamma' \in \bs{s'}$, there is an algorithm $\algoNext$ that determines whether $\gamma'$ is the next configuration obtained by running $\TurM$ for one step on input $x$. Moreover, $\algoNext$ can be computed by a logspace-uniform $O(m^3)$-size $O(\log m)$-depth formula and by an $O(m)$-space algorithm, where $m := |\TurM| + \log s + n + 2s'$ is the total number of input bits. \emph{(}Here, we assume that if $\gamma$ is the accepting state or the rejecting state, then the next configuration of $\gamma$ is always $\gamma$ itself.\emph{)}
\end{fact}

\subsection{Circuits Generated by Bounded-Space Turing Machines}\label{sec:circuit_uniformity}

\newcommand{\nin}{n_{\sf in}}
\newcommand{\nout}{n_{\sf out}}
\newcommand{\Ckt}{\mathsf{Circuit}}

In this paper we often use the following two representations of a circuit (recall that throughout this paper all circuits consist entirely of fan-in two $\mathtt{NAND}$ gates).

\begin{itemize}
	\item ({\bf Adjacency relation tensor.}) A circuit $C$ of size $T$ is given as a tensor ${\sf T}_{C} \in \bs{T \times T \times T}$ such that for every tuple $(u,v,w) \in [T]^3$, ${\sf T}_C(u,v,w) = 1$ if and only if the gates in $C$ indexed by $v$ and by $w$ feed into the gate in $C$ indexed by $u$.
	
	\item ({\bf Layered adjacency relation tensor.}) A circuit $C$ of width $T$ and depth $d$ is given as a list of $d$ tensors ${\sf T}_{C}^{(i)} \in \bs{T \times T \times T}$, where $i \in [d]$, such that for every  layer $i \in [d]$ and tuple $(u,v,w) \in [T]^3$, ${\sf T}^{(i)}_C(u,v,w) = 1$ if and only if the gates in the $(i-1)$-th layer of $C$ indexed by $v$ and by $w$ feed into the gate in the $i$-th layer of $C$ indexed by $u$. 
	
	Here, the input gates are on the $0$-th layer, and the output gates are on the $d$-th layer. Without loss of generality we can assume all layers have exactly $T$ gates.
\end{itemize}
In both cases above, when evaluating $C$ in a context, we will also specify two integers $\nin$ and $\nout$ to denote the number of input/output gates; see the definition of $\Ckt[T,s,\nin,\nout](\TurM)$ given below for details.

While we will mostly use the (unlayered) adjacency relation tensor representation, the layered variant will be very convenient in~\autoref{sec:layered-poly-rep}.

We define next a more general notion of a space-uniform circuit family with input parameters. This will be useful in some situations where we need to compute explicit space bounds for uniformity and index circuits by a tuple of parameters.

\begin{definition} [$\alpha$-Space-Uniform Circuits] \label{def:uniform-ckt}
	Let $k \in \N$ and $\alpha,T \colon \N^k \to \N$. We say that a circuit family with $k$ input parameters $\set{C_{\vec{\ell}}}_{\vec{\ell}\in\N^k}$ of size $T = T(\vec{\ell}\,)$ is {\sf $\alpha$-space-uniform} if there exists an algorithm $A$ such that:
	\begin{enumerate}
		\item ({\bf Decides the adjacency relation.}) The algorithm gets $\vec{\ell} \in \N^k$ and $(u,v,w)\in\bs{3\log(T)}$ as input and accepts if and only if the gates in $C_{\vec{\ell}}$ indexed by $v$ and by $w$ feed into the gate in $C_{\vec{\ell}}$ indexed by $u$. (That is, the algorithm computes the adjacency relation tensor of $C_{\vec{\ell}}$.)
		
		\item ({\bf Runs in $\alpha(\vec{\ell})$ space.}) For input parameters $\vec{\ell} \in \N^k$, the algorithm runs in space $\alpha(\vec{\ell}\,)$.
	\end{enumerate}

	We say $\set{C_{\vec{\ell}}}_{\vec{\ell}\in\N^k}$ is {\sf logspace-uniform} if it is $\mu \log T$-space-uniform for some constant $\mu$.
\end{definition}

\paragraph*{Circuit determined by a Turing machine through the adjacency relation tensor.} We will also consider the circuit determined by a Turing machine in the non-asymptotic setting. More specifically, given a Turing machine $\TurM \in \bs{*}$, parameters $T,s,\nin,\nout \in \N$, we use $\Ckt[T,s,\nin,\nout](\TurM)$ to denote the circuit whose adjacency relation is determined by running $\TurM$ with space bound $s$ over all triples $(u,v,w) \in \bs{3 \log T}$ with $u > v > w$. The first $\nin$ out of $T$ gates are the input gates, and the last $\nout$ out of $T$ gates are the output gates. If $\TurM$ fails to halt on some triples using $s$ bits of space, or the resulting circuit is invalid (\ie, inputs are not source, or outputs are not sink), we let $\Ckt[T,s,\nin,\nout](\TurM) = \bot$.

\newcommand{\algocomp}{\mathbb{A}_{\sf comp}}

Given two circuits $\bfunc{C_1}{n_1}{n_2}$ and $\bfunc{C_2}{n_2}{n_3}$, one can compose them into a single circuit $\bfunc{C_2 \circ C_1}{n_1}{n_3}$ in a natural way (\ie, by identifying the outputs of $C_1$ with the inputs of $C_2$). Suppose $C_1$ is a circuit of size $T_1$ and depth $d_1$, and $C_2$ is a circuit of size $T_2$ and depth $d_2$, then $C_2\circ C_1$ has size $T_1+T_2$ and depth $d_1+d_2$. Also, if $C_1,C_2$ are given by two Turing machines $\TurM_1$ and $\TurM_2$, we  can easily generate another Turing machine $\TurM_{3}$ that specifies $C_2 \circ C_1$. Formally, we will pick a universal machine such that we have the following simple fact on the description length of $\TurM_{3}$, whose proof we omit.

\begin{fact}[Turing Machine Description of Circuit Composition]\label{fact:comp-of-circuits}
	There is a universal constant $c_{\mathsf{comp}} \in \N$ such that the following holds. Given the descriptions of Turing machines $\TurM_1$ and $\TurM_2$, parameters 
	\[
	\vec{\ell}_1 = (T_1,s_1,n_1,n_2),\qquad\vec{\ell}_2 = (T_2,s_2,n_2,n_3) \in \N^4,
	\]
	and letting
	\[
	C_1 = \Ckt[\vec{\ell}_1](\TurM_1),~C_2 = \Ckt[\vec{\ell}_2](\TurM_2),~~\text{and}~~\vec{\ell}_3 = (T_1 + T_2, 2 \cdot (s_1 + s_2) + c_{\mathsf{comp}},n_1,n_3),
	\]
	there is a polynomial-time algorithm $\algocomp$ that given $\TurM_1,\TurM_2,\vec{\ell}_1,\vec{\ell}_2$ as input, outputs the description of a Turing machine $\TurM_3$ such that\,\footnote{We note that if either $C_1 = \bot$ or $C_2 = \bot$, then there is no guarantee on $\algocomp$'s behavior.}
	\[
	(C_2 \circ C_1) = \Ckt[\vec{\ell}_3](\TurM_3)~~\text{and}~~|\TurM_3| \le 2\cdot (|\TurM_1| + |\TurM_2| + \log n_2) + c_{\mathsf{comp}}.
	\]
\end{fact}

\subsection{Pseudorandom Generators and  Hitting Set Generators}

\begin{definition}[Avoiding and Distinguishing]
	Let $m,t \in \N$, $\bfunc{D}{m}{}$, and $Z = (z_{i})_{i \in [t]}$ be a list of strings from $\bs{m}$. Let $\eps \in (0,1)$. We say that $D$ \emph{$\eps$-distinguishes $Z$}, if
	\[
	\left|\Pr_{r \gets \bs{m}}[D(r) = 1] - \Pr_{i \gets [t]}[D(z_i) = 1] \right| \ge \eps.
	\]
	We say that $D$ \emph{$\eps$-avoids $Z$}, if $\Pr_{r \gets \bs{m}}[D(r) = 1] \ge \eps$ and $D(z_i) = 0$ for every $i \in [t]$.
\end{definition}

\section{Polynomial-Time Pseudodeterministic Constructions for Dense Properties}\label{sec:main_proof}

In this section, we prove our main result, restated below for convenience.

\ThmMain*

We will need the following reconstructive hitting set generator, which is a refinement of the generator in~\cite{ChenT21b} with parameters improved using techniques from~\cite{ShaltielU05}. The proof is presented in \autoref{sec:improved-CT}.

\newcommand{\HSGCT}{\HSG^{\sf ct}}
\newcommand{\Recon}{\mathsf{R}}
\newcommand{\ReconCT}{\Recon^{\sf ct}}

\begin{theorem}[Improved Chen--Tell Hitting Set Generator]\label{thm:HSG-Chen-Tell-poly}
	There exists a universal $c \in \posN$, a deterministic algorithm $\HSGCT$, and a probabilistic oracle algorithm $\ReconCT$ such that the following holds. Fix $\kappa,\rho \in \N$. Let $T,d,M,n \in \N$ all be sufficiently large such that $n \le T$, $d \le T$, and $c \cdot \log T \le M \le T^{1/(c\rho)}$. Denote $\vec{\ell} \eqdef (n,T,d,M,\kappa,\rho)$ as the input parameters.

	For a Turing machine $\TurM$ with description size $|\TurM| = \kappa \cdot \log T$, we let
	\[
	C_{\TurM} \eqdef \Ckt[T,\kappa \cdot \log T,n,n](\TurM).
	\]
        Assume the circuit $C_{\TurM} \ne \bot$ and $C_{\TurM}$ has depth at most $d$.
	
	\begin{itemize}
		
		\item [{\bf (Generator.)}] One can turn the $(d\times T)$ computational history of $C_{\TurM}(1^n)$ into a hitting set generator of size $\poly(T^\kappa)$.
        
        \underline{Formally}, the generator $\HSGCT_{\vec{\ell}}$ \emph{(}we write $\HSGCT_{\vec{\ell}}$ to denote that $\HSGCT$ takes $\vec{\ell}$ as input parameters\emph{)} takes the description of a Turing machine $\TurM \in \bs{\kappa \log T}$ as input, and outputs a list of exactly $T^{(c\cdot \kappa) / 2}$ many $M$-bit strings. %
		Let $ \widetilde{T} \eqdef T^{c \cdot \kappa}$ and $ \widetilde{d} \eqdef c \cdot (d \log T + \kappa^2 \log^2 T) + M^c$. There is a Turing machine $\TurM_{\HSG}$ with description length $c\log \WT{T}$ such that for 
		\[
		C_{\HSG} \eqdef \Ckt\mleft[\WT{T},\log \WT{T},n,\mleft(\WT{T}\mright)^{1/2} \cdot M\mright](\TurM_{\HSG}),
		\]
		it holds that \emph{(1)} $C_{\HSG}(1^n) = \HSGCT_{\vec{\ell}}(\TurM)$ and \emph{(2)} $C_{\HSG}$ has depth $\WT{d}$. Moreover, there is a polynomial-time\footnote{In this paper, whenever we say an algorithm $\mathbb{A}$ that generates Turing machines or other succinct descriptions \emph{runs in polynomial time}, we mean the running time is polynomial in the total number of input bits. In this case, the time bound is polynomial in the description length of $\vec{\ell}$ and $\TurM$, \ie, $\poly(\kappa\log T)$.} algorithm $\algoCT$ that on inputs $\vec{\ell}$ and $\TurM \in \bs{\kappa \log T}$, outputs the description of $\TurM_{\HSG}$.

		\item [{\bf (Reconstruction.)}] If our generator fails to hit some dense oracle $D$, then we can compute $C_{\TurM}(1^n)$ in $\poly(d, n, M^\rho)$ time. Crucially, this time bound is \emph{independent} of $T$.

        \underline{Formally}, the reconstruction algorithm $\ReconCT$ takes the description of a Turing machine $\TurM \in \bs{\kappa \log T}$ as input, receives an oracle $D\colon\{0, 1\}^M \to \{0, 1\}$, and satisfies the following:
		
		\begin{itemize}
			\item [{\bf (Soundness.)}] For every oracle $\bfunc{D}{M}{}$, $\left(\ReconCT\right)^D_{\vec{\ell}}(\TurM)$ runs in time $(d+n)\cdot M^{c \rho}$ and with probability at least $1 - 2^{-M}$, its output is either $C_{\TurM}(1^n)$ or $\bot$.
			
			\item [{\bf (Completeness.)}] If $D$ $(1/M^{\rho})$-avoids $\HSGCT_{\vec{\ell}}(\TurM)$, then $\left(\ReconCT\right)^D_{\vec{\ell}}(\TurM)$ outputs $C_{\TurM}(1^n)$ with probability at least $1 - 2^{-M}$.
		\end{itemize}
	\end{itemize}
\end{theorem}

We are now ready to prove~\autoref{theo:intro_main}.

\newcommand{\cCT}{c_{\sf ct}}
\newcommand{\cTM}{c_{\sf tm}}

\begin{proof}[Proof of~\autoref{theo:intro_main}]
	
	We start with some notations.
	
	\paragraph*{Notation.} Let $n_0 \in \N$ be sufficiently large. We define $n_0^{(0)}=n_0$, and for every $\ell \in \posN$,
	\[
	n^{(\ell)}_0 = 2^{2^{n^{(\ell-1)}_0}}.
	\] 
	
	Now, fix $\ell \in \N$. For simplicity of notation, in the following we will use $n_i,\HSG_i,t$ to denote $\nl_i,\Hl_i,t^{(\ell)}$, which will be defined later.
	
	\paragraph*{Construction of hitting sets.} For some parameter $t$ that we set later, we will define a sequence of input lengths $n_1,\dotsc,n_t$, with the hope that we can construct a string in $Q$ pseudodeterministically on at least one of the input lengths. 
	
	Let $\beta \in \posN$ be a sufficiently large constant to be chosen later. For every $i \in [t]$, we set $n_i = (n_{i-1})^\beta$. For each $i \in \{0,\dotsc,t\}$, we will construct a hitting set $\HSG_i \subseteq \{0,1\}^{n_i}$, which is computable by a logspace-uniform $T_i$-size $d_i$-depth circuit. As the base case, we set $\HSG_0$ as the whole set $\{0,1\}^{n_0}$. We note that there is a logspace-uniform $T_{0}$-size $d_0$-depth circuit that outputs all elements in $\HSG_0$, where $T_0 = 2^{2n_0}$ and $d_0 = 2n_0$.
	
	Let $\kappa \in \N$ be a large enough constant to be specified later. Let $c$ be the universal constant from~\autoref{thm:HSG-Chen-Tell-poly}.

	\paragraph*{Informal description.} We will first give a somewhat informal description of the construction of the $\HSG_i$, in particular, we will omit details about the uniformity of the circuits (whose analysis is rather tedious). We hope this can help the reader to gain some intuition first. Later we will carefully analyze the uniformity of the circuits for $\HSG_i$.
	
	For each $i \in [t]$, we construct $\HSG_i$ as follows:
	\begin{enumerate}
		\item We define $\BF_{i-1}$ as the circuit implementing the following algorithm: Enumerate every element in $\HSG_{i-1} \subseteq \{0,1\}^{n_{i-1}}$, and output the first element that is in $Q_{n_{i-1}}$; if no such element exists, then $\BF_{i-1}(n)$ outputs $\bot$;
		
		Using the assumed polynomial-time algorithm $A_Q$ for deciding membership in $Q$, $\BF_{i-1}$ can be implemented by  a $\TBF_{i-1}$-size $\dBF_{i-1}$-depth circuit, where 
		\[
		\TBF_{i-1} = T_{i-1} \cdot \poly(n_{i-1})~~\text{and}~~\dBF_{i-1} = d_{i-1} + \poly(n_{i-1}).
		\]

		\item We then set $\HSG_i$ as the hitting set from~\autoref{thm:HSG-Chen-Tell-poly} constructed with the Turing machine describing the circuit $\BF_{i-1}$ and output length $n_i$.\footnote{We do not discuss how to construct the Turing machine here, the details can be found in the formal construction below.} By~\autoref{thm:HSG-Chen-Tell-poly}, $\HSG_i$ can be implemented by a $T_{i}$-size $d_{i}$-depth circuit, where \[T_i = \poly(\TBF_{i-1})~~\text{and}~~d_i = O(\dBF_{i-1} \cdot \log \TBF_{i-1} + \log^2 \TBF_{i-1}) + \poly(n_i).\] (Here we are being informal, see below for a more precise description.)
	\end{enumerate}
	
	\paragraph*{Formal construction.} Next we carefully detail the construction. Let $\mu \in \posN$ be a large enough constant. First, we define a Turing machine $\TurM_{\HSG_0}$ of description size $\mu$ that describes a $T_0$-size $d_0$-depth circuit $C_{\HSG_0}$ for $\HSG_0$ on input $1^{n_0}$ in $\mu \log T_0$ space. Formally
	\[
	\Ckt[T_0,\mu \cdot \log T_0,n_0, \sqrt{T_0} \cdot n_0 ](\TurM_{\HSG_{0}}) = C_{\HSG_0}.
	\]
	
	Let $\tau \in \N$ be a large enough constant such that the running time of $A_Q$ on $n$-bit inputs is bounded by $n^{\tau /3}$. 
	
	We will make sure all $\HSG_i$ has exactly $\sqrt{T_i}$ elements. (This is satisfied for $i = 0$ since $T_0 = 2^{2 n _0}$.)
	
	Now, for each $i \in [t]$, we will define a Turing machine $\TurM_{\HSG_i}$ such that 
	\[
	\Ckt[T_i,\mu \cdot \log T_i,n_i, \sqrt{T_i} \cdot n_i ](\TurM_{\HSG_{i}}) = C_{\HSG_i},
	\]
	where $C_{\HSG_i}$ has depth at most $d_i$. We will also ensure the invariance that $|\TurM_{\HSG_i}| \le \mu \cdot \log T_i$. By our choice of $\mu$, the above is satisfied when $i = 0$. The machine $\TurM_{\HSG_i}$ is defined in two steps: In the first step we define a machine $\TurM_{\BF_{i-1}}$ describing the circuit $\BF_{i-1}$, and in the second step we plug $\TurM_{\BF_{i-1}}$ in \autoref{thm:HSG-Chen-Tell-poly} to obtain the machine $\TurM_{\HSG_i}$.

	\paragraph*{A Turing machine $\TurM_{\BF_{i-1}}$ for $\BF_{i-1}$.}
	We first define a Turing machine $\TurM_{\BF_{i-1}}$ such that $\TurM_{\BF_{i-1}}(1^{n_{i-1}})$ outputs a circuit for the algorithm $\BF_{i-1}$. Recall that $\BF_{i-1}$ works as follows: Enumerate every element in $\HSG_{i-1} \subseteq \{0,1\}^{n_{i-1}}$ and output the first element that is in $Q_{n_{i-1}}$; if no such element exists, then $\BF_{i-1}(n)$ outputs $\bot$;
	
	Using the assumed polynomial-time algorithm $A_Q$ for deciding membership in $Q$, we first construct a Turing machine $\TurM_{\sf test}$ with description size $\mu$ such that
	\[
		C_{\sf test} = \Ckt\mleft[T_{i-1} \cdot (n_{i-1})^{\tau/2}, \mu \cdot \log T_{i-1}, \sqrt{T_{i-1}} \cdot n_{i-1} ,n_{i-1}\mright](\TurM_{\sf test})
	\]
	has depth $(n_{i-1})^{\tau/2}$, takes a list of $(T_{i-1})^{1/2}$ strings from $\bs{n_{i-1}}$, and outputs the lexicographically first one in $Q_{n_{i-1}}$ (if no such string exists, outputs $\bot$ instead). 
	
	Applying~\autoref{fact:comp-of-circuits} to compose $C_{\HSG_{i-1}}$ and $C_{\sf test}$, we obtain the desired Turing machine $\TM_{\BF_{i-1}}$ that constructs a circuit $C_{\BF_{i-1}}$ computing $\BF_{i-1}$. Noting that $\mu$ is sufficiently large, we have that
	$\TurM_{\BF_{i-1}}$ takes 
	\[
	2 \cdot \left( \descTM{\TurM_{\HSG_{i-1}}} + \mu + \log n_{i-1} + \log T_{i-1} \right) \le 3\mu \cdot \log T_{i-1}
	\]
	bits to describe and uses 
	\[
	2 \cdot \left(\mu \cdot \log T_{i-1} + \mu \cdot \log T_{i-1} + \log T_{i-1}\right) + \mu \le 5 \mu \cdot \log T_{i-1}
	\] space. We now set $T'_{i-1} = T_{i-1} \cdot n_{i-1}^\tau$ and $d'_{i-1} = d_{i-1} + n_{i-1}^{\tau}$, and we have
	\[
	\Ckt\mleft[T'_{i-1},5 \mu \cdot \log T_{i-1},n_{i-1},n_{i-1}\mright](\TurM_{\BF_{i-1}}) = C_{\BF_{i-1}},
	\]
	where $C_{\BF_{i-1}}$ has depth at most $d'_{i-1}$.

	\paragraph*{The Turing machine $\TurM_{\HSG_{i}}$ for $\HSG_i$.}
	
	Recall that $\HSG_i$ is defined as the hitting set $\HSGCT$ of~\autoref{thm:HSG-Chen-Tell-poly} constructed with the circuit $\BF_{i-1}$ and output length $n_i$ in the informal argument. We now formally define $\HSG_i$ as the hitting set
	\[
	\HSGCT_{n_{i-1},T'_{i-1},d'_{i-1},n_i,\kappa,\rho}\mleft(\TurM_{\BF_{i-1}}\mright).
	\]
	
	To apply~\autoref{thm:HSG-Chen-Tell-poly}, we first need to ensure that
	\[
        5 \mu \cdot \log T_{i-1} \le \kappa \log T'_{i-1},
	\]
	which is satisfied by setting $\kappa \ge 5 \mu$. We also need to ensure that
	\begin{equation}\label{eq:cond}
	n_{i-1} \le T'_{i-1},~~d'_{i-1} \le T'_{i-1},~~\text{and}~~c \cdot \log T'_{i-1} \le n_i \le (T'_{i-1})^{1/(c\rho)}.
	\end{equation}

	By~\autoref{thm:HSG-Chen-Tell-poly}, we know that
	\[
	\TurM_{\HSG_i} = \algoCT_{n_{i-1},T'_{i-1},d'_{i-1},n_i,\kappa,\rho}\mleft(\TurM_{\BF_{i-1}}\mright)
	\]
	describes a $T_{i}$-size, $d_{i}$-depth circuit $C_{\HSG_i}$ such that $C_{\HSG_i}(1^{n_{i-1}})$ computes $\HSG_i$. Moreover, $\TurM_{\HSG_i}$ takes $c \cdot \kappa \cdot \log T'_{i-1} \le \mu \cdot \log T_i$ space and $c \cdot \log T_i$ bits to describe, where
	\[
	T_{i} = (T'_{i-1})^{c \cdot \kappa}~~\text{and}~~d_{i} = c \cdot ( d'_{i-1} \log T'_{i-1} + \kappa^2 \cdot \log^2 T'_{i-1}) + n_i^c.
	\]
	Formally, we have
	\[
	C_{\HSG_{i}} = \Ckt[T_i, \mu \cdot \log T_i,n_{i},\sqrt{T_i} \cdot n_i](\TurM_{\HSG_i})
	\]
	as desired. Our invariance on $\descTM{\TurM_{\HSG_{i}}}$ is satisfied by setting $\mu > c$. %
	
	\paragraph*{Analysis of $T_i$ and $d_i$ and justification of~\eqref{eq:cond}.} We set $t$ to be the first integer such that 
	\[
	n_{t+1} > T_t^{1/(c\rho)}.
	\]
	In the following we first show that $t \le \log n_0$.
		
		We first analyze the growth of $T_i$ and $\TBF_{i}$. For every $i < t$, by our choice of $t$, we have that $n_i < n_{i+1} \le T_i^{1/(c\rho)} < T_i$ and hence $\TBF_{i} = T_{i} \cdot n_{i}^{\tau} \le T_{i}^{\tau + 1}$. Then, from $T_{i+1} = (\TBF_{i})^{c \cdot \kappa}$, we have $T_{i+1} \le T_i^{c \cdot (\tau+1) \cdot \kappa}$ and consequently $\log T_{i+1} \le c \cdot (\tau+1) \cdot \kappa \cdot \log T_{i}$. Letting $\lambda = c \cdot (\tau+1) \cdot \kappa$, we have
		\[
		\log T_i \le \lambda^{i} \cdot \log T_{0} = \lambda^i \cdot 2 n_0
		\]
		for every $i \le t$.
		
		Recall that $n_i = n_{i-1}^\beta$, we have $\log n_i = \beta^i \cdot \log n_0$. For $T_t < n_t$ to hold, we only need to ensure the following:
		\begin{align*}
			&\, \lambda^{i} \cdot 2 n_0 < \beta^i \cdot \log n_0\\
			\iff&\, 2 n_0 / \log n_0 < (\beta/\lambda)^i.
		\end{align*}
	
		Now we will set $\beta \ge 100 \lambda$. Let $\bar{t} \le \log n_0$ be the first integer satisfying the above. We claim that $t \le \bar{t}$. Since otherwise $\bar{t} < t$, and we would have $n_{\bar{t}} > T_{\bar{t}}$ (which certainly implies $n_{\bar{t}+1} > T_{\bar{t}}^{1/(c\rho)}$) by our choice of $\bar{t}$. This contradicts our choice of $t$. Therefore, we have established that $t \le \log n_0$.
	
	Now we turn to analyze $d_i$ for $i \le t$. Note that $d_0 = 2 n_0$, and for $i \ge 1$, we have
	\[
	d_i = O\mleft(  (d_{i-1} + n_{i-1}^{\tau}) \cdot \log T'_{i-1} + \log^2 T'_{i-1}  \mright) + n_i^c.
	\]
	
	We will show that for every $i < t$, $d_{i} \le 2 n_{i}^c$. Clearly this holds for $i = 0$.
	
	Since $\log T'_{i-1} \le \log T_{i-1} + O(\log n_{i-1}) \le \lambda^{i-1} \cdot 2n_0 + O(\log n_{i-1}) \le n_{i-1}$ (recall here that $n_{i-1} = (n_0)^{\beta^{i-1}}$ and $\beta = 100 \lambda$), we have 
	\[
	d_i \le O\mleft(  (n_{i-1} + n_{i-1}^{\tau}) \cdot n_{i-1} + n_{i-1}^2  \mright) + n_i^c.
	\]
	
	We can set $\beta$ large enough so that $d_i \le (n_{i-1})^{\beta} + n_i^c \leq 2 \cdot n_i^c$. From definition, we also have $d'_{i} \le 2 n_i^c + n_i^{\tau}$ for every $i < t$. 
	
	Now we are ready to justify the conditions from~\eqref{eq:cond} are satisfied for $i \in [t]$. By our choice of $t$ and the definition of $T'_{i-1}$, we have $n_{i-1} \le T_{i-1} \le T'_{i-1}$. To see $d'_{i-1} \le T'_{i-1}$ holds, recall that $T'_{i-1} = T_{i-1} \cdot n_{i-1}^\tau$, and we have $d'_{i-1} \le 2 n_{i-1}^c + n_{i-1}^\tau \le T_{i-1} \cdot n_{i-1}^{\tau} = T_{i-1}'$ by setting $\tau > c$. We also have that $c\log T'_{i-1} = c(\log T_{i-1} + \tau \log n_{i-1}) = c(\lambda^i \cdot 2n_0 + \tau \log n_{i-1}) < n_i$ since $n_0 < (n_i)^{1/\beta}$ and $\lambda^i \le \log_{n_0}n_i$. Finally, by our choice of $t$, we have $n_{i} \le T_{i-1}^{1/(c\rho)} < \left(T_{i-1}'\right)^{1/(c\rho)}$.
	
	\paragraph*{Informal argument of the correctness.} We first give a somewhat informal argument below, and then give the precise argument later.
	
	We will argue that for every $\ell \in \N$, there exists an $i \in \{0,1,\dotsc, \tl\}$ that our polynomial-time pseudodeterministic algorithm for constructing an element from $Q$ works on input length $\nl_i$.
	
	Let $i \geq 0$ be the largest integer such that $\HSG_i \subseteq \{0,1\}^{n_i}$ is a hitting set of $Q_{n_i}$. (Note that such $i$ exists, since $\HSG_0 = \{0,1\}^{n_0}$ is a hitting set of $Q_{n_0}$.) If $i = t$, then we can simply run $\BF_{t}$ to obtain an element in $Q_{n_t}$ deterministically. Note that this takes time $\poly(T_{t}) = \poly(n_t)$, since by our choice of $t$, $T_t \le n_t^{c \cdot \beta \cdot \rho}$. 
		
	Otherwise, we have $i < t$. In this case, we know that $Q_{n_{i+1}}$ avoids the hitting set $\HSG_{i+1}$ (here we use the fact that $Q_{n_{i+1}}$ accepts more than an $n_{i+1}^{-\rho}$ fraction of strings from $\bs{n_{i+1}}$). By the reconstruction part of~\autoref{thm:HSG-Chen-Tell-poly}, there is a $\poly(n_{i+1}) \cdot \dBF_{i}$ randomized time algorithm that simulates $\BF_{i}$ with probability at least $1 - 2^{n_{i+1}}$. Since $\HSG_i$ is a hitting set for $Q_{n_i}$, this gives us a pseudodeterministic algorithm with $\poly(n_{i+1})$ time that finds a canonical element in $Q_{n_i}$. Since $n_{i+1} = \poly(n_i)$, our pseudodeterministic algorithm runs in polynomial time.

	\paragraph*{Formal description of the algorithm $\bm{B}$.} First, note that by our choice of $t$ and $\beta$, it holds that $n_0^{(\ell + 1)} > n_{t^{(\ell)}}^{(\ell)}$. On an input length $n \in \posN$, our algorithm $B$ is defined as follows:
	
	\begin{enumerate}
		\item Given input $1^n$ for $n \in \posN$.
		
		\item Compute the largest $\ell \in \N$ such that $n_0^{(\ell)} \le n$, then compute the largest $i$ such that $n_i^{(\ell)} \le n$. Output $\bot$ and abort immediately if $n_i^{(\ell)} \ne n$. From now on we use $n_i,T_i,d_i$, etc.~to denote $n_i^{(\ell)},T_i^{(\ell)},d_i^{(\ell)}$, etc. 
		
		\item For every $j \in \{0,1,\dotsc,i\}$, compute $T_j,T'_j,d_j,d'_j,\TurM_{\HSG_{j}},\TurM_{\BF_{j}}$. There are two cases:
		
		\begin{itemize}
			\item Case I: $n_{i+1} \le T_{i}^{1/(c\rho)}$: In this case, we have that $i < t$. Run 
			\[
			\left(\ReconCT\right)_{n_{i},T'_{i},d'_{i},n_{i+1},\kappa,\rho}^{Q_{n_{i+1}}}(\TurM_{\BF_{i}})
			\]
			and set $z_n$ be its output.
			
			\item Case II: $n_{i+1} > T_i^{1/(c\rho)}$: In this case, we have that $t \le i$. Compute $t$ first (recall that $t$ is the first integer such that $n_{t+1} > T_t^{1/(c\rho)}$). Output $\bot$ and abort immediately if $i > t$. Otherwise, construct $C_{\BF_{i}}$ from $\TurM_{\BF_{i}}$ and set $z_n = C_{\BF_{i}}(1^n)$.
		\end{itemize}
	
		\item Output $z_n$ if $A_Q(z_n) = 1$ and $\bot$ otherwise.
		
	\end{enumerate}

	From our choice of parameters and~\autoref{thm:HSG-Chen-Tell-poly}, the algorithm $B$ runs in $\poly(n)$ time.

	\paragraph*{Analysis of the algorithm $\bm{B}$.} Finally we show that the algorithm $B$ satisfies our requirements. We call an input length $n \in \posN$ \emph{valid} if there exist $\ell \in \N$ and $i \in \{0,\dotsc,\tl\}$ such that $n = \nl_i$, and we call $n$ \emph{invalid} otherwise. (By our choice of parameters, such pair $(\ell,i)$ is unique for a valid $n$.) For every $n \in \posN$, let $y_n$ be the lexicographically first element in $Q_n$.
	
	For every invalid $n \in \posN$, we simply set $x_n = y_n$. For every valid $n \in \posN$, we set $x_n$ as follows:
	\begin{align*}
		x_n =
		\begin{cases}
			C_{\BF_{i}}(1^{n_i}), & \text{if } C_{\BF_{i}}(1^{n_{i}}) \in Q_{n_i}, \\
			y_n, & \text{if } \text{otherwise}.
		\end{cases}
	\end{align*}

	We first observe that for all invalid $n \in \posN$, it holds that $B(1^n) = \bot$ with probability $1$. Now we are ready to show that for every $n \in \posN$, $\Pr_{B}[B(1^n) \notin \{ x_n,\bot \}] \le 2^{-n}$. Clearly we only need to consider valid $n$.

	Fix a valid $n \in \posN$. From the soundness of the reconstruction part of~\autoref{thm:HSG-Chen-Tell-poly}, it follows that $z_n \in \{ C_{\BF_{i}}(1^n),\bot \}$ with probability at least $1 - 2^{-n}$ (if $i = t$, then $z_n = C_{\BF_{i}}(1^n)$ with probability $1$). If $C_{\BF_{i}}(1^{n_{i}}) \in Q_{n_i}$, then $x_n = C_{\BF_i}(1^{n_i})$ and $z_n \in \{x_n, \bot\}$ with high probability; otherwise we have $z_n = \bot$. In both cases the soundness of $B$ holds.
	
	Next, we show that for infinitely many $n \in \posN$, we have $\Pr_{B}[B(1^n) = x_n] \ge 1 - 2^{-n}$. Following the informal argument, for every $\ell \in \N$, let $i \geq 0$ be the largest integer such that $\HSG_i \subseteq \{0,1\}^{\nl_i}$ is a hitting set of $Q_{\nl_i}$. Letting $n = \nl_i$, we will show that $B(1^n)$ outputs $x_{n}$ with probability at least $1 - 2^{-n}$, which would finish the proof.
	
	If $i =t$, since $\HSG_i$ is a hitting set for $Q_{n}$, it follows that $z_n = C_{\BF_{i}}(1^n) \in Q_n$, and we have $B(1^n) = x_n$ with probability $1$. If $i < t$, we know that $Q_{n_{i+1}}$ $(1/n_{i+1}^{\rho})$-avoids the hitting set $\HSG_{i+1}$. By the completeness of the reconstruction part of~\autoref{thm:HSG-Chen-Tell-poly}, we have that $z_n = 	\left(\ReconCT\right)_{n_{i},T'_{i},d'_{i},n_{i+1},\kappa,\rho}^{Q_{n_{i+1}}}(\TurM_{\BF_{i}})$ equals $C_{\BF_{i}}(1^n)$ with probability at least $1 - 2^{-n}$. Moreover, in this case, since $\HSG_i$ is a hitting set of $Q_{n}$, we know $z_n \in Q_n$ and $z_n = x_n$, which completes the proof.
\end{proof}

Let $B$ be the algorithm given by \autoref{theo:intro_main}. We note that, by using $1$ bit of advice to encode if a given input length $n$ satisfies $\Pr_B[B(1^n) = x_n] \ge 1 - 2^{-n}$, we can obtain an efficient algorithm that outputs a canonical answer with high probability (\ie, satisfies the promise of a pseudodeterministic algorithm) \emph{on all input lengths} and is correct on infinitely many of them. We state the result below as it might be useful in future work.

\begin{corollary}[Pseudodeterministic Polynomial-Time Construction with $1$ Bit of Advice that Succeeds Infinitely Often]\label{cor:advicebit}
	Let $Q$ be a dense and easy language. There exist a polynomial-time probabilistic algorithm $A$ and a sequence of advice bits $\{\alpha_i\in\{0,1\}\}_{i\in\posN}$ such that
 \begin{itemize}
     \item for all $n\in\posN$, $A(1^n, \alpha_n)$ outputs a canonical $x_n\in\{0,1\}^n$ with probability at least $1-2^{-n}$, and
     \item for infinitely many $n\in\posN$, $x_n\in Q\cap\{0,1\}^n$.
 \end{itemize}
\end{corollary}

\def\desc{\mathsf{desc}}
\def\size{\mathsf{size}}
\def\depth{\mathsf{depth}}

\section{Modified Shaltiel--Umans Generator with Uniform Learning Reconstruction}\label{sec:SU-Gen}

\newcommand{\SU}{\mathsf{SU}}
\newcommand{\HSU}{\mathsf{HSU}}
\newcommand{\RSU}{\mathsf{RSU}}

In order to prove \autoref{thm:HSG-Chen-Tell-poly}, we will need the following result.

\begin{restatable}[A HSG with Uniform Learning Reconstruction]{theorem}{ThmSUMain}\label{theo:SU-main}
    There exist an algorithm $\mathsf{H}$ and a probabilistic oracle algorithm $\mathsf{R}^{(-)}$ such that the following holds.	
    Let $p$ be a nice power of $2$, $m$ be a power of $3$, $\Delta, M\in\mathbb{N}$ with $p > \Delta^2m^7M^9$, and let $\vec{\ell} \eqdef (p,m,M,\Delta)$ be the input parameters.
	\begin{itemize}
		\item The generator $\mathsf{H}_{\vec{\ell}}$ takes as input a polynomial $P \colon \F_{p}^m \to \F_p$ with total degree at most $\Delta$, specified as a list of $p^m$ evaluations of $P$ on all points from $\F_p^m$ in the lexicographic order, and outputs a set of strings in $\{0,1\}^M$. Moreover, $\mathsf{H}_{\vec{\ell}}$ can be implemented by a logspace-uniform circuit of size $\poly(p^m)$ and depth $\poly(\log p, m, M)$. 
		\item The reconstruction algorithm $\mathsf{R}^{D,P}_{\vec{\ell}}$, where $\bfunc{D}{M}{}$ is any function that $(1/M)$-avoids $\mathsf{H}_{\vec{\ell}}(P)$, runs in time $\poly(p,m)$ and outputs, with probability at least $1 - 1/p^m$, a $D$-oracle circuit that computes $P$.
	\end{itemize}
\end{restatable}

The rest of this section is dedicated to the proof of \autoref{theo:SU-main}.

\subsection{Technical Tools}

\subsubsection{Error-Correcting Codes}

\begin{theorem}[{List-Decoding Reed--Solomon Codes \cite{Sudan97}}]\label{thm: list-decoding Reed-Solomon}
    Let $b$, $a$, and $d$ be integers such that $a > \sqrt{2d\cdot b}$. Given $b$ distinct pairs $(x_i, y_i)$ in a field $\F$, there are at most $2\cdot b/a$ polynomials $g$ of degree $d$ such that $g(x_i) = y_i$ for at least $a$ pairs. Furthermore, a list of all such polynomials can be computed in time $\poly(b, \log |\F|)$.
\end{theorem}

In particular, if $a = \alpha \cdot b$ for some $0 < \alpha \leq 1$, provided that $\alpha > \sqrt{2d/b}$, there are at most $O(1/\alpha)$ degree-$d$ polynomials that agree with an $\alpha$-fraction of the $b$ pairs.

\subsubsection{Generator Matrices}
\begin{definition}[Generator Matrices]
	Let $p$ be a power of $2$ and $m\in\mathbb{N}$. We say that $A\in\F_p^{m\times m}$ is a \emph{generator matrix} for $\F_p^m$ if $A$ is invertible, $A^{p^m-1}=I$, and $\{A^{i}\cdot \vec{v}\}_{1\leq i< p^m}=\F_p^m\setminus\{\vec{0}\}$ for any nonzero $\vec{v}\in\F_p^m$.\footnote{In fact, it is not hard to see that the third condition implies the first two. We include those two conditions in this definition as they will be useful later.}
\end{definition}

\begin{theorem}[\cite{Shou92}]\label{t:Sho92}
	Let $n\in\mathbb{N}$. Given any irreducible polynomial $f$ of degree $n$ over $\F_2$, one can deterministically construct in time $\poly(n)$ a set $S_{n}$ that contains at least one primitive root of the multiplicative group of $\mathbb{F}_{2}[\fx]/(f)$.
\end{theorem}

We need the following lemma to deterministically construct generator matrices. Note that it is unclear how to deterministically construct a single generator matrix. Instead, we reduce the task of constructing such matrices to the task of constructing primitive roots of $\F_{p^m}$. Then, we invoke \autoref{t:Sho92} to construct a \emph{set} of matrices that contains at least one generator matrix. It turns out that this set of matrices suffices for our purposes.

\begin{lemma}\label{t:Sho92-matrix}
	Let $p$ be a nice power of $2$ and $m$ be a power of $3$. One can deterministically construct in time $\poly(\log p,m)$ a set of matrices in $\F_p^{m\times m}$ that contains at least one generator matrix for $\F_p^m$.
\end{lemma}
\begin{proof}[Proof Sketch]
    Roughly speaking, every primitive root of $\F_{p^m}$ corresponds to a generator matrix for $\F^m_p$, so the lemma is implied by \autoref{t:Sho92}.

    First, if we let $p=2^{2\cdot 3^{\alpha}}$ and $m=3^{\beta}$, where $\alpha,\beta\in\mathbb{N}$, then the fields $\F_p$ and $\F_{p^m}$ have explicit representations
	\[
	\F_{p^m} = \frac{\F_2[\fx]}{\mleft(\fx^{2\cdot 3^{\alpha+\beta}}+\fx^{3^{\alpha+\beta}}+1\mright)}\text{ \quad and\quad } 	\F_{p} = \frac{\F_2[\fy]}{\mleft(\fy^{2\cdot 3^{\alpha}}+\fy^{3^{\alpha}}+1\mright)}.
	\]

    Note that $\F_{p^m}$ can be viewed as a linear space over $\F_p$ of dimension $m$ (i.e., $\F_p^m$) by identifying $\fx^{3^\beta}$ with $\fy$. The (field) addition operation over $\F_{p^m}$ coincides with the (linear space) addition operation over $\F_p^m$. For every element $g\in\F_{p^m}$, multiplication by $g$ corresponds to a linear transformation over $\F_p^m$, i.e., there is a matrix $A_g$ computable in polynomial time given $g$ such that for every $a\in \F_{p^m}$, $ga$ (as the product of two elements in the field $\F_{p^m}$) is equal to $A_g\cdot a$ (as a matrix-vector product over the vector space $\F_p^m$). It is easy to see that if $g$ is a primitive root of the multiplication group $\F_{p^m}$, then $A_g$ is a generator matrix for $\F_p^m$. The lemma now follows from \autoref{t:Sho92}.
\end{proof}

\subsubsection{Random Self-Reducibility for Discrete Logarithm}

\begin{lemma}\label{l:discrete-log-rsr}
	There is a probabilistic polynomial-time oracle algorithm $\mathsf{DLCorr}^{(-)}$ such that the following holds. Let $p$ be a power of $2$, $m\in \mathbb{N}$, $\varepsilon>0$, $A$ be a generator matrix for $\F_p^m$, and let $g$ be any probabilistic procedure that satisfies 
	\[
		\Pr_{\vec{v}, g} \mleft[\textnormal{$g(\vec{v})$ outputs  $i\in [p^m-1]$ such that $A^i\cdot \vec{1}=\vec{v}$} \mright]\geq \varepsilon, 
	\]
    where the probability is over both a random $\vec{v}\gets \F_p^m\setminus \{\vec{0}\}$ and the internal randomness of $g$. Then for every $\vec{u}\in\F_{p}^m \setminus \{\vec{0}\}$, $\mathsf{DLCorr}^{g}(p,m,1^{\lceil 1/\varepsilon\rceil},A,\vec{u})$ outputs $\ell\in [p^m-1]$ such that $A^\ell\cdot \vec{1}=\vec{u}$ with probability at least $2/3$.
\end{lemma}
\begin{proof}[Proof Sketch]
	The algorithm is an adaptation of the worst-case to average-case reduction for the discrete logarithm problem. Given $\vec{u} \in \F_p^m\setminus\{\vec{0}\}$, we pick a random $j \in [p^m - 1]$ and set $\vec{v} \eqdef A^j \cdot \vec{u}$. Let $i \eqdef g(\vec{v})$. Since $\vec{v}$ is uniformly distributed, with probability at least $\varepsilon$ we have $A^i \cdot \vec{1} = \vec{v}$. We check if this is the case in polynomial time (note that we can compute $A^i$ in polynomial time by repeated squaring). Suppose this is indeed the case, then $A^i \cdot \vec{1} = \vec{v} = A^j \cdot \vec{u}$. Recall that $A$ is invertible. If $i > j$, we output $\ell \eqdef i - j$. If $i = j$, we have $\vec{u} = \vec{1}$. In this case, we output $\ell \eqdef p^m - 1$. Finally, if $j > i$, we output $\ell \eqdef t - (j - i)$.
	
	By sampling $O(1/\varepsilon)$ many values of $j$, with probability at least $2/3$, there is at least one invocation $i \eqdef g(\vec{v})$ such that $A^i\cdot \vec{1} = \vec{v}$ indeed holds. Therefore, the success probability of our algorithm is at least $2/3$.
\end{proof}

\subsubsection{Pseudorandom Generators from One-Way Permutations}
\begin{theorem}[\cite{BlumM84,Yao82,GoldreichL89}]\label{t:GL}
	There exist a deterministic oracle algorithm $\mathsf{CryptoG}^{(-)}$ and a probabilistic oracle algorithm $\mathsf{Invert}^{(-)}$ such that the following holds. Let $s,M\in\mathbb{N}$ be the input parameters, and let $f\colon\{0,1\}^s\to\{0,1\}^s$ be a permutation.
	\begin{enumerate}
		\item \label{i:GL-Gen} $\mathsf{CryptoG}_{s,M}^{f}$ outputs a set of $2^{2s}$ $M$-bit strings. Moreover, $\mathsf{CryptoG}^{f}_{s,M}$ can be implemented by a logspace-uniform $f$-oracle circuit of size $\poly(2^{s},M)$ and depth $\poly(s, M)$.
  
		\item \label{i:GL-Rec} $\mathsf{Invert}_{s,M}^{(-)}$ takes $x\in\{0,1\}^s$ as input and runs in $\poly(s,M)$ time. For any function $D\colon\{0,1\}^M\to\{0,1\}$ that $\varepsilon$-distinguishes $\mathsf{CryptoG}^{f}_{s,M}$ from $\{0,1\}^M$, we have
		\[
		\Pr_{x\gets\{0,1\}^s}\!\left[\mathsf{Invert}_{s,M}^{f,D}\!\left(x\right)=f^{-1}(x)\right]\geq \frac{\varepsilon}{\poly(M)}.
		\]
	\end{enumerate}
\end{theorem}
\begin{proof}[Proof Sketch]
	The generator $\mathsf{CryptoG}^{(-)}$ follows from the well-known construction of pseudorandom generators from one-way permutations using the Goldreich--Levin Theorem \cite{GoldreichL89}. More specifically, 
	\[
		\mathsf{CryptoG}^{f}_{s,M} \eqdef \bigcup_{x,r\in \{0,1\}^s} \mleft( \mleft\langle x ,r \mright\rangle, \mleft\langle f(x) ,r \mright\rangle, \mleft\langle f(f(x)) ,r \mright\rangle,\dots, \mleft\langle f^{(M-1)}(x) ,r \mright\rangle\mright),
	\]
	where $\langle \cdot \rangle$ denotes the inner product mod $2$ function and $f^{(i)}$ denotes the composition of $f$ with itself $i$ times.
	
	The ``inverting'' algorithm $\mathsf{Invert}^{(-)}$ and its correctness rely on standard techniques in pseudorandomness such as the hybrid argument, Yao's theorem on the equivalence between pseudorandomness and unpredictability \cite{Yao82}, and the Goldreich--Levin decoding algorithm \cite{GoldreichL89}. (See \eg, \cite[Section 9.3]{AB09-book}.)
	
	Finally, to see that $\mathsf{CryptoG}^{f}_{s,M}$ can be implemented by a logspace-uniform $f$-oracle circuit of size $\poly(2^{s},M)$ and depth $\poly(M)$, we first note that there is a Turing machine that given $s,M\in\mathbb{N}$ and $x,r\in \{0,1\}^s$, computes the $M$-bit string $\mleft\langle x ,r \rangle, \langle f(x) ,r \mright\rangle, \mleft\langle f(f(x)) ,r \mright\rangle,\dots, \mleft\langle f^{M-1}(x) ,r \mright\rangle$ in $\poly(s,M)$ time using $f$ as an oracle. Then by the fact that any time-$t$ Turing machine can be simulated by a logspace-uniform circuit of size $O(t^2)$, computing a single $M$-bit string in $\mathsf{CryptoG}^{f}_{s,M}$ can be done using a logspace-uniform $f$-oracle circuit of size $\poly(s,M)$. The desired conclusion follows from the observation that we can compute these $2^{2s}$ $M$-bit strings in parallel.
\end{proof}

\subsubsection{Self-Correction for Polynomials}

\begin{theorem}[A Self-Corrector for Polynomials, cf.~\cite{DBLP:journals/ipl/GemmellS92, Sud95}]\label{t:SC}
	There is a probabilistic oracle algorithm $\mathsf{PCorr}^{(-)}$ such that the following holds. Let $p$ be a power of $2$, $m,\Delta\in\mathbb{N}$ such that $\Delta<p/3$. Let $g\colon \F_{p}^m\to \F_p$ be such that there exists a polynomial $P$ of total degree at most $\Delta$ for which 
	\[
	\Pr_{\vec{x}\gets{\F_p^m}}[g(\vec{x}) \ne P(\vec{x})] \le 1/4.
	\]
	Then for all $\vec{x}\in\F_p^m$, $\mathsf{PCorr}^{g}(p,m,\Delta,\vec{x})$ runs in time $\poly(\Delta,\log p,m)$ and outputs $P(\vec{x})$ with probability at least $2/3$.
\end{theorem}

\subsection{The Shaltiel--Umans Generator}

We state a version of the hitting set generator constructed by Shaltiel and Umans \cite{ShaltielU05} that will be convenient for our purposes. 

\begin{theorem}[Implicit in \cite{ShaltielU05}]\label{t:HSU}
    There exist a deterministic algorithm $\HSU$ and a probabilistic oracle algorithm $\RSU^{(-)}$ such that the following holds.
    Let $p$ be a power of $2$, $m,M,\Delta\in\mathbb{N}$ with $p > \Delta^2 m^7M^9$, $\vec{\ell} \eqdef (p,m,M,\Delta)$ be the input parameters, and let
	\begin{itemize}
		\item $P \colon \F_{p}^m \to \F_p$ be a polynomial with total degree at most $\Delta$, given as a list of $p^m$ evaluations of $P$ on all points from $\F_p^m$ in lexicographic order, and
		\item $A$ be a generator matrix for $\F_p^m$.
	\end{itemize} 
	Then
	\begin{enumerate}
		\item\label{i:SU-Gen} The generator $\HSU_{\vec{\ell}}(P,A)$ outputs a set of strings in $\{0,1\}^M$. Moreover, $\HSU_{\vec{\ell}}$ can be implemented by a logspace-uniform circuit of size $\poly(p^m)$ and depth $\poly(\log p,m)$.
		\item\label{i:SU-Rec}  The reconstruction algorithm $\RSU^{D, P}_{\vec{\ell}}(A)$, where $\bfunc{D}{M}{}$ is any function that $(1/M)$-avoids $\HSU_{\vec{\ell}}(P,A)$, runs in $\poly(p,m)$ time and outputs, with probability at least $1 - 1/p^{2m}$, a vector $\vec{v} \in \F_p^m\setminus \{\vec{0}\}$ and a $D$-oracle circuit $C:[p^m - 1] \to \F_p$ such that
		\[
			C(i) = P(A^{i}\cdot \vec{v})\text{ for every }i\in[p^m-1].
		\]
	\end{enumerate}
\end{theorem}

The statement of \autoref{t:HSU} and the HSG result of \cite{ShaltielU05} differ in two aspects:\begin{itemize}
    \item First, we use a \emph{polynomial} instead of a Boolean function to construct the HSG, which fits more naturally into the framework of Chen--Tell \cite{ChenT21b} (see also \autoref{sec:improved-CT}).
    \item Second, we explicitly calculated a circuit depth upper bound for computing the HSG, which is not stated in \cite{ShaltielU05}.
\end{itemize}
Nevertheless, \autoref{t:HSU} easily follows from the arguments in \cite{ShaltielU05}. For completeness, we review the construction of \cite{ShaltielU05} and present a proof sketch of \autoref{t:HSU} in this subsection.

\def\Gary{G_{p\textsf{-}\mathsf{ary}}}
\paragraph{The generator.} We first construct $m$ candidate ``$p$-ary PRGs'' $\Gary^{(0)}, \Gary^{(1)}, \cdots, \Gary^{(m-1)}: \F_p^m \to \F_p^M$; note that the inputs and outputs of these ``$p$-ary PRGs'' are elements in $\F_p$. In particular:
\[\Gary^{(j)}(\vec{x}) = \mleft(P(A^{p^j\cdot 1}\vec{x}), P(A^{p^j\cdot 2}\vec{x}), \cdots, P(A^{p^j\cdot M}\vec{x})\mright).\]

Then we convert each $p$-ary PRG into a (usual binary) PRG by invoking \cite[Lemma 5.6]{ShaltielU05}. More precisely, for each $0\le j < m$, we interpret $\Gary^{(j)}$ as a PRG that takes a binary seed of length $m\log p$ and outputs $M$ elements in $\{0, 1\}^{\log p}$, using the canonical bijection $\kappa^{(\log p)}$ between $\F_p$ and $\{0, 1\}^{\log p}$. Then, for $\Gary^{(j)}:\{0, 1\}^{m\log p} \to (\{0, 1\}^{\log p})^M$, given seeds $x \in \{0, 1\}^{m\log p}$ and $r \in \{0, 1\}^{\log p}$, we define
\[G^{(j)}(x, r) = \mleft(\langle \Gary^{(j)}(x)_1, r\rangle, \langle \Gary^{(j)}(x)_2, r\rangle, \dots, \langle \Gary^{(j)}(x)_M, r\rangle\mright).\]
Here, $\langle \cdot \rangle$ denotes the inner product mod $2$ function. In other words, we combine $\Gary^{(j)}$ with the \emph{Hadamard code} to obtain a Boolean PRG $G^{(j)} \colon \{0,1\}^{m\log p + \log p} \to \{0,1\}^M$.

Finally, our HSG will be the union of all PRGs $G^{(j)}$. That is, our algorithm $\HSU_{\vec{\ell}}(P, A)$ enumerates every $0\le j < m$, $x \in \{0, 1\}^{m\log p}$, $r \in \{0, 1\}^{\log p}$, and prints the string $G^{(j)}(x, r)$.

To see that $\HSU_{\vec{\ell}}$ can be computed by a logspace-uniform low-depth circuit, we argue that given appropriate indexes $j$ and $i$, the $i$-th bit of $G^{(j)}(x, r)$ can be computed by a logspace-uniform low-depth circuit. The bit we want to compute is
\[G^{(j)}(x, r)_i = \langle \Gary^{(j)}(x)_i, r\rangle = \langle P(A^{p^j\cdot i}\vec{x}), r\rangle,\]
where $\vec{x}$ is the vector in $\F_p^m$ encoded by $x$. By repeated squaring, we can output a (logspace-uniform) circuit of size and depth $\poly(\log p, m)$ that computes $A^{p^j\cdot i}$. Multiplying $A^{p^j\cdot i}$ with $\vec{x}$, indexing (\ie, finding the ($A^{p^j\cdot i}\vec{x}$)-th entry of $P$), and computing Boolean inner product have logspace-uniform circuits of size $\poly(M, p^m) = \poly(p^m)$ and depth $\poly(m, \log p, \log M) = \poly(\log p, m)$. Since we need to output $m\cdot p^{m+1}$ strings of length $M$ and each output bit can be computed by a logspace-uniform circuit of size $\poly(p^m)$ and depth $\poly(\log p, m)$, the complexity upper bounds for computing $\HSU_{\vec{\ell}}$ follows.\\

Now we consider the reconstruction algorithm. Suppose there is an adversary $D:\{0, 1\}^M \to \{0, 1\}$ that $(1/M)$-avoids $\HSU_{\vec{\ell}}(P, A)$. It follows that $D$ $(1/M)$-distinguishes every binary PRG $G^{(j)}$.

\def\Dbit{D_{\sf bit}}
\paragraph{From distinguishers to next-element predictors.} For each $0\le j<m$, we use $D$ to build a ``next-element predictor'' $D^{(j)}$ for $\Gary^{(j)}$. Since $D$ $(1/M)$-distinguishes $G^{(j)}$, it can be used to build a next-\emph{bit} predictor $\Dbit^{(j)}$ such that
\[\Pr_{i\gets[M], x\gets \{0, 1\}^{m\log p}, r\gets \{0, 1\}^{\log p}}\mleft[\Dbit^{(j)}\mleft(G^{(j)}(x, r)_1, \dots, G^{(j)}(x, r)_{i-1}\mright) = G^{(j)}(x, r)_i\mright] \ge 1/2 + 1/M^2.\]
Therefore, with probability $\ge 1/2M^2$ over $i\gets [M]$ and $x\gets \{0, 1\}^{m\log p}$, the probability over $r\gets \{0, 1\}^{\log p}$ that $\Dbit^{(j)}$ predicts the $i$-th bit of $G^{(j)}(x, r)$ given its first $i-1$ bits correctly is at least $1/2 + 1/2M^2$. In this case, using the list-decoding algorithm for Hadamard code \cite{GoldreichL89}, we can find a list of $O(M^4)$ elements that contains $\Gary^{(j)}(x)_i$. (In fact, the trivial list-decoding algorithm suffices here, since it runs in time $\poly(p)$.) We call this procedure the \emph{next-element predictor $D^{(j)}$}; it takes as input
\[u_{M-1} \eqdef P(A^{-(M-1)p^j}\vec{x}), u_{M-2} \eqdef P(A^{-(M-2)p^j}\vec{x}), \dots, u_1 \eqdef P(A^{-p^j}\vec{x}),\]
where $\vec{x}\gets \F_p^m$ is a random vector. It randomly selects $i\gets [M]$, invokes $\Dbit^{(j)}$ and the list-decoding algorithm for the Hadamard code, and outputs a list of $O(M^4)$ elements. With probability $\Omega(1/M^2)$ over $\vec{x}\gets \F_p^m$ and the internal randomness of $\Dbit^{(j)}$, this list will contain $P(\vec{x})$.

We repeat $D^{(j)}$ for $O(m\log p)$ times and fix its internal randomness, so that in what follows we can assume $D^{(j)}$ is deterministic. With probability at least $1-1/(10p^{2m})$, for every $0\le j < m$, $D^{(j)}$ will be correct in the following sense: For some $\rho \eqdef 1/\Theta(M^2m\log p)$, $D^{(j)}$ outputs $\rho^{-2}$ elements, and
\begin{equation*}%
    \Pr_{\vec{x}\gets \F_p^m}\mleft[P(\vec{x}) \in D^{(j)}\mleft(u_{M-1}, u_{M-2}, \dots, u_1\mright)\mright] > \rho.
\end{equation*}

\def\LearnNextCurve{\textsc{Learn Next Curve}}

\paragraph{Learn Next Curve.} We will use the following notation from \cite{ShaltielU05}. Let $r \eqdef O(m\log p)$ be a parameter denoting the number of {\sf reference points}, and $v \eqdef (m+1)r-1$ denotes the degree of curves.\footnote{The parameter $v$ is set in the proof of \cite[Lemma 5.14]{ShaltielU05}.} A \emph{curve} is a polynomial $C:\F_p \to \F_p^m$ with degree $v$. (That is, each coordinate of $C$ is a univariate polynomial of degree $v$ over $\F_p$.) Recall that $A \in \F_p^{m\times m}$ is the generator matrix. We use $AC$ to denote the curve where for each $t\in\F_p$, $AC(t) = A\cdot C(t)$ (note that $AC$ is still a degree-$v$ polynomial). We also use $P(C)$ to denote the function such that for every $t\in\F_p$, $P(C)(t) = P(C(t))$; the \emph{evaluation table} of $P(C)$ is the length-$p$ vector where for every $t\in\F_p$, the $t$-th entry of the vector is $P(C(t))$.

Now, we recall the implementation of an important subroutine called $\LearnNextCurve$ as defined in \cite[Section 5.5]{ShaltielU05}. $\LearnNextCurve$ takes as input a {\sf next curve} $C:\F_p \to \F_p^m$, a set of {\sf reference points} $R\subseteq \F_p$ of size $r$, a {\sf stride} $0\le j < m$, as well as {\sf input evaluations}; the {\sf input evaluations} consist of two parts, namely the evaluation tables of $P(A^{-ip^j}C)$ for every $1\le i < M$ and the values of $P(C(t))$ for every $t \in R$. The intended {\sf output evaluations} consist of the evaluation table of $P(C)$.

In particular, $\LearnNextCurve$ starts by obtaining a set of $\rho^{-2}$ values
\[S_t \eqdef D^{(j)}\mleft(P(A^{-(M-1)p^j}C(t)), P(A^{-(M-2)p^j}C(t)), \dots, P(A^{-p^j}C(t))\mright)\]
for each $t \in \F_p$. Then it calls the algorithm from \autoref{thm: list-decoding Reed-Solomon} on the pairs $\{(t, e)\}_{t \in \F_p, e \in S_t}$ to obtain the list of all polynomials $Q$ such that $Q(t) \in S_t$ for many coordinates $t$. (This takes $\poly(p\rho^{-2}, \log p) \le \poly(p, m)$ time.) If this list contains a unique polynomial $Q$ such that $Q(t) = P(C(t))$ for every $t \in R$, then we output this polynomial; otherwise we output $\bot$. It is clear that $\LearnNextCurve$ runs in $\poly(p, m)$ time.

We say $\LearnNextCurve$ \emph{succeeds} (on {\sf next curve}, {\sf reference points}, and {\sf stride}), if whenever the {\sf input evaluations} are the intended values, the {\sf output evaluations} are also the intended values. Let
\[\eps_{\rm LNC} \eqdef O(v\rho^{-1}/p)^{v/2} + (8\rho^{-3})(v\deg(P)/p)^r.\]
It is proven in \cite[Lemma 5.12]{ShaltielU05} that, assuming $p > 32\deg(P)v/\rho^4$, if the {\sf next curve} and {\sf reference points} are chosen uniformly at random, $\LearnNextCurve$ succeeds with probability $1-\eps_{\rm LNC}$. Since $\deg(P) = \Delta$, $\rho^{-1} = \Theta(M^2m\log p)$, $v = O(m^2\log p)$, and $p > \Delta^2 m^7M^9$, it is indeed the case that $p > 32\deg(P)v/\rho^4$. Also note that
\[\eps_{\rm LNC} \le O(\rho^3/32\deg(P))^{v/2} + (8\rho^{-3})(\rho^4/32)^r \le (1/2)^{r-1} \ll 1/(10p^{4m}).\]

A first attempt for the reconstruction algorithm would be the following. Let $i \in [p^m - 1]$, and suppose that we want to compute $P(A^i\vec{1})$. We write $i$ in $p$-ary as $i = \sum_{j=0}^{m-1}i_j p^j$ (where each $i_j \in \{0, 1, \dots, p-1\}$). Recall that for each {\sf next curve} $C$ and {\sf stride} $j$, given the evaluation tables of $P(A^{-kp^j}C)$ for every $1\le k < M$, we can learn the evaluation table of $P(C)$ in one invocation of $\LearnNextCurve$. Therefore, we proceed in $m$ rounds, where for each $0\le l<m$, the $l$-th round performs the following computation:\begin{itemize}
    \item Let $i' \eqdef \sum_{j=0}^{l-1}i_jp^j$. Suppose that at the beginning of the $l$-th round, we already know the evaluation tables of $P(A^{kp^l + i'}C)$ for each $1\le k < M$. (For $l = 0$, these values can be hardcoded as advice; for $l\ge 1$, they should be obtained from the previous round.) We invoke $\LearnNextCurve$ $M(p-1)$ times with stride $l$ to obtain the evaluation tables of $P(A^{kp^l + i'}C)$ for each $1\le k < M\cdot p$. The $l$-th round ends here; note that we have obtained the evaluation tables required in the $(l+1)$-th round (namely $P(A^{kp^{l+1} + i_lp^l + i'}C)$ for every $1\le k < M$).
\end{itemize}

However, there is one issue with this approach: to learn a curve $C$, we also need to provide $\LearnNextCurve$ with the evaluations of some {\sf reference points} on $C$. To resolve this issue, \cite{ShaltielU05} introduced an \emph{interleaved learning} procedure that involves two curves $C_1$ and $C_2$. These two curves possess nice intersection properties that for certain choices of $k$ and $l$, $A^kC_1$ and $A^lC_2$ intersect on at least $r$ points. This enables us for example to learn the evaluation table of $P(A^lC_2)$ whenever we know the evaluation table of $P(A^kC_1)$, by using the evaluations of $P(A^kC_1)$ at {\sf reference points} $R$, where $R$ is the intersection of $A^kC_1$ and $A^lC_2$.

\paragraph{Interleaved learning.} In what follows, we use $[C_1 \cap C_2]$ to denote the set $\{t \in \F_p: C_1(t) = C_2(t)\}$. We say two curves $C_1$ and $C_2$ are \emph{good} if they satisfy the following properties:\begin{itemize}
    \item $C_1(1)\ne \vec{0}$;
    \item for all $1\le i < p^m$ and all $0 \le j < m$, $[A^{i+p^j}C_1 \cap A^iC_2]$ and $[A^iC_1 \cap A^iC_2]$ are of size $\ge r$;
    \item for all $1\le i < p^m$ and all $0\le j < m$, $\LearnNextCurve$ succeeds given {\sf next curve} $A^{i+p^j}C_1$, {\sf reference points} $[A^{i+p^j}C_1 \cap A^iC_2]$, and {\sf stride} $j$; and
    \item for all $1\le i < p^m$ and all $0\le j < m$, $\LearnNextCurve$ succeeds given {\sf next curve} $A^iC_2$, {\sf reference points} $[A^iC_1 \cap A^iC_2]$, and {\sf stride} $j$.
\end{itemize}

By \cite[Lemma 5.14]{ShaltielU05}, there is a $\poly(v, p)$-time randomized algorithm that, with probability $1-2mp^m\cdot \eps_{\rm LNC} \ge 1-1/(10p^{2m})$, outputs two curves $C_1$ and $C_2$ that are good.

The basic step in the reconstruction algorithm is called \emph{interleaved learning} in \cite{ShaltielU05}. This step has the following guarantee: For a {\sf stride} $j$, given the correct evaluation tables of $P(A^{i-kp^j}C_1)$ and $P(A^{i-kp^j}C_2)$ for every $1\le k < M$, we can compute the correct evaluation tables of $P(A^iC_1)$ and $P(A^iC_2)$. In particular, \emph{interleaved learning} consists of the following two steps:\begin{itemize}
    \item first, we invoke $\LearnNextCurve$ with {\sf next curve} $A^iC_1$, {\sf reference points} $[A^{i-p^j}C_2 \cap A^iC_1]$, and {\sf stride} $j$;
    \item then, we invoke $\LearnNextCurve$ with {\sf next curve} $A^iC_2$, {\sf reference points} $[A^iC_1 \cap A^iC_2]$, and {\sf stride} $j$.
\end{itemize}
Note that we assume that all invocations of $\LearnNextCurve$ succeed, as this happens with high probability ($1-1/(10p^{2m})$).

\paragraph{The reconstruction algorithm.} Recall that our reconstruction algorithm needs to output two elements: a vector $\vec{v} \in \F_p^m\setminus\{\vec{0}\}$ and a $D$-oracle circuit $C:[p^m - 1] \to \F_p$ such that $C(i) = P(A^i\cdot \vec{v})$ for every $i\in [p^m-1]$.

We first compute the curves $C_1$ and $C_2$ that are good with probability $1 - 1/(10p^{2m})$. Our reconstruction algorithm will be correct provided that $C_1$ and $C_2$ are good (and that we fixed good internal randomness of our next-element predictors $D^{(j)}$); this happens with probability $\ge 1-1/(10p^{2m})-1/(10p^{2m})\ge 1-1/p^{2m}$. The vector we output will be $\vec{v} \eqdef C_1(1)$ (which is non-zero if $C_1$ and $C_2$ are good). It remains to output a circuit $C$ such that for every $i \in [p^m-1]$, $C(i) = P(A^i\cdot \vec{v})$.

Given an integer $i$, our circuit $C$ first writes $i$ in $p$-ary as $i = \sum_{j=0}^{m-1}i_jp^j$. Then, it proceeds in $m$ rounds, where for each $0\le l<m$, the $l$-th round performs the following:
\begin{itemize}
    \item Let $i' \eqdef \sum_{j=0}^{l-1}i_jp^j$. Suppose that at the beginning of the $l$-th round, we already know the evaluation tables of $P(A^{kp^l + i'}C_1)$ and $P(A^{kp^l + i'}C_2)$ for each $1\le k < M$. We perform interleaved learning $M(p-1)$ times with stride $l$ to obtain the evaluation tables of $P(A^{kp^l + i'}C_1)$ and $P(A^{kp^l + i'}C_2)$ for each $1\le k < M\cdot p$. The $l$-th round ends here; note that we have obtained the evaluation tables required to perform the $(l+1)$-th round (namely, $P(A^{kp^{l+1} + i_lp^l + i'}C_1)$ and $P(A^{kp^{l+1} + i_lp^l + i'}C_2)$ for every $1\le k < M$).
\end{itemize}
Finally, after the $(m-1)$-th round, we have obtained the evaluation table of $P(A^iC_1)$, and we can simply output $P(A^iC_1(1)) = P(A^i\vec{v})$ as the answer.

Note that the interleaved learning procedure needs to invoke the next-element predictor, therefore our circuit $C$ will be a $D$-oracle circuit. Also, at the beginning of the first ($0$-th) round, we need the evaluation tables of $P(A^kC_1)$ and $P(A^kC_2)$ for each $0\le k < M$. Our reconstruction algorithm can simply query the polynomial $P$ to obtain these values and hardcode them into our circuit $C$. It is clear that our reconstruction algorithm runs in $\poly(p, m)$ time and succeeds with probability $\ge 1-1/p^{2m}$.

\subsection{Modified Shaltiel--Umans Generator: Proof of \texorpdfstring{\autoref{theo:SU-main}}{Theorem 5.1}}

In this subsection, we prove \autoref{theo:SU-main}, which is restated below.
\ThmSUMain*
\begin{proof}
	One difference between our generator and the Shaltiel--Umans generator (\autoref{t:HSU}) is that the reconstruction procedure in the latter only learns a circuit $C_0$ that computes the mapping $i\mapsto P(A^{i}\cdot \vec{v})$ (for some $\vec{v}$ output by the reconstruction procedure), where $A$ is the generator matrix used in the Shaltiel--Umans construction, instead of a circuit that computes $P$ itself. Let us assume for simplicity that the circuit $C_0$ computes $i\mapsto P(A^{i}\cdot \vec{1})$. Note that if given $\vec{x}\in\F_p^m\setminus\{\vec{0}\}$ (which is the input on which we intend to evaluate $P$), we could \emph{efficiently} compute the value $i\in[p^m-1]$ such that $A^i\cdot \vec{1}=\vec{x}$, then we would be able to combine this with the circuit $C_0$ to compute $P$ (roughly speaking, by first computing $i$ and then outputting $C_0(i)$). However, there are two issues with this approach:

    \begin{enumerate}
        \item First, we do not know the generator matrix $A$, as we need our reconstruction algorithm to be uniform and thus we cannot hardcode $A$.
        \item Second, the task of finding such $i$ given $\vec{x}$ and $A$ is essentially the \emph{discrete logarithm problem}, for which no efficient algorithm is known.
    \end{enumerate}
	
	To handle the first issue, we will construct our generator by using the Shaltiel--Umans construction based on a generator matrix that is from a small set $S$ given by \autoref{t:Sho92-matrix}. Then in the reconstruction, we will try all the matrices from $S$, which can be generated efficiently, to obtain a list of candidate circuits. We then select from the list a circuit that is sufficiently close to $P$ and use a self-corrector to compute $P$ everywhere. For the second issue, we first observe that the mapping $f\colon i\mapsto A^{i}\cdot \vec{1}$ is a \emph{permutation}. Treating $f$ as a ``cryptographic one-way permutation'' and invoking \autoref{t:GL}, we can construct a ``cryptographic pseudorandom generator'', which has a uniform reconstruction algorithm. We can then combine the output of this ``cryptographic pseudorandom generator'' with that of the Shaltiel--Umans generator so that if there is an algorithm $D$ that avoids this combined generator, then $D$ can also be used to invert $f$ efficiently! Details follow.

    \paragraph{The construction of $\mathsf{H}$.}
	For a matrix $A\in \F_p^{m\times m}$, let $f_{A}\colon[p^m-1]\cup\{0\}\to\F_p^m$ be such that
	\[	
	f_{A}(i) \eqdef \begin{cases}
	0^n  & \text{if } i=0 \\
	A^{i}\cdot\vec{1} & \text{if } 1 \leq i<p^m.
\end{cases}	
	\]
	We will also view $f$ as a function mapping $s$ bits to $s$ bits, where $s \eqdef m\cdot \log p$. Also note that if $A$ is a generator matrix for $\F_p^m$, then $f_{A}$ is a permutation. 
 
    Let $\HSU$ be the generator from \autoref{t:HSU} and $\mathsf{CryptoG}^{(-)}$ be the generator from \autoref{t:GL}. Also, let $S\subseteq \F_p^{m\times m}$ be the set of matrices constructed using \autoref{t:Sho92-matrix}. We define
	\[
	\mathsf{H}_{\vec{\ell}}(P) \eqdef \bigcup_{A\in S}\, \left(\HSU_{\vec{\ell}}(P,A) \,\mathsmaller{\bigcup}\, \mathsf{CryptoG}^{f_A}_{s,M}\right).
	\]

    \paragraph{The complexity of $\mathsf{H}$.} We argue that $\mathsf{H}_{\vec{\ell}}$ can be implemented by a logspace-uniform circuit of size $\poly(p^m)$ and depth $\poly(\log p, m, M)$. 
 
    First note that given $A$, $f_A$ can be computed in $\poly(\log p,m)$ time. Then again by the fact that every time-$t$ Turing machine can be simulated by a logspace-uniform circuit of size $O(t^2)$, $f_A$ can be computed by a logspace-uniform circuit of size $\poly(\log p,m)$. This means given $A$, $\mathsf{CryptoG}^{f_A}_{s,M}$, which by \autoref{t:GL} has a logspace-uniform $f_A$-oracle circuit of size $\poly(2^s,M)$ and depth $\poly(s, M)$, can be implemented by a logspace-uniform circuit of size $\poly(p^m)$ and depth $\poly(\log p, m, M)$, where we have used that $s = m \cdot \log p$ and $M \leq p^{1/9}$. Also, by \autoref{t:HSU}, $\HSU_{\vec{\ell}}$ has a logspace-uniform circuit of size $\poly(p^m)$ and depth $\poly(\log p,m,M)$. To compute $\mathsf{H}_{\vec{\ell}}(P)$, we just need to compute $\HSU_{\vec{\ell}}(P,A)$ and $\mathsf{CryptoG}^{f_A}_{s,M}$ for all $A\in S$ in parallel, where $S$ can also be computed in time $ \poly(\log p,m)$ and hence has logspace-uniform circuit of size $\poly(\log p,m)$. This yields a logspace-uniform circuit of size $\poly(p^m)$ and depth $\poly(\log p, m, M)$ computing $\mathsf{H}_{\vec{\ell}}$.

    \paragraph{The reconstruction.} Given oracle access to the polynomial $P$ and a function $D$ that $(1/M)$-avoids $\mathsf{H}_{\vec{\ell}}(P)$, we want to output a $D$-circuit that computes $P$. We do this in two stages. In the first stage, we obtain a list of candidate circuits, one for each $A\in S$, that (with high probability) contains at least one circuit that computes $P$. In the second stage, we will select, from the list of candidate circuits, one that is sufficiently close to $P$ and combine it with a self-corrector to obtain a circuit that computes $P$ on all inputs.
	
	We now describe the first stage. Let $A^{*}$ be the lexicographically first matrix in $S$ that is a generator matrix for $\F_p^m$, and consider the two sets 
	\[
	\HSU_{\vec{\ell}}(P,A^*) \text{\quad and \quad} \mathsf{CryptoG}^{f_{A^*}}_{s,M},
	\]
	which are subsets of $\mathsf{H}_{\vec{\ell}} (P)$. Since $D$ avoids $\mathsf{H}_{\vec{\ell}}$, it also avoids \emph{both} $\HSU_{\vec{\ell}}(P,A^*)$ and $\mathsf{CryptoG}^{f_{A^*}}_{s,M}$.
	
	Assume for a moment that we are given the matrix $A^*$. We will construct a circuit $C_{A^*}$ as follows. Let $\RSU^{(-)}$ and $\mathsf{Invert}^{(-)}$ be the oracle algorithms from \autoref{t:HSU} and \autoref{t:GL} respectively. We first run $\RSU^{D,P}_{\vec{\ell}}(A^*)$ to obtain a $D$-oracle circuit $C'_{A^*}$ and some $\vec{v}\in\F_p^m\setminus{\{\vec{0}\}}$. By the property of $\RSU^{(-)}$ (\autoref{i:SU-Rec} of \autoref{t:HSU}) and the fact that $D$ avoids $\HSU_{\vec{\ell}}(P,A^*)$, we get that, with probability at least $1-1/p^{2m}$, for every $i\in [p^m-1]$,
	\begin{equation}\label{eq:circuit_1}
		C'_{A^*}(i) = P((A^*)^{i}\cdot \vec{v}).
	\end{equation}
 
	Similarly, by the property of $\mathsf{Invert}^{(-)}$ (\autoref{i:GL-Rec} of \autoref{t:GL}) and the fact that $D$ avoids $\mathsf{CryptoG}^{f_{A^*}}_{s,M}$, we get that
	\[
		\Pr_{x\gets\{0,1\}^s}\!\left[
		\mathsf{Invert}^{f_{A^*},D}_{s,M}(x)=f_{A^*}^{-1}(x)\right]\geq \frac{1}{\poly(M)}.
	\]
	By combining 
	\[
	g \eqdef\mathsf{Invert}^{f_{A^*},D}_{s,M}
	\]
	 with the algorithm $\mathsf{DLCorr}^{(-)}$ from \autoref{l:discrete-log-rsr}, we get that for every $\vec{x}\in\F_{p}^m$,	with probability at least $2/3$ over the internal randomness of $\mathsf{DLCorr}^{g}$,
	 \[
	 	\mathsf{DLCorr}^{g}\mleft(p,m,1^{\poly(M)},A^{*},\vec{x}\mright)=f^{-1}_{A^*}(\vec{x}).
	 \]

    Note that $\mathsf{DLCorr}^{g}\mleft(p,m,1^{\poly(M)},A^{*},\vec{x}\mright)$ runs in time
\[
\poly(\log p, m, M)\cdot \poly(s, M) = \poly(p, m)
\]
with oracle access to $D$. Also, using standard error-reduction techniques, we can reduce its error probability from $2/3$ to $1/(10p^{2m})$ at the cost of a multiplicative blowup of $O(m \log p)$ in the running time. Then by a union bound, with probability at least $1 - 1/(10p^{m})$ over its internal randomness, the procedure correctly computes $f_{A^*}^{-1}(\vec{x})$ for all $\vec{x}\in\F_p^m$. Finally, by sampling and fixing the internal randomness uniformly at random, we can construct, in time $\poly(p,m)$ and with probability at least $1 - 1/(10p^m)$, a $D$-oracle circuit $C''_{A^*}$ such that for every $\vec{x}\in\F_{p}^m$,
	\begin{equation}\label{eq:circuit_2}
		C''_{A^*}(\vec{x}) =f_{A^*}^{-1}(\vec{x}).
	\end{equation}
	That is, given $\vec{x}\in\F_{p}^m\setminus\{\vec{0}\}$, $C''_{A^*}(\vec{x})$ outputs $i\in[p^m-1]$ such that $(A^*)^{i}\cdot \vec{1}=\vec{x}$.
	This is almost what we need except that we want the circuit to output $i$ such that $(A^*)^{i}\cdot \vec{v}=\vec{x}$. We further construct such a circuit $C'''_{A^*}$ as follow. Given $\vec{x}\in\F_{p}^m$, we first compute
	\[
		j \eqdef C''_{A^*}(\vec{v}) \text{\quad and \quad} k \eqdef C''_{A^*}(\vec{x}).
	\]
	That is, $\vec{v} = (A^*)^{j}\cdot \vec{1}$ and $\vec{x}= (A^*)^{k}\cdot \vec{1}$. We then output $i$ depending on the values of $j$ and $k$. On the one hand, if $j<k$, we let $i \eqdef k-j$. Then
	\begin{align*}
		(A^*)^i\cdot \vec{v} = (A^*)^{k-j}\cdot (A^*)^{j}\cdot \vec{1} = (A^*)^{k}\cdot \vec{1} = \vec{x}.
	\end{align*}
On the other hand, if $k\leq j$, we let $i\eqdef p^m-1-(j-k)$ , which yields
	\[
		(A^*)^i\cdot \vec{v}= (A^*)^{p^m-1-j+k} \cdot (A^*)^{j} \cdot \vec{1}= I\cdot (A^*)^{k}\cdot \vec{1}  =\vec{x}.
	\]
 
    Now we have a circuit $C'''_{A^*}$ that given $\vec{x}\in\F_{p}^m\setminus\{\vec{0}\}$, outputs $i\in[p^m-1]$ such that $(A^*)^{i}\cdot \vec{v}=\vec{x}$ and a circuit $C'_{A^*}$ that given $i\in [p^m-1]$, computes $P((A^*)^{i}\cdot \vec{v})$. We then construct the circuit
	\[
	C_{A^*}(\vec{x}) \eqdef \begin{cases}
		P(\vec{0})  & \text{if } \vec{x}=\vec{0} \\
		C'_{A^*}(C'''_{A^*}(\vec{x})) & \text{if } \vec{x}\in\F_{p}^m\setminus\{\vec{0}\}.
	\end{cases}	
	\]
    Note that we can hardwire the value of $P(\vec{0})$. Also notice that if both Equations~\ref{eq:circuit_1} and~\ref{eq:circuit_2} are true (which happens with probability at least $1-1/(9p^m)$) we will get that for all $\vec{x}\in\F_{p}^m$,
	\[
	C_{A^*}(\vec{x}) = P(\vec{x}).
	\]
 
	However, we don't know the matrix $A^*$. Instead, we will run the above procedure for each $A\in S$ to obtain a list $\calC \eqdef \{C_A\}_{A\in S}$ of candidate circuits $C_A$. Then, with probability at least $1-1/(9p^m)$, $\calC$ contains at least one circuit (in particular, $C_{A^*}$) that computes the polynomial $P$. %
	
	Given the list of candidate circuits $\calC$, we now describe the second stage. First of all, given a circuit $C_A\in\calC$, we want to check if $C_A$ is sufficiently close to $P$.
	\begin{claim}\label{c:valid-test}
		There is a randomized algorithm $\mathsf{IsClose}$ that, given a circuit $B\colon\F_{p}^{m}\to\F_p$, $\delta\in(0,1]$, and oracle access to the polynomial $P$, runs in time $\poly(|B|)\cdot \log(1/\delta)$ such that
		\begin{itemize}
			\item if $\Pr_{\vec{x}}[B(\vec{x})=P(\vec{x})]=1$, the algorithm accepts with probability $1$, and
			\item if $\Pr_{\vec{x}}[B(\vec{x})= P(\vec{x})]\leq 3/4$, the algorithm rejects with probability at least $1-\delta$.
		\end{itemize}
	\end{claim}
	\begin{proof}[Proof of \autoref{c:valid-test}]\renewcommand{\qedsymbol}{$\diamond$}
		The algorithm picks $3\log(1/\delta)$ points uniformly at random from $\F_p^m$ and checks if $B$ and $P$ agree on all those points. If so, the algorithm accepts; otherwise it rejects. Note that if $\Pr_{\vec{x}}[B(\vec{x}) = P(\vec{x})]\leq 3/4$, then the probability that it accepts is at most $(3/4)^{3\log(1/\delta)}<\delta$. 
	\end{proof}
	For each $C_A\in \calC$, we run $\mathsf{IsClose}^{P}(C_A, \delta\eqdef1/(4|\calC|p^m))$ and pick the first one that the algorithm accepts. By the fact that $\calC$ contains at least one circuit that computes $P$ and by the property of the algorithm $\mathsf{IsClose}$ (\autoref{c:valid-test}), with probability at least $1-1/(4p^m)$, we will obtain some $D$-oracle circuit $C_{\mathsf{close}}$ such that
	\begin{equation}\label{eq:close}
		\Pr_{\vec{x}\gets\F_p^m}[C_{\mathsf{close}}(\vec{x})= P(\vec{x})] >  3/4.
	\end{equation}
	Conditioned on \autoref{eq:close}, by combining $C_{\mathsf{close}}$ with the self-corrector $\mathsf{PCorr}^{(-)}$ from \autoref{t:SC}, we get that for every $\vec{x}\in\F_p^m$, $\mathsf{PCorr}^{C_{\mathsf{close}}}(p,m,\Delta,\vec{x})=P(\vec{x})$ with probability at least $2/3$ (over the internal randomness of $\mathsf{PCorr}^{C_{\mathsf{close}}}$). Again, by using standard error reduction techniques and by picking a randomness uniformly at random, we can obtain in time $\poly(p,m)$, with probability at least $1-1/(4p^m)$, a $D$-oracle circuit $C$ that computes $P$.
	
	By a union bound, the above procedure gives, with probability at least $1-1/p^m$, a $D$-oracle circuit that computes the polynomial $P$.
 
    Finally, it is easy to verify that the running time is $\poly(p,m)$. %
\end{proof}

\section{Improved Chen--Tell Targeted Hitting Set Generator}\label{sec:improved-CT}

In this section, we prove~\autoref{thm:HSG-Chen-Tell-poly}, showing how to build a reconstructive hitting set generator from any uniform low-depth circuit.

\newcommand{\id}{\mathsf{id}}
\newcommand{\algopoly}{\mathbb{A}^{\sf poly}}

\subsection{Layered-Polynomial Representation}\label{sec:layered-poly-rep}

In this subsection we prove a self-contained version of the \emph{layered-polynomial representation}
that underlies the uniform hardness-vs-randomness framework of Chen and Tell \cite{ChenT21b}.
Conceptually, we start from a logspace-uniform low-depth circuit $C$ that computes a multi-output
function $C(1^n)$, and we transform its computation into a sequence of low-degree polynomials
$P_1,\dots,P_{d'}$ over a finite field. Each polynomial $P_i$ encodes one ``layer'' of the
computation of a carefully constructed circuit $D$ that has the same outputs as $C$, and whose
wiring is extremely uniform. The last polynomial $P_{d'}$ succinctly represents the output of
$C(1^n)$, and the sequence $(P_i)_i$ enjoys strong algebraic and algorithmic properties
(low degree, downward self-reducibility and uniform computability of their truth tables)
that will later allow us to build a reconstructive hitting set generator.

In more detail,  the first step is to ``arithmetize'' our low-depth circuit into a layered-polynomial representation. Roughly speaking, given a (uniform) circuit $C$ of depth $d$ and size $T$, we will produce a table of size $d'\times T'$ where $d' \approx d$ and $T' = \poly(T)$, such that the following key properties hold:\begin{itemize}
    \item [{\bf (Low-degree.)}] Each row is the ``truth table'' of a low-degree polynomial (thus admits self-correction properties).
    \item [{\bf (Faithful representation.)}] Given oracle access to the $d'$-th row, we can compute the output of $C(1^n)$ quickly.
    \item [{\bf (Downward self-reducibility.)}] For each $2\le i\le d'$, given oracle access to the $(i-1)$-th polynomial, we can quickly compute the output of the $i$-th polynomial on a given input. Moreover, the entries of the first row (corresponding to $i = 1$) can be computed quickly.
\end{itemize}
Later, as alluded to above, we will use these properties of the layered-polynomial representation to compile them into a reconstructive HSG.

We now formally describe our layered-polynomial representation, which can be proved by modifying
the construction in~\cite[Proposition~4.7]{ChenT21b}, itself based on the doubly efficient interactive
proof system of Goldwasser, Kalai and Rothblum~\cite{gkr15} and Goldreich's streamlined exposition
in~\cite{gol17a,gol18}. For the reader's convenience, we spell out the relevant parts of that construction
and reprove the parameter bounds that we need.

 In the following, letting $p$ be a power of $2$, and $\func{f}{\F_p^\ell}{\F_p}$, we use $\tt(f)$ to denote the length-$(p^\ell\cdot \log p)$ Boolean string that consists of $p^\ell$ blocks, where the $i$-th block corresponds to the Boolean encoding of the $i$-th element in $\F_p$.

\begin{theorem}[Layered-Polynomial Representation]\label{theo:poly-rep}
	There exist universal constants $c,c',\beta>1$ such that the following holds. Fix $\kappa \in \N$, and let $T,d,n,h,p \in \N$ all be sufficiently large such that \emph{(1)} $d \le T$ and $n \le T$, and \emph{(2)} $h,p$ are both nice powers of $2$ and $\log T \le h < p \le h^{27} \leq T$. \emph{(}Recall that $p$ is a nice power of $2$ if $p = 2^{2 \cdot 3^\lambda}$ for some $\lambda \in \N$.\emph{)} %
	
	Let $\vec{\ell} \eqdef (\kappa,T,d,n,h,p)$ be the input parameters, and let $\F \eqdef \F_p$. For a Turing machine $\TurM$ with description size $|\TurM| = \kappa \cdot \log T$, let \[C_{\TurM} \eqdef \Ckt[T,\kappa \cdot \log T,n,n](\TurM).\] Assuming $C_{\TurM} \ne \bot$ and $C_{\TurM}$ has depth at most $d$, there are $d' \eqdef c \kappa \cdot \log^2 T \cdot (d + \kappa^2 \log T)$ polynomials $\mleft(P_i^{\vec{\ell},\TurM}\mright)_{i \in [d']}$ such that the following hold \emph{(}below we write $P_i^{\vec{\ell},\TurM}$ as $P_i$ for simplicity\emph{)}:
	
	\begin{enumerate}
		\item Let $H \subset \F$ be the first $h$ elements of $\F$, and let $m$ be the smallest power of $3$ such that $h^m \ge T^{\beta\kappa}$. Each polynomial $P_i$ is from $\F^{3m}$ to $\F$ and has total degree at most $\Delta = c \cdot h \cdot \log^{3}(T)$.

		\item $P_{d'}$ represents the outputs of $C_{\TurM}$. 
        
        \underline{Formally}, fix an injective function $\func{\id}{[n]}{H^m}$ in an arbitrary but canonical way.\footnote{For simplicity, we will ignore the complexity of computing $\id$ and its inverse since it is negligible.} We have that for every $i \in [n]$, $\left(C_{\TurM}(1^n)\right)_i = P_{d'}(\id(i),0^{2m})$.
		
		\item The polynomials $(P_i)_{i\in [d']}$ can be computed by uniform low-depth circuits.
        
        \underline{Formally}, let $ T_{\sf poly} \eqdef T^{c \cdot \kappa}$ and $ d_{\sf poly} \eqdef c \cdot (d \log T + \kappa^2 \log^2 T)$. There is a Turing machine $\TurM_{\sf poly}$ of description length $\log T_{\sf poly}$ such that for 
		\[
		C_{\sf poly} \eqdef \Ckt\mleft[T_{\sf poly}, \log T_{\sf poly},\log d' , |\F|^{3m} \cdot \log |\F| \mright](\TurM_{\sf poly}),
		\]
		it holds that \emph{(1)} for every $i \in [d']$ $C_{\sf poly}(i) = \tt(P_i)$ and \emph{(2)} $C_{\sf poly}$ has depth $d_{\sf poly}$.

            Moreover, there is a polynomial-time algorithm $\algopoly_{\vec{\ell}}$ that takes $\TurM \in \bs{\kappa \log T}$ as input, and outputs the description of $\TurM_{\sf poly}$.

		\item The polynomials $(P_i)_{i\in [d']}$ are downward self-reducible.
        
        \underline{Formally}, there is a $\max(n,h) \cdot h^{c'}$-time algorithm $\Base$ that takes inputs $\vec{\ell}$, $\TurM \in \bs{\kappa \cdot \log T}$, and $\vec{w} \in \F^{3m}$, and outputs $P_1(\vec{w})$. Also, there is an $h^{c'}$-time oracle algorithm $\DSR$ that takes inputs $\vec{\ell}$, $\TurM \in \bs{\kappa \cdot \log T}$, $i \in \{2,\dotsc,d'\}$, and $\vec{w} \in \F^{3m}$, and oracle access to a polynomial $\func{\WT{P}}{\F^{3m}}{\F}$, such that when it is given $P_{i-1}$ as the oracle, it outputs $P_i(\vec{w})$.
	\end{enumerate}
	
\end{theorem}
\begin{proof}
	We will follow the proof of~\cite[Proposition~4.7]{ChenT21b}, which in turn follows~\cite{gol17a} (see also~\cite{gol18}). %
    
    Recall that we use $\vec{\ell}$ to denote the input parameters $(\kappa,T,d,n,h,p)$.  In the following, we will simply use $C$ to denote the (low-depth) circuit $C_{\TurM} = \Ckt\mleft[T,\kappa \cdot \log T,n,n\mright](\TurM)$ for notational convenience, but we stress that $C$ depends on both $\vec{\ell}$ and $\TurM$ (and so do the later circuits constructed from $C$).
	
	\subsubsection{Construction of a Highly Uniform Circuit \texorpdfstring{$D$}{D}}

    The goal of this subsection is to replace $C$ by a layered circuit $D$ of comparable depth
	and polynomially larger width, whose wiring is so uniform that it can be decided by very small
	formulas. This is what ultimately allows us to arithmetize the adjacency relation of $D$.
	
	In more detail, recall that $C$ is a \emph{logspace}-uniform low-depth circuit. The first step is to construct a \emph{logtime}-uniform low-depth circuit $D$ with the same outputs as $C$, i.e., $D(1^n) = C(1^n)$. Given input $1^n$, $D$ first computes a description of $C = \Ckt\mleft[T,\kappa \cdot \log T,n,n\mright](\TurM)$ (represented as a $T \times T \times T$ adjacency relation tensor) and then computes the $\Eval$ function $\langle \langle C\rangle,n,d \rangle \mapsto C(1^n)$. Let $s \eqdef \kappa \cdot \log T$ and $s' \eqdef O(s + \log(3\log T))$ be such that each configuration of $\TurM$ on $3\log T$-bit inputs can be described by $s'$ bits.
	
	The circuit $D$ is constructed by composing the following three subcircuits. Let $\mu \in \N$ be a sufficiently large universal constant. We will describe and analyze their complexities (or state the complexity bounds proved in~\cite{ChenT21b,gol17a}).
	
	\begin{enumerate}
		\item ({\bf Computing the adjacency matrices for configurations.}) The first circuit $D^{(1)}$ takes $n$ bits as input (which are supposed to be $1^n$), and outputs a list of $T^3$ matrices from $\bs{2^{s'} \times 2^{s'}}$, such that the $(u,v,w)$-th matrix\footnote{We use $(u,v,w) \in [T]^3$ to denote the integer $(u-1)T^2+(v-1)T+w \in [T^3]$.} $M^{(u,v,w)}$ satisfies the following condition: for every $\gamma,\gamma' \in \bs{s'}$, $M^{(u,v,w)}[\gamma,\gamma'] = 1$ if and only if $\algoNext(\TurM,s,(u,v,w),\gamma,\gamma')$ (\ie, $\gamma'$ is the configuration obtained by running $\TurM$ for one step on configuration $\gamma$ and input $(u,v,w)$ with space bound $s$). That is, $M^{(u, v, w)}$ is the transition matrix of $\TM$ on input $(u, v, w)$; looking ahead, we will compute the $(u, v, w)$-th entry of the adjacency relation of $C$, \ie $\TM(u, v, w)$, by taking powers of $M^{(u, v, w)}$. We also recall we assumed that if $\gamma$ is the accepting or the rejecting configuration, then its next configuration is $\gamma$ itself.
		
		{\bf Complexity of $D^{(1)}$.} $D^{(1)}$ can be implemented by a projection (\ie, depth $d_{D^{(1)}} = 2$ and size $T_{D^{(1)}} =  T^3 \cdot 2^{2s'}$).\footnote{Note that we can implement projections and restrictions of input bits to $0$ and $1$ using two layers of $\NAND$ gates.} Moreover, from~\autoref{fact:check-next-configuration}, given $\vec{\ell}$ and $\TurM$, in polynomial time we can compute a Turing machine $\TurM_{D^{(1)}} \in \bs{(\kappa + \mu) \cdot \log T}$ such that
		\[
		\Ckt\mleft[T_{D^{(1)}}, s_{D^{(1)}},n,T^3 \cdot 2^{2s'}\mright]\mleft(\TurM_{D^{(1)}}\mright) = D^{(1)},
		\]
		where $s_{D^{(1)}} = \mu \cdot s'$.
		
		\item ({\bf Computing the adjacency relation tensor of $C$ via matrix multiplication.}) The second circuit $D^{(2)}$ takes a list of $T^3$ matrices from $\bs{2^{s'} \times 2^{s'}}$ as input, and outputs the layered adjacency tensor ${\sf T}_C \in \bs{T \times T \times T}$ of $C$ followed by the encoding of a pair $(n,d)$. That is, given the output of $D^{(1)}(1^n)$, for every $(u, v, w) \in [T]^3$, ${\sf T}_C[u, v, w] = 1$ if and only if $\TurM(u, v, w) = 1$, which means that the gates in $C$ indexed by $v$ and by $w$ feed into the gate in $C$ indexed by $u$. The tensor ${\sf T}_C$ can be computed by performing repeated matrix squaring $s'$ times: ${\sf T}_C[u, v, w] = (M^{(u, v, w)})^{2^{s'}}$.
		
		{\bf Complexity of $D^{(2)}$.} $D^{(2)}$ can be implemented by a circuit of depth $d_{D^{(2)}} = \mu \cdot (s')^2$ and size $T_{D^{(2)}} =  T^3 \cdot 2^{\mu s'}$. Moreover, from~\cite{ChenT21b,gol17a} (note that $D^{(2)}$ does not depend on $\TurM$), given $\vec{\ell}$, in polynomial time we can compute a Turing machine $\TurM_{D^{(2)}} \in \bs{\mu \cdot \log T}$ such that
		\[
		\Ckt\mleft[T_{D^{(2)}}, s_{D^{(2)}},T^3 \cdot 2^{2s'},T^3 + |(n,d)|\mright]\mleft(\TurM_{D^{(2)}}\mright) = D^{(2)},
		\]
		where $s_{D^{(2)}} = \mu \cdot s'$.
		
		\item ({\bf Computing $\Eval$.}) The final circuit $D^{(3)}$ takes $\langle \langle C\rangle,n,d \rangle$ as input, and outputs $\Eval(\langle C\rangle,n,d)$.
		
		{\bf Complexity of $D^{(3)}$.} $D^{(3)}$ can be implemented by a circuit of depth $d_{D^{(3)}} = \mu \cdot d \cdot \log T$ and size $T_{D^{(3)}} =  T^{\mu}$. Moreover, from~\cite{ChenT21b,gol17a} (note that $D^{(3)}$ does not depend on $\TurM$), given $\vec{\ell}$, in polynomial-time we can compute a Turing machine $\TurM_{D^{(3)}} \in \bs{\mu \cdot \log T}$ such that
		\[
		\Ckt\mleft[T_{D^{(3)}}, s_{D^{(3)}},T^3 + |(n,d)|, n\mright]\mleft(\TurM_{D^{(3)}}\mright) = D^{(3)},
		\]
		where $s_{D^{(3)}} = \mu \cdot s'$.
	\end{enumerate}
	
	Formally, we have
	\[
	D = D^{(3)} \circ D^{(2)} \circ D^{(1)}.
	\]
	
	Let $\beta \in \N$ be a sufficiently large constant such that $\beta \ge 100\mu$. The following claim summarizes the required properties of $D$ for us.
	
	\newcommand{\dDTurM}{d_{D}}
	\newcommand{\algoPhip}{\mathbb{A}_{\Phi'}}
	
	\begin{claim}\label{claim: uniformity of D}
		The following properties about the circuit $D$ are true.
		\begin{enumerate}
			\item The depth of $D$ is $\dDTurM = \beta \cdot (\kappa^2 \cdot \log^2 T + d \cdot \log T)$ and the width of $D$ is $T'_D=T^{\beta\kappa}$.
			\item The layered adjacency relation function $\func{\Phi'}{[\dDTurM]\times\bs{3\log(T'_D)}}{\bs{}}$ of $D$ can be decided by a formula of depth $O(\log\log T)$ and size $O(\log^3 T)$. Moreover, there is an algorithm $\algoPhip$ that given $\vec{\ell}$ and $\TurM$ as input, outputs the formula above in $O(\kappa \log T)$ space.
			
			\item There is a Turing machine $\TurM_D \in \bs{\beta\kappa\log T}$ such that
			\[
			\Ckt[T_D,s_D,n,n](\TurM_D) = D,\footnote{Note that $\Ckt$ generates the unlayered version of $D$ of size $T'_D \cdot (\dDTurM+1)$, Without loss of generality we can assume the first $T'$ gates are on the first layer, the next $T'$ gates are on the second layer, and so on.}
			\]
			where $T_D = T'_D \cdot (\dDTurM+1)$ and $s_D = \beta\kappa\log T$. Moreover, given $\vec{\ell}$ and $\TurM$ as input, the description of $\TurM_D$ can be computed in polynomial time.
		\end{enumerate}
	\end{claim}
	\begin{proof}[Proof of~\autoref{claim: uniformity of D}]\renewcommand{\qedsymbol}{$\diamond$}            
		By construction, the size of $D$ is bounded by $\poly(T) \cdot 2^{O(s')} \le T^{O(\kappa)}$ (recall that $s' = O(s+\log(3\log T))$ and $s = \kappa \log T$), and its depth is bounded by $O(s^2 + d \cdot \log T)$. The first bullet then follows directly from the fact that $\beta$ is sufficiently large. %
		
		Recall that the $D^{(1)}$ part of $D$ has depth $d_{D^{(1)}} = 2$. To see the complexity of computing $\Phi'(i, -, -, -)$ for $i > 2$, we note that the layers corresponding to $D^{(2)}$ and $D^{(3)}$ \emph{do not} depend on $\TurM$. Hence the complexity of computing their layered adjacency  relation function follows directly from~\cite[Claim~4.7.1]{ChenT21b}.\footnote{We note that although~\cite{ChenT21b,gol17a} only claim a $\polylog(T)$ bound on the formula size, the formula is indeed very simple and its size and depth can be easily bounded by $O(\log^3 T)$ and $O(\log\log T)$, respectively; see~\cite[Page 8-9]{gol17a} for details.} Also, the complexity of computing $\Phi'(i, -, -, -)$ for $i \in \{1,2\}$ follows directly from~\autoref{fact:check-next-configuration}. To see the moreover part, again, the case for $i > 2$ follows from~\cite[Claim~4.7.1]{ChenT21b}, and the case for $i \in \{1,2\}$ follows from~\autoref{fact:check-next-configuration}.\footnote{Strictly speaking we need to combine the formulas for two cases to obtain a single formula for $\Phi'$. The overhead of doing so is negligible so we omit this discussion here.}
		
		Finally, to obtain the algorithm that computes $\TurM_D$, we simply apply the composition $\algocomp$ (from~\autoref{fact:comp-of-circuits}) twice to compose the circuits $D^{(1)}, D^{(2)}, D^{(3)}$ in order and add some dummy gates to the circuit. The space bound and the description size bound can also be verified easily.
	\end{proof}

	\subsubsection{Arithmetization of \texorpdfstring{$D$}{D}}
	
	\newcommand{\emd}{\mathsf{emd}}
We now pass from the highly uniform layered circuit $D$ to a collection of low-degree
	polynomials over a finite field that encode its layer-by-layer computation. This is
	the algebraic ``arithmetization'' step in the spirit of \cite{gkr15,gol17a,ChenT21b}.
    
	The construction of the polynomials and their corresponding algorithms can then be done in the same way as in~\cite{ChenT21b}. We only state the necessary changes to establish our theorem and, in particular, we make the dependence on parameters explicit. 
	
	Note that $|\F|^{m} = p^m \le \poly(h^{27m}) \le T^{O(\beta\kappa)} \le T^{O(\kappa)}$ ($\beta$ is a universal constant), from our assumption that $p \le h^{27}$ and our choice of $m$.

	First, we need an arithmetization of each layered adjacency relation $\Phi'_i \eqdef \Phi'(i, -, -, -)$.
	
	\def\APhi{\mathbb{A}_{\hat{\Phi}}}
	\def\Earith{E^{\mathsf{arith}}}
	\def\TMarith{\TurM^{\mathsf{arith}}}
	
	\begin{claim}\label{claim:hat-Phi}
		For $i\in[\dDTurM]$ there exists a polynomial $\func{\hat{\Phi}_i}{\F^{3m}}{\F}$ that satisfies the following:
		\begin{enumerate}
			\item $\hat{\Phi}_i$ is a low-degree extension of $\Phi'_i$.
            
            \underline{Formally}, For every $(\vec{w},\vec{u},\vec{v})=z_1,...,z_{3m}\in H^{3m}$ we have that $\hat{\Phi}_i(\vec{w},\vec{u},\vec{v})=1$ if gate $\vec{w}$ in the $i^{th}$ layer of $D$ is fed by gates $\vec{u}$ and $\vec{v}$ in the $(i-1)^{th}$ layer of $D$, and $\hat{\Phi}_i(\vec{w},\vec{u},\vec{v})=0$ otherwise. 
			
			Moreover, the degree of $\hat{\Phi}_i$ is at most $O(h\cdot\log^3 T)$, and there exists an algorithm that on input $\vec{\ell}, \TM, i, \vec{w}, \vec{u}, \vec{v}$, computes $\hat{\Phi}_i(\vec{w},\vec{u},\vec{v})$ in $\poly(h)$ time.
			
			\item The truth table of $\hat{\Phi}_i$ can be computed in low depth.
            
            \underline{Formally}, for a universal constant $c_1>1$, there exists a circuit $C_{\hat{\Phi}}$ of size $T_{\hat{\Phi}} \eqdef T^{c_1\kappa}$ and depth $c_1\kappa \cdot \log T$ that on input $i \in [d_D]$ outputs $\tt(\hat{\Phi}_i) \in \F^{|\F|^{3m}}$ \emph{(}represented as a Boolean string\emph{)}. Moreover, there is a polynomial-time algorithm $\APhi$ that takes $\vec{\ell}$ and $\TurM \in \{0, 1\}^{\kappa\log T}$ as input, and outputs the description of a Turing machine $\TM_{\hat{\Phi}} \in \{0, 1\}^{c_1\kappa\log T}$ such that
			\[C_{\hat{\Phi}} = \Ckt\mleft[T_{\hat{\Phi}}, c_1\cdot \kappa\log T, \log(d_D+1), |\F|^{3m}\log |\F|\mright](\TM_{\hat{\Phi}}).\]
		\end{enumerate}
	\end{claim}
        \begin{proof}[Proof Sketch of~\autoref{claim:hat-Phi}]\renewcommand{\qedsymbol}{$\diamond$}
            
            We first define $\hat{\Phi}_i$ and then establish each item separately. 
            
            \paragraph*{Construction of $\hat{\Phi}_i$.} Let $F_{\Phi'}$ be the $O(\log\log T)$-depth $O(\log^3 T)$-size formula computing $\Phi' \colon [\dDTurM] \times \{0,1\}^{3 \cdot \log T'_D} \to \bs{}$ from~\autoref{claim: uniformity of D}. For every $i \in [\dDTurM]$, let $F_{i}$ be the restriction of $F_{\Phi'}$ by fixing the first input to be $i$. Then, we arithmetize $F_{i}$ by replacing every $\texttt{NAND}$ gate in $F_{i}$ by an arithmetic gate $\widetilde{\texttt{NAND}}:\F^2 \to \F$ computing $\widetilde{\texttt{NAND}}(u, v) \eqdef 1-uv$. Denote the new formula (which is now an \emph{arithmetic} formula) by $\hat{F}_{i}$. 
            
            For each $j$, let $\pi_j \colon H \to \{0,1\}$ be the mapping that maps $z \in H$ to the $j$-th %
            bit of the encoding  of $z$. Note that since $H$ consists of the smallest $h$ elements in $\F$, we know that $\pi(z) = (\pi_1(z),\dotsc,\pi_{\log h}(z))$ is a bijection between $H$ and $\bs{\log h}$.\footnote{More specifically, by our specific encoding of $H$ as strings from $\{0,1\}^{\log |\F|}$, $\pi(z)$ is simply the first $\log h$ bits of the encoding of $z$, hence it can be computed by a projection.}
            
            For each $j \in [\log h]$, let $\hat{\pi}_j \colon \F \to \F$ be the unique degree-$(h-1)$ extension of $\pi_j$ to $\F$, that can be computed via standard interpolation via logspace-uniform circuits of $O(\log (h \cdot \log |\F|)) = 
            O(\log T)$
            depth  and
            $\polylog(T)$ size~\cite{HesseAB02,HealyV06} (see  \cite[Claim 4.7.2]{ChenT21b} for the details). We also let $\hat{\pi}(z) = (\hat{\pi}_1(z),\dotsc,\hat{\pi}_{\log h}(z))$. Then, we set
            \[
            \hat{\Phi}_i(z_1, \dots, z_{3m}) \eqdef \hat{F}_{i}\mleft(\hat{\pi}(z_1), \hat{\pi}(z_2), \dots, \hat{\pi}(z_{3m})\mright).
            \]
            
            We also use $F_{\hat{\Phi}_i}$ to denote the \emph{arithmetic} formula on the right side above that computes  $\hat{\Phi}_i$.
            
            From the construction above, the first item of the claim can be proved identically as~\cite[Claim~4.7.2]{ChenT21b}. In particular, the degree of each $\hat{F}_i$ is at most the formula complexity of $F_{\Phi'}$ and the degree of $\hat{\pi}$ is at most $O(h)$, hence the degree of $\hat{\Phi}_i$ is $O(h\cdot \log^3 T)$. %
			It remains to establish the second item. 
            
            \paragraph*{Construction of $C_{\hat{\Phi}}$.} We hardwire the description of $F_{\Phi'}$ into $C_{\hat{\Phi}}$. The circuit $C_{\hat{\Phi}}$ takes $i \in [\dDTurM]$ as input, performs the above computation to obtain a description of the arithmetic formula $F_{\hat{\Phi}_i}$, and then outputs the truth table of $F_{\hat{\Phi}_i}$ by evaluating it on all vectors in $\F^{3m}$.
            
            In more detail, computing the description of $F_{\hat{\Phi}_i}$ from the description of $F_{\Phi'}$ and $i$ can be done %
            in $O(\log T)$ depth and
            $\polylog(T)$ size. $C_{\hat{\Phi}}$ then evaluates $F_{\hat{\Phi}_i}$ on all vectors from $\F^{3m}$, which can be done in $\poly(|\F|^m)$ size and $O(\log (|\F|^m))$ depth. The second item (except for the moreover part) then follows by setting $c_1$ to be sufficiently large and recalling that $|\F|^m \le T^{O(\beta\kappa)}$.

            \paragraph*{Establishing the uniformity.}
            
            \newcommand{\muhat}{\mu_{\hat{\Phi}}}
            Finally, we establish the moreover part of the second item. Let $\muhat \in \N$ be a sufficiently large universal constant that depends on the space complexity of the algorithm $\algoPhip$ from~\autoref{claim: uniformity of D}.
            
            Our algorithm $\APhi$ works as follows:
            
            \begin{enumerate}
            	\item We first construct a Turing machine $\TurM_{[1]}$ with $\vec{\ell}$ and $\TurM$ hardwired. $\TM_{[1]}$ corresponds to a circuit $C_{[1]}$ that takes $i \in [\dDTurM]$ as input and outputs $i$ together with the description of $F_{\Phi'}$.\footnote{Precisely, $\TM_{[1]}$ simulates $\algoPhip$ on input $\vec{\ell}$ and $\TurM$ to construct a projection that maps $i$ to $(i,F_{\Phi'})$.} $C_{[1]}$ has depth $d_{[1]} = O(1)$ and size $T_{[1]} = \polylog(T)$. Let $s_{[1]} = \muhat \cdot \kappa \log T$, we have
            	\[
            	C_{[1]} = \Ckt\mleft[T_{[1]},s_{[1]},\log d_D, \log d_D + |F_{\Phi'}|\mright]\mleft(\TM_{[1]}\mright)
            	\]
            	and $\TM_{[1]}$ has description size at most $|\TurM| + \muhat \cdot \log T = (\kappa + \muhat) \cdot \log T$.
            	
            	Here, we crucially use the fact that the algorithm $\mathbb{A}_{\Phi'}$ from~\autoref{claim: uniformity of D} runs in $O(\kappa\log T)$ space (and $\muhat$ is sufficiently large).

            	\item Then we construct a Turing machine $\TurM_{[2]}$ that corresponds to a circuit $C_{[2]}$ that takes $i \in [\dDTurM]$ together with the description of $F_{\Phi'}$ as input and outputs $\tt(\hat{\Phi}_i)$. By the discussion above, from $\vec{\ell}$ we can compute a Turing machine $\TurM_{[2]} \in \bs{\muhat \kappa \log T}$ such that for $T_{[2]} = \poly(|\F|^m) \le T^{\muhat \kappa}$, $d_{[2]} = O(\log (|\F|^m)) \le \muhat \kappa \cdot \log T$, $s_{[2]} = \muhat \kappa \log T$, we have
            	\[
            	C_{[2]} = \Ckt\mleft[T_{[2]},s_{[2]},\log d_D+|F_{\Phi'}|,|\F|^{3m}\mright]\mleft(\TurM_{[2]}\mright)\mathcomma
            	\]
            	and $C_{[2]}$ has depth $d_{[2]}$.

            	\item Finally, $\APhi$ composes $C_{[1]}$ and $C_{[2]}$ by applying~\autoref{fact:comp-of-circuits} and outputs the obtained Turing machine as $\TurM_{\hat{\Phi}}$. Setting $c_1$ sufficiently large completes the proof.\qedhere
            \end{enumerate}
            
        \end{proof}

	Then we define the following polynomials, according to~\cite[Definition~4.6]{ChenT21b}.

	\paragraph*{Input polynomial.} Let $\func{\al_0}{H^m}{\bs{}}$ represent the string $1^n0^{h^m-n}$, and let $\func{\hat{\al}_0}{\F^m}{\F}$ be the ``arithmetization'' of $\al_0$, defined by 
	\[
	\hat{\al}_0(\vec{w}) = \sum_{\vec{z}\in H^{m'}\times\set{0}^{m-m'}}\de_{\vec{z}}(\vec{w})\cdot \al_0(\vec{z}).
	\]
	Here, $m' \leq m$ is the minimal integer such that $h^{m'} \geq n$, and $\de_{\vec{z}}$ is Kronecker's delta function (\ie, $\de_{\vec{z}}(\vec{w})=\prod_{j\in[m]}\prod_{a\in H\setminus\set{z_j}}\frac{w_j-a}{z_j-a}$). 
	
	\paragraph*{Layer polynomials.} For each $i\in[\dDTurM]$, let $\func{\al_i}{H^m}{\bs{}}$ represent the values of the gates at the $i^{th}$ layer of $D$ in the computation of $D(1^n)$ (with zeroes in locations that do not index valid gates). Recall that we consider circuits consisting of  $\mathsf{NAND}$ gates, where for $a,b \in \{0,1\}$ we have $\mathsf{NAND}(a,b) = 1 - a \cdot b$. We define $\func{\hat{\al}_i}{\F^m}{\F}$ as
	\begin{equation}\label{eq: layered polynomials}
	\hat{\al}_i(\vec{w}) = \sum_{\vec{u},\vec{v}\in H^m} \hat{\Phi}_i(\vec{w},\vec{u},\vec{v})\cdot\left(1-\hat{\al}_{i-1}(\vec{u})\cdot\hat{\al}_{i-1}(\vec{v})\right).
	\end{equation}
	
	Note that $\hat{\alpha}_i$ is an ``arithmetization'' of $\alpha_i$ in the sense that for every $\vec{w} \in H^m$, $\alpha_i(\vec{w}) = \hat{\alpha}_i(\vec{w})$.
	
	\paragraph*{Sumcheck polynomials.} For each $i\in[\dDTurM]$, let $\func{\hat{\al}_{i,0}}{\F^{3m}}{\F}$ be the polynomial
	\begin{equation}\label{eq: first sumcheck polynomial}
	\hat{\al}_{i,0}(\vec{w},\sigma_{1},...,\sigma_{2m}) = \hat{\Phi}_i(\vec{w},\sigma_{1,...,m},\sigma_{m+1,...,2m})\cdot\left(1-\hat{\al}_{i-1}(\sigma_{1,...,m})\cdot\hat{\al}_{i-1}(\sigma_{m+1,...,2m})\right).
	\end{equation}
	
	For every $j\in[2m]$, let $\func{\hat{\al}_{i,j}}{\F^{3m-j}}{\F}$ be the polynomial
	\begin{align}
		&\hat{\al}_{i,j}(\vec{w},\sigma_{1},...,\sigma_{2m-j}) = \nonumber\\
		&\sum_{\sigma_{2m-j+1},...,\sigma_{2m}\in H}\hat{\Phi}_i(\vec{w},\sigma_{1,...,m},\sigma_{m+1,...,2m})\cdot\left(1-\hat{\al}_{i-1}(\sigma_{1,...,m})\cdot\hat{\al}_{i-1}(\sigma_{m+1,...,2m})\right),\label{eq: sumcheck polynomials}
	\end{align}
	where the notation $\sigma_{k,...,k+r}$ means $\sigma_k,\sigma_{k+1},...,\sigma_{k+r}$. It is easy to check that $\hat{\al}_{i,2m}=\hat{\al}_i$.	Clearly, we also have
    \begin{equation}
        \hat{\alpha}_{i, j}(\vec{w}, \sigma_1, \dots, \sigma_{2m-j}) = \sum_{\sigma_{2m-j+1} \in H}\hat{\alpha}_{i, j-1}(\vec{w}, \sigma_1, \dots, \sigma_{2m-j}, \sigma_{2m-j+1}).
        \label{eq: dsr for sumcheck polynomials}
    \end{equation}

	We are now ready to define the sequence $(P_i)_{i \in [d']} = \mleft(P_i^{\vec{\ell},\TurM}\mright)_{i \in [d']}$. We set $d' \eqdef (2m+1) \cdot \dDTurM + 1$ and 
	\[
	(P_i)_{i \in [d']} =
	(\hat{\al}_0,\hat{\al}_{1,0},\dotsc,\hat{\al}_{1,2m},\hat{\al}_{2,0},\dotsc,\hat{\al}_{2,2m},\dotsc,\hat{\al}_{\dDTurM,0},\dotsc,\hat{\al}_{\dDTurM,2m}).
	\]
	For those $\hat{\al}_{i,j}$ (and $\hat{\al}_0$) that take less than $3m$ variables, we add some dummy variables at the end to make all polynomials taking exactly $3m$ variables.
	
	From the definitions of $m$ and $\dDTurM$, we have $m \le 3 \cdot \beta \kappa \cdot \log T + 1$ and $\dDTurM = \beta \cdot (\kappa^2 \cdot \log^2 T + d \cdot \log T)$. Hence, we have $d' =(2m+1)\cdot \dDTurM + 1 \le c \kappa \cdot \log^2 T \cdot (d + \kappa^2 \log T) $ as desired.\footnote{We can add identical polynomials at the end to make $d'$ exactly $c \kappa \cdot \log^2 T \cdot (d + \kappa^2 \log T)$ as in the theorem statement.}
	
	Below we verify the desired properties of the sequence $(P_i)_{i \in [d']}$. 
	
	\paragraph*{Arithmetic setting, faithful representation, and downward self-reducibility.} First, the degree bounds of all these polynomials follow directly from their definitions and from the degree bound on $\hat{\Phi}_i$ (from \autoref{claim:hat-Phi}). The faithful representation property also follows directly from the definition of $\alpha_{\dDTurM}$ and $\hat{\al}_{\dDTurM,2m} = \hat{\al}_{\dDTurM}$.
    
    Finally, the downward self-reducibility of the polynomials follows  from the complexity of computing $\hat{\Phi}_i$ (from~\autoref{claim:hat-Phi}) and the definitions of these polynomials, similarly to the proof of \cite[Proposition~4.7]{ChenT21b}. For example, for each $2\le i\le d_D$ and each $\vec{x} \in \F^{3m}$, one can compute $\hat{\alpha}_{i, 0}(\vec{x})$ in $\poly(h)$ time with oracle access to $\hat{\alpha}_{i-1, 2m} = \hat{\alpha}_{i-1}$ using \eqref{eq: first sumcheck polynomial}; for each $1\le i\le d_D$, $1\le j\le 2m$ and each $\vec{x} \in \F^{3m-j}$, one can compute $\hat{\alpha}_{i, j}(\vec{x})$ in $\poly(h)$ time with oracle access to $\hat{\alpha}_{i, j-1}$ using \eqref{eq: dsr for sumcheck polynomials}. %
	
	\subsubsection{Complexity of the Polynomials}
	
	\newcommand{\arith}{\mathsf{arith}}
	\newcommand{\tthatPhi}{\tt\text{-}\lower0.2ex\hbox{\scriptsize$\hat{\Phi}$}}

    We finally verify that the truth tables of all
	polynomials $P_i$ can be computed by a single uniform low-depth circuit $C_{\sf poly}$ whose
	description can be produced from $\vec{\ell}$ and $\TM$ in polynomial time (Item (3) of  \autoref{theo:poly-rep}).

    The argument below is straightforward but tedious. We first give a high-level overview.
	
	\paragraph*{High-level overview of the construction.} To construct the circuit $C_{\sf poly}$ that maps $i' \in [d']$ to $\tt(P_{i'})$, we will construct three subcircuits $C_{\al}$, $C_{\tthatPhi}$, and $C_{\arith}$ such that:
	\begin{enumerate}
		\item $C_{\al}$ maps $i'\in[d']$ to $(\tt(\alpha_{i-1}), i, j)$. Here, if $i' \ge 2$, then $i\in \{1, \dots, d_D\}$ and $j\in\{0, 1, \dots, 2m\}$ satisfies that $P_{i'} = \hat{\alpha}_{i, j}$ and $\tt(\alpha_{i-1}) \in \bs{h^m}$ denotes the values of the gates at the $i$-th layer of $D$. If $i' = 1$, then we consider $i=j=0$ and $C_{\al}$ outputs $(\tt(\alpha_0), 0, 0)$. %
		\item $C_{\tthatPhi}$ maps $(\tt(\alpha_{i-1}), i, j)$ to $(\tt(\alpha_{i-1}), i, j, \tt(\hat{\Phi}_i))$.
		\item $C_{\arith}$ maps $(\tt(\alpha_{i-1}), i, j, \tt(\hat{\Phi}_i))$ to $\tt(\hat{\alpha}_{i, j}) \in \F^{|\F|^{3m}}$.
	\end{enumerate}
	
	$C_{\sf poly}$ is then simply $C_{\arith} \circ C_{\tthatPhi} \circ C_{\al}$. To compute the Turing machine $\TurM_{\sf poly}$ that corresponds to $C_{\sf poly}$, we construct the Turing machines $\TurM_{\al}$, $\TurM_{\tthatPhi}$, and $\TurM_{\arith}$ corresponding to the three circuits above, and compose them using~\autoref{fact:comp-of-circuits}.
	
	\paragraph*{Construction of $C_{\al}$ and $\TurM_{\al}$.} First, we construct a circuit $C_{\al}$ that takes as input $i'\in[d']$ and outputs $(\tt(\alpha_{i-1}), i, j)$. To construct $C_\al$, we first compute $i$ and $j$ from $i'$ using basic arithmetic, and then truncate $D$ up to its $i$-th layer. It is easy to see that given the Turing machine $\TurM_D$ that specifies the circuit $D$, in polynomial-time we can construct a Turing machine $\TurM_{\al} \in \bs{|\TurM_D| + \mu}$ such that (in what follows, we write $|\langle i,j\rangle| = \log(\dDTurM+1) + \log(2m+1)$ for convenience):
	\[
	\Ckt[T_{\al},s_\al, \log (d'),h^m + |\langle i, j\rangle|](\TurM_\al) = C_\al,
	\]
	where $T_\al = \mu \cdot T_D$, $s_\al = 2 s_{D}$. Moreover, $C_{\al}$ has depth at most $d_{\al} = 2 \cdot \dDTurM$.
	
	\paragraph*{Construction of $C_{\tthatPhi}$ and $\TurM_{\tthatPhi}$.} Let $c_1$ be the universal constant from~\autoref{claim:hat-Phi}. Next we construct a circuit $C_{\tthatPhi}$ that on input $(\tt(\alpha_{i-1}), i, j)$, outputs $(\tt(\alpha_{i-1}), i, j, \tt(\hat{\Phi}_i))$. It is straightforward to obtain this circuit from the circuit $C_{\hat{\Phi}}$ constructed in \autoref{claim:hat-Phi}. In other words, given $\vec{\ell}$ and $\TurM_{\hat{\Phi}} \in \{0, 1\}^{c_1\kappa\log T}$ as input (where $\TurM_{\hat{\Phi}}$ is the Turing machine that generates $C_{\hat{\Phi}}$ as defined in \autoref{claim:hat-Phi}), we can compute a Turing machine $\TurM_{\tthatPhi} \in \{0, 1\}^{2\cdot c_1\kappa\log T}$ such that
    \[
        \Ckt[T_{\tthatPhi}, s_{\tthatPhi}, h^m + |\langle i, j\rangle|, h^m + |\langle i, j\rangle| + |\F|^{3m}\log |\F|](\TurM_{\tthatPhi}) = C_{\tthatPhi},
    \]
    where $T_{\tthatPhi} = T^{2c_1\kappa}$, $s_{\tthatPhi} = 2c_1\kappa\log T$. Moreover, $C_{\tthatPhi}$ has depth $d_{\tthatPhi} = 2c_1\kappa\log T$.
    
    \paragraph*{Construction of $C_{\arith}$ and $\TurM_{\arith}$.} We construct a circuit $C_{\arith}$ that maps $(\tt(\alpha_{i-1}), i, j, \tt(\hat{\Phi}_i))$ to $\tt(\hat{\alpha}_{i, j}) \in \F^{|\F|^{3m}}$. (Recall that throughout this proof we view $\hat{\al}_{i,j}$ as a $3m$-variable polynomial by adding dummy variables at the end.) Suppose that $i \geq 1$ (the base case $i = j = 0$ can be handled similarly). If $j = 0$, $C_{\arith}$ computes $\tt(\hat{\alpha}_{i, 0})$ using \autoref{eq: first sumcheck polynomial}, otherwise ($j \ge 1$) $C_{\arith}$ computes $\tt(\hat{\alpha}_{i, j})$ using \autoref{eq: sumcheck polynomials}. (Note that both \autoref{eq: first sumcheck polynomial} and \autoref{eq: sumcheck polynomials} only depend on the values of $\hat{\alpha}_{i-1}$ over $H^m$, which is exactly $\tt(\alpha_{i-1})$ due to our arithmetization.) Since arithmetic operations over $\F$ (including iterated addition, multiplication, and inverse) are in logspace-uniform $\NC^1$~\cite{HesseAB02,HealyV06},\footnote{It is in fact in logtime-uniform $\TC^0$, but here we only need it to be in logspace-uniform $\NC^1$.} it follows that $C_{\arith}$ is of $T_{\sf arith} \eqdef \poly(|\F|^m)$ size and $d_{\sf arith}\eqdef O(m \log |\F|)$ depth. Moreover, $C_{\arith}$ does not depend on $\TM$, and we can compute a Turing machine $\TMarith$ from $\vec{\ell}$ in time $\polylog(T)$ such that
    \[
    \Ckt[T_{\sf arith}, s_{\arith}, h^m +|\langle i, j\rangle| + |\F|^{3m}\log |\F|, |\F|^{3m}\log|\F|](\TMarith) = C_{\arith},
    \]
    where $s_{\arith} = \mu \cdot \beta\kappa\log T$.
	
	Composing $\TurM_{\al}$, $\TurM_{\tthatPhi}$, and $\TurM_{\arith}$ by applying~\autoref{fact:comp-of-circuits} twice gives the desired Turing machine $\TurM_{\sf poly}$ that computes the desired circuit $C_{\sf poly}$.
	
	\paragraph*{Complexity of $C_{\sf poly}$ and $\TurM_{\sf poly}$.} 
	Finally, we verify that $\TurM_{\sf poly}$ and $C_{\sf poly}$ satisfy our requirements. First, from the discussions above, we can bound the size of $C_{\sf poly}$ by $T_{\al} + T_{\tthatPhi} + T_{\arith} \le T_{\sf \poly} = 2^{c \cdot \kappa \log T}$ by picking a sufficiently large $c$. Note that $m \log |\F| = \log (p^m) \le O(\kappa\log T)$. The depth of $C_{\sf poly}$ can be bounded by $d_{\sf poly} = d_{\al} + d_{\tthatPhi} + d_{\arith} \le c \cdot (\kappa^2 \cdot \log^2 T + d \cdot \log T)$ as desired.
	
	From~\autoref{fact:comp-of-circuits}, we have that 
	\[
	|\TurM_{\sf poly}| \le 100 \cdot \left(|\TurM_{\al}| + |\TurM_{\tthatPhi}| + |\TurM_{\arith}| + \log (|\F|^{3m})\right) \le c \cdot \kappa \log T = \log T_{\sf poly}
	\]
	by setting $c$ sufficiently large. The space complexity of $\TurM_{\sf poly}$ can be bounded by 
	\[
	100 \cdot \left(s_{\al} + s_{\tthatPhi} + s_{\arith} \right) \le c \cdot \kappa \log T = \log T_{\sf poly}
	\]
	as well. This completes the proof of \autoref{theo:poly-rep}.
\end{proof}

\subsection{Improved Chen--Tell Generator: Proof of~\autoref{thm:HSG-Chen-Tell-poly}}

\newcommand{\HSGpoly}{\HSG^{\sf layer}}
\newcommand{\Reconpoly}{\Recon^{\sf layer}}

Now we are ready to prove~\autoref{thm:HSG-Chen-Tell-poly} by plugging every polynomial from \autoref{theo:poly-rep} into our modified Shaltiel--Umans generator (\autoref{theo:SU-main}).

\begin{proof}[Proof of~\autoref{thm:HSG-Chen-Tell-poly}]
	We first observe that we can assume $\rho = 1$ without loss of generality. To see how the general case follows from the case that $\rho = 1$, letting $M' = M^\rho$, we can simply define $\HSGCT_{n,T,d,M,\kappa,\rho}$ as the set of strings obtained by truncating every string from $\HSGCT_{n,T,d,M',\kappa,1}$ to their first $M$ bits. The reconstruction algorithm $\ReconCT_{n,T,d,M,\kappa,\rho}$ can then be obtained by slightly modifying $\ReconCT_{n,T,d,M',\kappa,1}$.
	
	\newcommand{\ellct}{\vec{\ell}_{\sf ct}}
	\newcommand{\elllayer}{\vec{\ell}_{\sf layer}}
	\newcommand{\ellpoly}{\vec{\ell}_{\sf poly}}
	
	Let \[
	\ellct = (n,T,d,M,\kappa,1)
	\]
	be the input parameters from the theorem statement and $c$ be a sufficiently large universal constant. From the assumption, we have $n \le T, d \le T$, and $c \cdot \log T \le M \le T^{1/c}$. Let
	\[
	C_{\TurM} = \Ckt[T,\kappa \cdot \log T,n,n](\TurM).
	\]
	
	\paragraph{The layered-polynomial representation.} Let $c_0,c_0',\beta$ be the universal constants from~\autoref{theo:poly-rep}. Let $h$ be the \emph{smallest} nice power of $2$ such that $h \ge M$, $p \eqdef h^{27}$, $m$ be the smallest power of $3$ such that $h^m \ge T^{\beta \kappa}$, and $\F = \F_p$. Note that $p$ is also a nice power of $2$ and $h \le M^3$. %
	
	We will invoke~\autoref{theo:poly-rep} with input parameters \[\ellpoly = (\kappa,T,d,n,h,p).\] Note that from their definitions and our assumption $M \ge c \log T$, we have $\log T\le h < p \le h^{27} \leq M^{81} \leq T$ (assuming $c \geq 81$ is large enough), meaning that the requirements on the input parameters of~\autoref{theo:poly-rep} are satisfied. 
	
	We first apply~\autoref{theo:poly-rep} with input parameters $\ellpoly$ and Turing machine $\TurM$ to obtain $d' = c_0 \kappa \cdot \log^2 T \cdot (d + \kappa^2 \log T)$ polynomials $(P_i)_{i \in [d']} = \mleft(P_i^{\vec{\ell},\TurM}\mright)_{i \in [d']}$. 
	
	\paragraph*{Hitting set $\HSGCT$.} Let $\HSGpoly$ and $\Reconpoly$ denote the $\HSG$ and $\Recon$ algorithms from~\autoref{theo:SU-main}, respectively.\footnote{The superscript $\textsf{layer}$ highlights the fact that they are applied to each layer of the polynomial representation of the circuit.} Let $\Delta = c_0 h \log^3(T)$, 
        \[\elllayer = (p,3m,M,\Delta)\]
    be the input parameters when applying~\autoref{theo:SU-main}. We can verify that $p > \Delta^2 (3m)^7M^9$, \ie, the requirement on the input parameters of~\autoref{theo:SU-main} is satisfied.
	
	We then define $\HSGCT_{\ellct}(\TurM)$ as the union of $\HSGpoly_{\elllayer}(P_i)$ for every $i \in [d']$. Next we analyze the complexity of computing $\HSGCT_{\ellct}(\TurM)$. First, from~\autoref{theo:poly-rep}, letting $ T_{\sf poly} = T^{c_0 \cdot \kappa}$ and $ d_{\sf poly} = c_0 \cdot (d \log T + \kappa^2 \log^2 T)$, there is a polynomial-time algorithm $\algopoly_{\vec{\ell}}$ that takes  $\TurM \in \bs{\kappa \log T}$ as input, and outputs a description of Turing machine $\TurM_{\sf poly} \in \bs{\log T_{\sf poly}}$ such that for 
	\[
	C_{\sf poly} = \Ckt\mleft[T_{\sf poly}, \log T_{\sf poly},\log d' , |\F|^{3m} \cdot \log |\F| \mright](\TurM_{\sf poly})
	\]
	it holds that (1) for every $i \in [d']$ $C_{\sf poly}(i) = \tt(P_i)$ and (2) $C_{\sf poly}$ has depth $d_{\sf poly}$.
	
	Second, from~\autoref{theo:SU-main}, there is a logspace-uniform circuit family with input parameters $\elllayer$, size $\poly(p^m)$, and depth $\poly(\log p,m,M)$ such that for every $i \in [d']$, it outputs $\HSGpoly_{\elllayer}(P_i)$ when taking $\tt(P_i)$ as input. Note that $\poly(p^m) \le T^{O(\beta\kappa)}$ and $\poly(\log p,m,M) \le \poly(M)$. Applying~\autoref{fact:comp-of-circuits} to compose the machines above and enumerating over all $i \in [d']$,\footnote{Enumerating all $i \in [d']$ only adds a $O(\log d')$ additive overhead in depth and a  $O(d')$ multiplicative blowup in size, which are negligible.} we obtain the desired circuit $C_{\sf H}$ (note that $c$ is sufficiently large).

	\paragraph*{Reconstruction $\ReconCT$.} 
	
	For every $i \in \{2, \dotsc, d'\}$, the reconstruction algorithm $\ReconCT$ attempts to construct a $\poly(p,m,\log (Md'))$-size $D$-oracle circuit $E_i$ that computes $P_i$. A formal description of $\ReconCT$ is as follows:
	
	\begin{itemize}
		\item We start with the circuit $E_1(\vec{x}) = \Base(\vec{\ell}, \TM, \vec{x})$ that computes the polynomial $P_1$.
		\item For every $i \in \{2,\dotsc,d'\}$:
		\begin{enumerate}
			\item We first construct a procedure $\WT{P}_i$ computing $P_i$ using the $D$-oracle circuit $E_{i-1}^D$ for $P_{i-1}$ and the downward self-reducibility for $P_i$. In particular, on input $\vec{x} \in \F^{3m}$, let 
			\[\widetilde{P}_i(\vec{x}) \eqdef \DSR^{E_{i-1}^D}(\vec{\ell}, \TurM, i, \vec{x}).\]
			\item Run $\left(\Reconpoly\right)^{D,\WT{P}_i}_{\elllayer}$ which outputs a $D$-oracle circuit $\WT{E}_i^D$ in $\poly(p,m,M)$ time.

			\item Let $t \eqdef c_1 \cdot m \cdot \log p$ for a sufficiently large constant $c_1 > 1$. Take $t$ i.i.d.~samples $\vec{x}_1,\dotsc,\vec{x}_t$ from $\F^{3m}$. Check that for every $j \in [t]$, $\WT{E}^{D}_i(\vec{x}_j) = \WT{P}_i(\vec{x}_j)$. If any condition does not hold, the algorithm outputs $\bot$ and aborts immediately.
			
			\item Let $E_i$ be a $D$-oracle circuit constructed as follows:
			\begin{enumerate}
				\item Draw $t = \Theta(m \log p)$ i.i.d.~samples of random strings $r_1,\dotsc,r_t$ used by $\mathsf{PCorr}$. (Recall that $\mathsf{PCorr}$ is the self-corrector for low-degree polynomials in \autoref{t:SC}.)
				\item Set $E_i(\vec{x}) = \mathsf{MAJ}_{k \in [t]}~\mathsf{PCorr}^{\WT{E}_i}(p,3m,\Delta,\vec{x}; r_k)$ for all $\vec{x} \in \F^{3m}$.
			\end{enumerate}
		\end{enumerate}
		
		\item For every $j \in [n]$, output $E_{d'}^D(\id(j),0^{2m})$.
	\end{itemize}
	
	For ease of notation, for every $i' \in \{2,\dotsc,d'\}$, we use $\tau_{i'}$ to denote the randomness used when running the algorithm above with $i = i'$, and we use $\tau_{\le i}$ to denote $\tau_{1},\dotsc,\tau_{i}$. Also, if $E_i$ is not constructed by the algorithm (meaning that the algorithm aborts before constructing $E_i$), we set $E_i = \bot$. 
	
	From~\autoref{t:SC},~\autoref{theo:SU-main}, and~\autoref{theo:poly-rep}, the running time of the algorithm above can be bounded by 
	\[
	\poly(p,m,h,\log (Md')) \cdot (d' + n) \le \poly(M) \cdot (d' + n) \le \poly(M) \cdot (d + n).
	\]
	The last inequality follows from the fact that $M \ge \log T$ and hence $d' =c_0 \kappa \cdot \log^2 T \cdot (d + \kappa^2 \log T) \le \poly(M) \cdot d$.
	Now we establish the soundness and completeness of the reconstruction. We show the following claim.
	\begin{claim}\label{claim:one-step}
		Fix $\bfunc{D}{M}{}$. For every $i \in \{2,\dotsc,d'\}$, for every fixed $\tau_{\le i-1}$, if $E_{i-1}^{D}$ computes $P_{i-1}$ or $i = 2$,\footnote{Note that $\tau_{\le i-1}$ determines $E_{i-1}$.} then with probability at least $1 - 1/p^m$ over $\tau_{i}$ the following holds:
		\begin{itemize}
			\item {\bf (Soundness.)} If $E_i \ne \bot$, then $E_i^{D}$ computes $P_i$.
			\item {\bf (Completeness.)} If $D$ $(1/M)$-avoids $\HSGpoly_{\elllayer}(P_i)$, then $E_i^{D}$ computes $P_i$.
		\end{itemize}
	\end{claim}
	
	Before establishing the claim, we show it implies the completeness and soundness of the reconstruction. To see the soundness, note that by induction over all $i \in \{2,\dotsc,d'\}$, with probability at least $1 - d'/p^m > 9/10$, it holds that if $E_{d'} \ne \bot$, then $E_{d'}$ computes $P_{d'}$, meaning the reconstruction outputs the correct output $C_{\TurM}(1^n)$. To see the completeness, note that an oracle $\bfunc{D}{M}{}$ that $(1/M)$-avoids $\HSGCT_{\vec{\ell}}(\TurM)$ also $(1/M)$-avoids $\HSGpoly_{\elllayer}(P_i)$ for every $i \in [d']$. Hence, by induction over $i \in \{2,\dotsc,d'\}$, with probability at least $1 - d'/p^m > 9/10$, it holds that $E_i$ computes $P_i$ for every $i \in \{2,\dotsc,d'\}$. Thus the reconstruction will output $C_{\TurM}(1^n)$. The success probability $9/10$ can be amplified to $1-2^{-M}$ by running the reconstruction algorithm $O(M)$ times independently and outputting the answer that occurs most frequently.
	
	Finally, we prove the claim.
	\begin{proof}[Proof of~\autoref{claim:one-step}]\renewcommand{\qedsymbol}{$\diamond$}
		We first establish the soundness. From the assumption that $E_{i-1}^{D}$ computes $P_{i-1}$ or $i = 2$ and the downward self-reducibility property of~\autoref{theo:poly-rep}, it follows that $\WT{P}_i$ computes $P_i$. Therefore, $E_i \ne \bot$ means that $\WT{E}_i$ has passed the test in Step 3, meaning that with probability at least $1 - p^{-4m}$ over the randomness in Step 3, it holds that $\WT{E}_i$ agrees $P_i$ on at least $3/4$ fraction of inputs from $\F^{3m}$. This then means that with probability at least $1 - p^{-3m}$ over the randomness in Step 4(a), we have $E_i^D$ computes $P_i$.
		
		The completeness follows immediately from~\autoref{theo:SU-main}. (Here $\WT{E}_i^D$ already computes $P_i$ with probability at least $1 - 1/p^m$.)
	\end{proof}
	This completes the proof of~\autoref{thm:HSG-Chen-Tell-poly}.
\end{proof}

\ifnum\Anonymity=0
\section*{Acknowledgments}
Lijie Chen is supported by a Miller Research Fellowship. Zhenjian Lu is supported by an NSERC Postdoctoral Fellowship. Igor C.~Oliveira received support from the EPSRC New Horizons Grant EP/V048201/1, the Royal Society University Research Fellowship URF$\setminus$R1$\setminus$191059, and the Centre for Discrete Mathematics and its Applications (DIMAP) at the University of Warwick. Hanlin Ren received support from DIMACS through grant number CCF-1836666 from the National Science Foundation. Rahul Santhanam received support from the EPSRC New Horizons Grant EP/V048201/1. This work was done in part while the authors were visiting the Simons Institute for the Theory of Computing. 
\fi

{\small \bibliography{main}}

\newpage
\appendix

\newpage
\listoffixmes

\end{document}